\documentclass[10pt,a4paper]{article}
\usepackage{fullpage,graphicx,amssymb,amsmath,xspace,amsfonts,tabularx,rotating,url,float,color}

\newcommand{\gate}{g_}
\newcommand{\Reals}{\mathbb{R}}
\def\bbox{\ensuremath{\mathrm{bbox}}\xspace}
\newcommand{\WLMan}{\ensuremath{\mathrm{WD}_1}}
\newcommand{\WLEuc}{\ensuremath{\mathrm{WD}_2}}
\newcommand{\WLMax}{\ensuremath{\mathrm{WD}_{\infty}}}
\newcommand{\WBV}{\ensuremath{\mathrm{WBV}}}
\newcommand{\WBS}{\ensuremath{\mathrm{WBS}}}
\newcommand{\WS}{\ensuremath{\mathrm{WS}}}

\def\smallminus{\hbox{\footnotesize-}}
\newcommand\descr[1]{{\ensuremath{\thinmuskip=0mu\let\m\smallminus#1}}}
\def\fwd[#1]{\ensuremath{\left[\begin{smallmatrix}#1\end{smallmatrix}\right\}}}
\def\rev[#1]{\ensuremath{\left\{\begin{smallmatrix}#1\end{smallmatrix}\right]}}
\def\edge[#1]{\ensuremath{\begin{smallmatrix}#1\end{smallmatrix}}}
\def\vtx#1{\ensuremath{\left(\begin{smallmatrix}#1\end{smallmatrix}\right)}}
\def\xor<#1,#2>{\ensuremath{\frac{#1}{#2}}}
\newcommand\curvename[1]{\leavevmode\hbox{\normalfont\ttfamily #1}}

\newtheorem{lemma}{Lemma}
\newtheorem{corollary}{Corollary}
\newtheorem{theorem}{Theorem}
\newtheorem{finding}{Finding}
\newenvironment{proof}{Proof:}{\qed}
\def\squareforqed{\hbox{\rlap{$\sqcap$}$\sqcup$}}
\def\qed{\ifmmode\squareforqed\else{\unskip\nobreak\hfil
\penalty50\hskip1em\null\nobreak\hfil\squareforqed
\parfillskip=0pt\finalhyphendemerits=0\endgraf}\fi}
\newenvironment{howfound}{How found:}{\qed}

\title{How many three-dimensional Hilbert curves are there?}
\author{Herman Haverkort\\Eindhoven University of Technology}
\date{1 October 2016}

\begin{document}
\maketitle

\begin{abstract}
Hilbert's two-dimensional space-filling curve is appreciated for its good locality-preserving properties and easy implementation for many applications. However, Hilbert did not describe how to generalize his construction to higher dimensions. In fact, the number of ways in which this may be done ranges from zero to infinite, depending on what properties of the Hilbert curve one considers to be essential.

In this work we take the point of view that a Hilbert curve should at least be self-similar and traverse cubes octant by octant. We organize and explore the space of possible three-dimensional Hilbert curves and the potentially useful properties which they may have. We discuss a notation system that allows us to distinguish the curves from one another and enumerate them. This system has been implemented in a software prototype, available from the author's website.
Several examples of possible three-dimensional Hilbert curves are presented, including a curve that visits the points on most sides of the unit cube in the order of the two-dimensional Hilbert curve; curves of which not only the eight octants are similar to each other, but also the four quarters; a curve with excellent locality-preserving properties and endpoints that are not vertices of the cube; a curve in which all but two octants are each other's images with respect to reflections in axis-parallel planes; and curves that can be sketched on a grid without using vertical line segments. In addition, we discuss several four-dimensional Hilbert curves.
\end{abstract}

\section{Introduction}\label{sec:introduction}
\begin{figure}[b]
\centering
\includegraphics[width=\hsize]{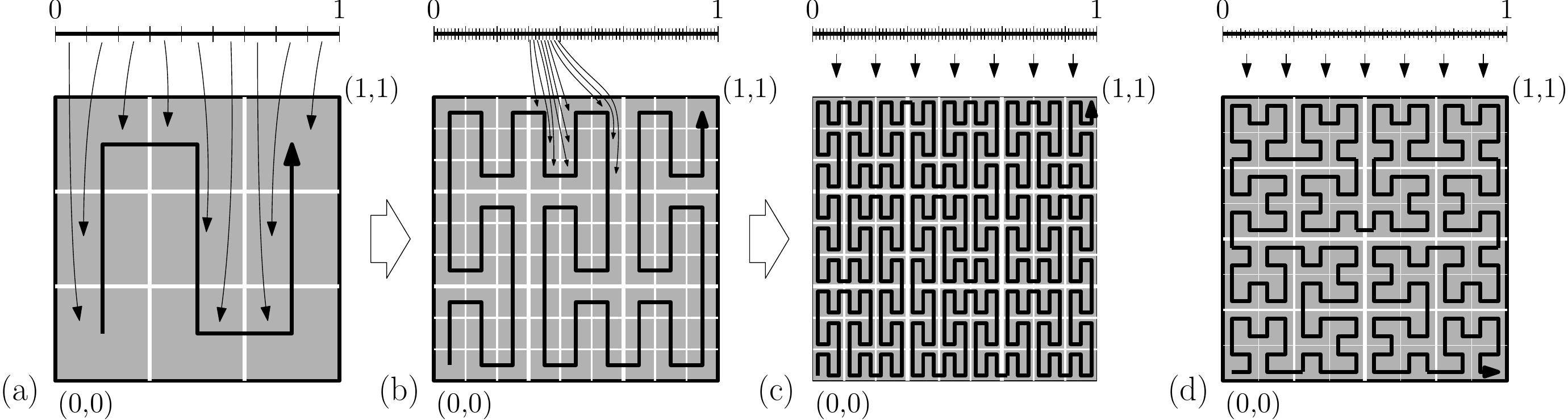}
\caption{(a--c) A sketch of Peano's space-filling curve. (d) A sketch of Hilbert's space-filling curve.}
\label{fig:peano2d}\label{fig:hilbert2d}
\end{figure}

A space-filling curve in $d$ dimensions is a continuous, surjective mapping from $\Reals$ to $\Reals^d$. In the late 19th century Peano~\cite{Peano} was the first to present such a mapping. It can be described as a recursive construction that maps the unit interval $[0,1]$ to the unit square $[0,1]^2$. The unit square is divided into a grid of $3 \times 3$ square cells, while the unit interval is subdivided into nine subintervals. Each subinterval is then matched to a cell; thus Peano's curve traverses the cells one by one in a particular order. The procedure is applied recursively to each subinterval-cell pair, so that within each cell, the curve makes a similar traversal (see Figure~\ref{fig:peano2d}a--c). By carefully reflecting and/or rotating the traversals within the cells, one can ensure that each cell's first subcell touches the previous cell's last subcell. The result is a fully-specified, continuous, surjective mapping from the unit interval to the unit square. This mapping can be extended to a mapping from $\Reals$ to $\Reals^2$ by inverting the recursion, recursively considering the unit interval and the unit square as a subinterval and a cell of a larger interval and a larger square.

In response to Peano's publication, Hilbert~\cite{Hilbert} sketched a space-filling curve based on subdividing a square into only four squares (Figure~\ref{fig:hilbert2d}d). Since then, quite a number of space-filling curves have appeared in the literature~\cite{boxquality,Sagan}, and space-filling curves have been applied in diverse areas such as indexing of multidimensional points~\cite{asano,Kamel,LawderKing,Liao}, load balancing in parallel computing~\cite{Bungartz,Harlacher}, improving cache utilization in computations on large matrices~\cite{matrixcomputations} or in image rendering~\cite{voorhies}, finite element methods~\cite{Bader}, image compression~\cite{imagecompression}, and combinatorial optimization~\cite{tsp}---to give only a few examples of applications and references. The function of the space-filling curve typically lies in providing a way to traverse points or cells of a square or a higher-dimensional space in such a way that consecutive elements in the traversal tend to lie very close to each other, and elements that lie very close to each other tend to be close to each other in the traversal order. In other words, the space-filling curve \emph{preserves locality}: this effect is captured by various metrics which we will discuss in Section~\ref{sec:localityproperties}.

For many applications, Hilbert's curve, rather than Peano's, appears to be the curve of choice, sometimes for its better locality-preserving properties (points close to each other along the curve tend to be close to each other in the plane)~\cite{voorhies}, but more commonly for the fact that Hilbert's curve is based on subdividing squares into only four subsquares. The latter property does not only make the Hilbert curve well suitable for the traversal of quadtrees and of grids whose width is a power of two, but it also matches very well with binary representations of coordinates of points. In particular, the cell in which a given point $p$ lies can be determined by inspecting the binary representations of the coordinates of $p$ bit by bit, and as a consequence, the order in which two points $p$ and $q$ appear along the curve can be determined without relatively time-consuming arithmetic such as divisions.

Peano's two-dimensional curve, based on a subdivision of a square into nine squares, generalizes in a natural way to a three-dimensional curve, based on a subdivision of a cube into 27 cubes (Peano also described this), or even a $d$-dimensional curve, based on a subdivision of a hypercube into $3^d$ hypercubes. However, generalization of the Hilbert curve to higher dimensions is not as straightforward---Hilbert's publication does not discuss it. Naturally, a generalization to three dimensions would be based on subdividing cubes into eight cells. The tricky part is how to choose the traversals within the cells, so that each cell's first subcell touches the previous cell's last subcell and continuity of the mapping is ensured.

Butz's solution to this problem~\cite{Butz} is fairly well-known, but many other solutions are possible. Documentation of existing applications or implementations of three-dimensional Hilbert curves is not always explicit about the fact that a particular, possibly arbitrary, curve was chosen out of many possible three-dimensional Hilbert curves. However, different curves have different properties: which three-dimensional Hilbert curve would constitute the best choice depends on what properties of a Hilbert curve are deemed essential and what qualities of the space-filling curve one would like to optimize for a given application. This gave rise to efforts to set up frameworks to describe such curves~\cite{Alber} and to analyse differences in their properties so that one can identify optimal curves for different applications~\cite{chochia,gotsman,niedermeier}, including recent efforts by my co-authors and myself~\cite{hyperorthogonal,jea}. However, the scope of these studies has been fairly limited, each of them considering only a subset of possible Hilbert curves and focusing on one particular quality to optimize.

\paragraph{Contents of this article}
In this work we dive into the question what defines a Hilbert curve. Different answers may unlock different worlds of three-dimensional space-filling curves. Is each of them as good as any other? What do the different curves have in common and what are their differences? Can we enumerate them within reasonable time to analyse their properties? Can we also answer these questions for higher-dimensional Hilbert curves?

\begin{figure}
\centering
\includegraphics[width=\hsize]{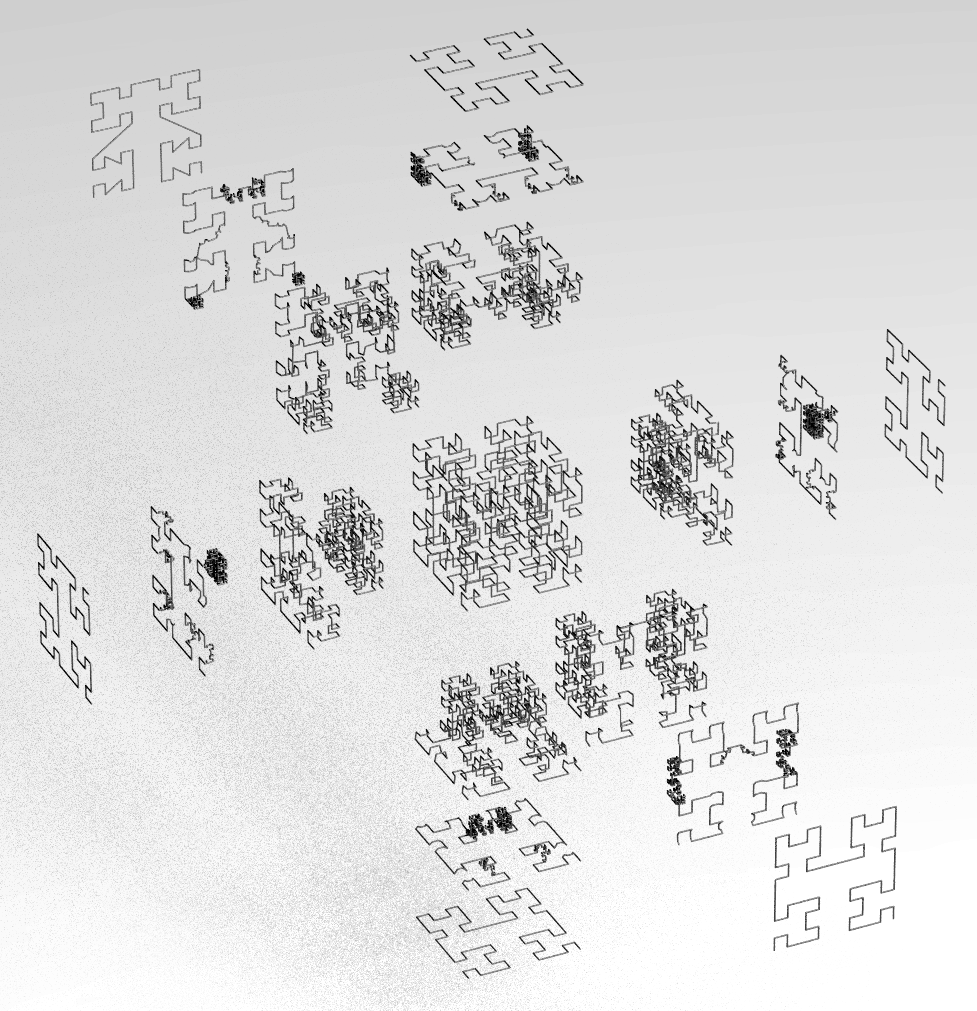}
\caption{The three-dimensional \emph{harmonious Hilbert curve}. In the centre of the figure, there is a (rather impenetrable) sketch of the order in which the curve traverses points in a $8 \times 8 \times 8$ grid. To the left we see what happens if one contracts the curve onto the left facet of the $8 \times 8 \times 8$ cube, maintaining only the points on that facet and skipping all other points of the cube. The result corresponds to a two-dimensional Hilbert curve. Similarly, we find two-dimensional Hilbert curves on the right, top, bottom and front facets---only on the back facet we find another pattern. The harmonious Hilbert curve is the only three-dimensional Hilbert curve that has two-dimensional Hilbert curves on five facets. Most three-dimensional Hilbert curves, including all those depicted in Figure~\ref{fig:fivecurves}, do not match a two-dimensional Hilbert curve on any facet.}
\label{fig:allharmonious}
\end{figure}

The goal of this work is to explore and organize the space of possible three-dimensional Hilbert curves and the properties which they may have, to find interesting three-dimensional space-filling curves, and to generate ideas for further generalization to four or more dimensions. Among the newly discovered curves in the present article are:\begin{itemize}
\item the three-dimensional \emph{harmonious Hilbert} curve (sketched in Figure~\ref{fig:allharmonious} and Figure~\ref{fig:edgecrossingcurves}a), which has unique the property that the points on five of the six two-dimensional facets of the unit cube are visited in the order of the two-dimensional Hilbert curve
    (in four dimensions we found such properties to be relevant to R-tree construction~\cite{jea});
\item a curve (sketched in Figure~\ref{fig:edgecrossingcurves}d) of which not only the eight octants are similar to each other, but also the four quarters and the two halves, and which minimizes the worst-case relative size of the boundary of any curve section (a quality measure relevant to load-balancing applications~\cite{hungershoefer});
\item a curve (sketched in Figure~\ref{fig:fivecurves}(centre) and Figure~\ref{fig:vertexedgegatedcurves}a) which, similar to the two-dimensional Hilbert curve, is only rotated in the first and the last octant, whereas the curve within each of the remaining octants is obtained from the complete curve by a combination of only scaling, translation, reversal and/or reflection in axis-parallel planes;
\item curves along which consecutive subcubes are never directly on top of each other (Figures \ref{fig:building}, \ref{fig:vertexedgegatedcurves}c and \ref{fig:vertexedgegatedcurves}d):
    if one sketches the curve by connecting the centre points of the cells in a regular grid in the order in which they are traversed by the curve, then there are no vertical edges.
\end{itemize}
Some more examples are shown in Figures \ref{fig:fivecurves}, \ref{fig:edgecrossingcurves}, \ref{fig:facetcrossingcurves}, \ref{fig:vertexedgegatedcurves}, \ref{fig:boringcurves}, and~\ref{fig:facetgatedcurve}.

\begin{figure}
\centering
\includegraphics[width=\hsize]{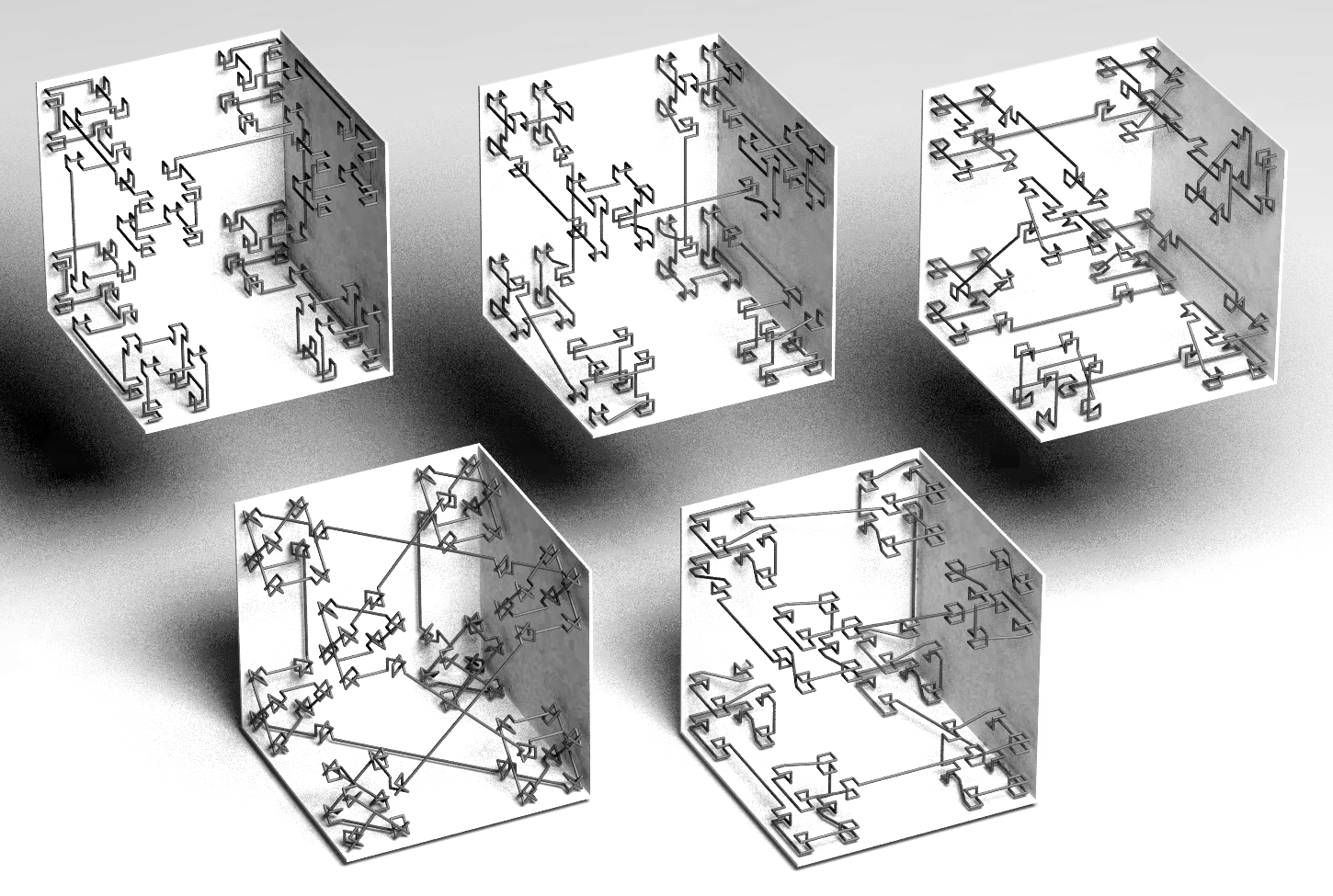}
\caption{Sketches of the order in which five different three-dimensional Hilbert curves traverse the points in a $8 \times 8 \times 8$ grid. Compared to Figure~\ref{fig:allharmonious}, the sketches have been made more legible by introducing extra spacing between the eight octants and between the eight suboctants within each octant. In clockwise order, starting at the top left: Butz's curve; the base-camp curve; a curve with many sections that fill four cubes in a row; a curve with a helix-shaped base pattern; and a curve with many self-intersections.}
\label{fig:fivecurves}
\end{figure}

Furthermore, this article sets up a notation and naming system that is compact, yet sufficiently powerful to distinguish between 10\,694\,807 different three-dimensional Hilbert curves (modulo rotation, reflection, translation, scaling and reversal), assigning a unique name to each such curve. The system comes with a prototype of a software tool that can enumerate the curves, or determine the name of a curve from the order in which it traverses the cubes in a grid. This may facilitate the automatic identification, verification and comparison of curves implemented in existing code, whose documentation does not always explicitly specify exactly what three-dimensional Hilbert curve is used, out of the many possible curves.

This article is structured as follows.

In Section~\ref{sec:notation}, I describe a notation system that allows us to describe Hilbert curves and discuss their properties. We discuss the characteristic properties of the two-dimensional Hilbert curve and their possible generalizations to higher dimensions in Section~\ref{sec:properties}. From the (generalized) properties of the two-dimensional Hilbert curve, we select some as defining properties for Hilbert curves in arbitrary dimensions. In support of this selection and as a warming-up for what follows, we prove that in two dimensions, the known Hilbert curve is the unique curve that has all of the defining properties (Section~\ref{sec:twodimensional}). At the heart of our proof is a case distinction by different possible locations of the end points of the curve. We find that in two dimensions, the only combination of end points that can be realized by a curve that has all of the selected properties consists of two vertices on the same edge of the square.

We then turn to exploring the space of three-dimensional Hilbert curves. A straightforward encoding of Hilbert curve descriptions in the notation presented in Section~\ref{sec:notation} does not allow us to enumerate such curves efficiently. To overcome this problem, we set up a framework for a more compact naming scheme for three-dimensional curves in Section~\ref{sec:scheme}, which will also make symmetries in curves easier to recognize. In Section~\ref{sec:inventory} we fill in the details, again making a case distinction by different possible locations of the end points of the curve. We prove that only a limited number of end points are possible, explain how to enumerate the names of the possible curves for each possible combination of end points, and show examples. Next we see how we can establish or verify the presence or absence of combinations of certain properties in curves in Section~\ref{sec:observations}, and I report on the locality-preserving properties of the curves. Section~\ref{sec:software} briefly describes a prototype of a software tool to enumerate, identify, analyse and sketch the curves.

Having established a way to explore and structure the space of three-dimensional Hilbert curves, we can now try to answer the title question of this article in Section~\ref{sec:howmanyin3D}: how many three-dimensional Hilbert curves are there? We discuss four-dimensional curves in Section~\ref{sec:howmanyin4D}, and conclude with a discussion of the implications of our findings and questions raised by them in Section~\ref{sec:evaluation}.

Illustrated examples of curves appear throughout this article. Appendix~\ref{apx:examples} gives the definitions and lists properties of all of these curves.

This article extends, improves and replaces most of my brief preliminary manuscript ``An inventory of three-dimensional Hilbert space-filling curves''~\cite{inventory}\footnote{However, the present article does not cover the previous manuscript entirely. The previous manuscript~\cite{inventory} focuses more on certain metrics of locality-preserving properties and includes some results on non-self-similar, ``poly-Hilbert'' curves that are not covered here.}.

\section{Defining self-similar traversals}\label{sec:notation}

\subsection{Defining self-similar traversals by figure}\label{sec:definitionbyfigure}
We can define a \emph{self-similar traversal} of points in a $d$-dimensional cube as follows. We consider the unit cube $C$ to be subdivided into $2^d$ subcubes of equal size. We specify a \emph{base pattern}: an order in which the traversal visits these subcubes. Let $C_1,...,C_{2^d}$ be the subcubes indexed by the order in which they are visited. Moreover, we specify, for each subcube $C_i$, a transformation $\sigma_i$ that maps the traversal of the cube as a whole to the traversal of $C_i$. More precisely, each $\sigma_i$ can be thought of as a triple $(\gamma_i, \rho_i, \chi_i)$, where $\gamma_i: C \rightarrow C$ is one of the $2^d d!$ symmetries of the unit cube, $\rho_i: C \rightarrow C_i$ translates the unit cube and scales it down to map it to $C_i$, and $\chi_i: [0,1] \rightarrow [0,1]$ is a function that specifies whether or not to reverse the direction of the traversal: it is defined by $\chi_i(t) = t$ for a forward traversal, and by $\chi_i(t) = 1-t$ for a reversed traversal.

When $d=2$ or $d=3$, it is feasible to give such a specification in a graphical form, as follows. We draw a cube, and indicate, by a thick arrow along the vertices of the cube, the order in which its vertices, and hence its $2^d$ \emph{first-level} subcubes $C_1,...,C_{2^d}$, are visited by the traversal. This is the \emph{first-order approximating curve} (see Figure~\ref{fig:graphicalnotation}a). In fact, we can omit the unit cube from the drawing, as it is implied by the arrow. Inside the cube, we draw the \emph{second-order approximating curve}: a polygonal curve that connects the centres of the $4^d$ \emph{second-level} subcubes of the unit cube in the order in which they are visited by the traversal (Figure~\ref{fig:graphicalnotation}b). Finally, we mark, with an open dot, the vertex that represents $C_1$, and the vertices that represent the corresponding second-level subcubes within their respective first-level subcubes. The arrow head on the first-order approximating curve is now redundant and can be removed (Figure~\ref{fig:graphicalnotation}c).

\begin{figure}
\centering
\includegraphics[width=\hsize]{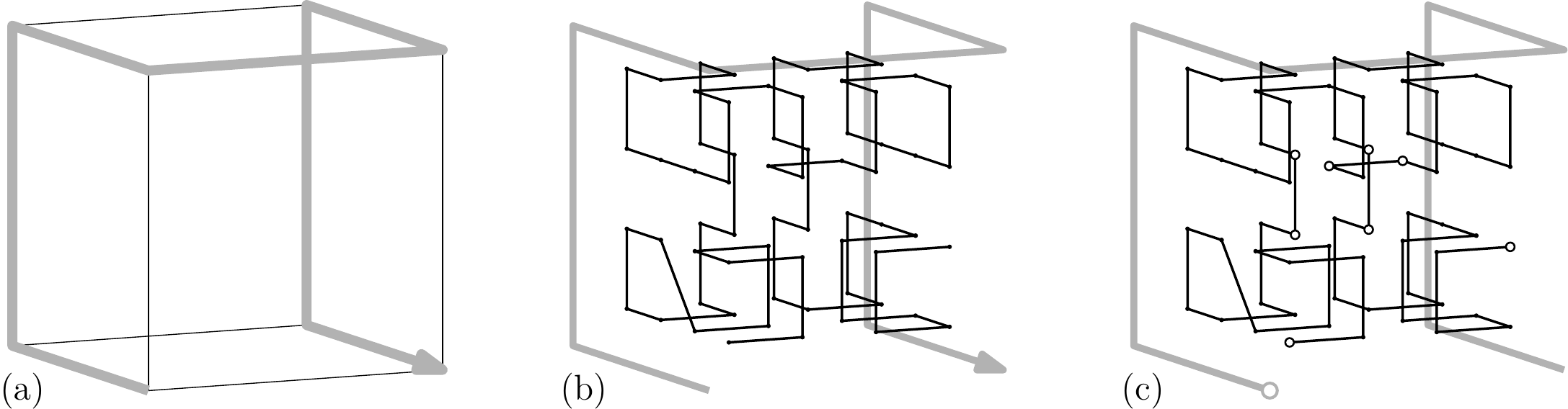}
\caption{Example of a graphical definition of a three-dimensional Hilbert curve. (a) First stage: the first-order approximating curve. (b) Second stage: the second-order approximating curve. (c) Third stage: marking the vertex representing the first first-level subcube and the vertices representing the corresponding second-level subcubes.}
\label{fig:graphicalnotation}
\end{figure}

Note how the open dots specify the direction functions $\chi_i$: if, within a given subcube $C_i$, the marked vertex is the first one visited by the second-order approximating curve, then $\chi_i(t) = t$; if the marked vertex is the last one visited by the second-order approximating curve, then $\chi_i(t) = 1 - t$. Given $\chi_i$, the transformations $\gamma_i$ and $\rho_i$ are implied by the shapes of the first- and second-order approximating curves: these curves show how the base pattern (and hence, the whole traversal) is rotated and/or reflected in each octant. If the first-order approximating curve is asymmetric (as in Figures \ref{fig:facetcrossingcurves}efh and~\ref{fig:vertexedgegatedcurves}e), the functions $\chi_i$ are implied by the drawing of the second-order curve even without the dots, but we draw the dots nevertheless for clarity. If the second-order approximating curve is symmetric (as in Figures \ref{fig:edgecrossingcurves}abdefh and \ref{fig:facetcrossingcurves}abcdg), the whole traversal is symmetric, and the dots are without effect---in this case we omit the dots to emphasize the symmetry. If the first-order approximating curve is symmetric but the second-order approximating curve is not (as in Figures \ref{fig:edgecrossingcurves}cg, \ref{fig:vertexedgegatedcurves}abcd, \ref{fig:boringcurves} and~\ref{fig:facetgatedcurve}) the dots are necessary for the unambiguous definition of a self-similar traversal: Figure~\ref{fig:needthedots} illustrates how moving a dot on the second-order approximating curve leads to differences in the third-order approximating curve.

\begin{figure}
\hbox to\hsize{%
\vbox{\hsize=0.2\hsize
\hbox{\rlap{(a)}\includegraphics[width=\hsize,page=1]{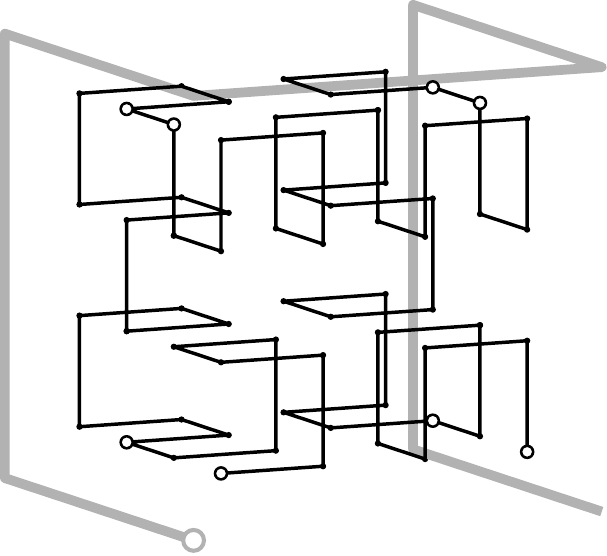}}
\hbox{\rlap{(b)}\includegraphics[width=\hsize,page=2]{alpha-and-wrong.pdf}}
}\hfill
\hbox to .35\hsize{\rlap{\includegraphics[width=0.35\hsize]{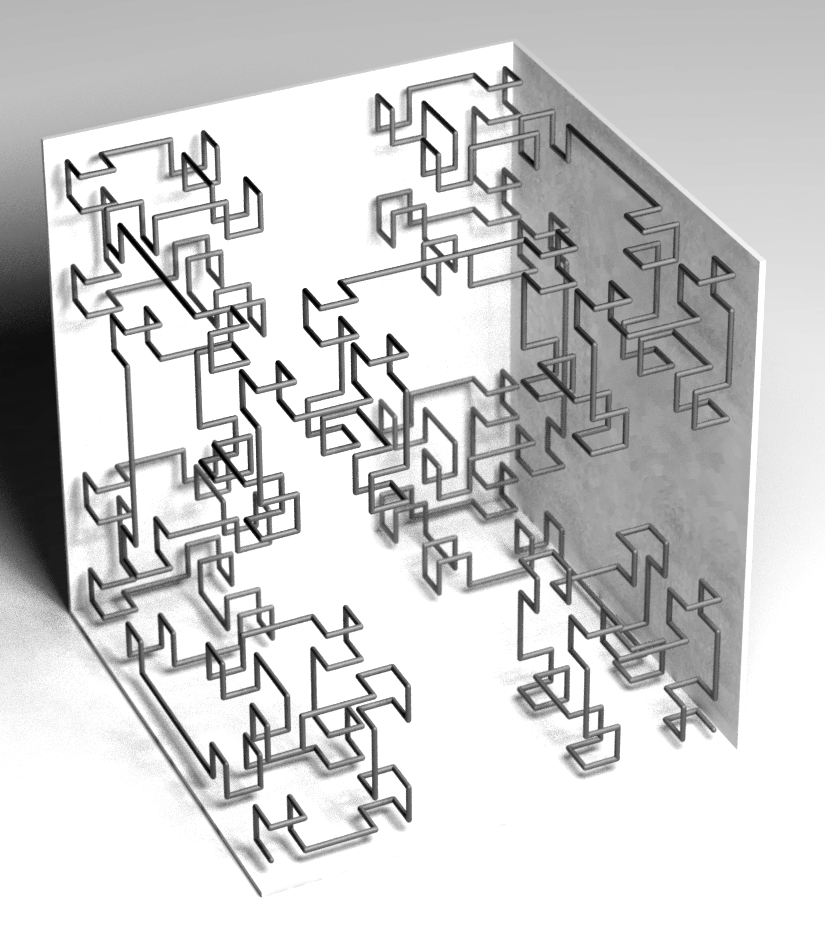}}(c)\hss}\hfill
\hbox to .35\hsize{\rlap{\includegraphics[width=0.35\hsize]{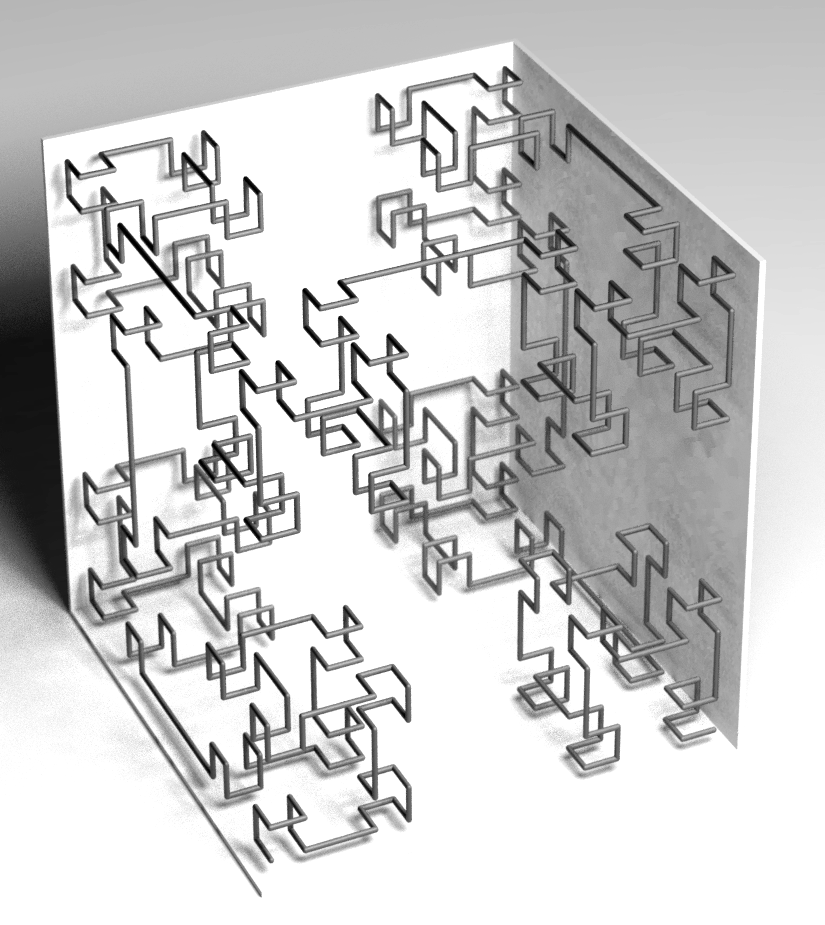}}(d)\hss}%
}
\caption{(a,b) Two subtly different definitions of three-dimensional Hilbert curves: the only difference is the location of the open dot in the last octant. (c,d) The corresponding third-order approximating curves, which differ in the last octant.}
\label{fig:needthedots}
\end{figure}

\subsection{Mapping the unit interval to the unit cube}\label{sec:mapping}
As illustrated in Figure~\ref{fig:hilbert2d}, we can think of a traversal as mapping segments of the unit interval to subcubes of the unit cube~$C$. For a given level of refinement $k$, consider the unit interval subdivided into $2^{kd}$ segments of equal length, and the unit cube subdivided into $2^{kd}$ subcubes of equal size. Let $s(i,k)$ be the $i$-th segment of the unit interval, that is, the interval $[(i-1)\cdot 2^{-kd}, i\cdot 2^{-kd}]$. Let $c(i,k)$ be the $i$-th subcube in the traversal. We can determine $c(i,k)$ from the transformations $\gamma, \rho$ and $\chi$ as follows. If $k = 0$, then $i$ must be 1 and $c(i,k) = C$. Otherwise, let $z = 2^{d(k-1)}$ be the number of subcubes within a first-level subcube, let $b = \lceil i/z\rceil$ be the index of the first-level subcube that contains $c(i,k)$, and let $j$ be the index of $c(i,k)$ within $C_b$. More precisely, if $\chi_b$ indicates a forward traversal of $C_b$, then $j = i - (b-1)z$, and if $\chi_b$ indicates a reverse traversal of $C_b$ then $j = bz - i + 1$. Then we have $c(i,k) = \rho_b(\gamma_b(c(j,k-1)))$, and the traversal maps the segment $s(i,k)$ to the cube $c(i,k)$.

As $k$ goes to infinity, the segments $s(i,k)$ and the cubes $c(i,k)$ shrink to points, and the traversal defines a mapping from points on the unit interval to points in the unit cube. By construction, the mapping is surjective. However, it may be ambiguous, as some points in the unit interval lie on the boundary between segments for any large enough $k$. We may break the ambiguity towards the left or towards the right, by considering segments to be relatively open on the left or on the right side, respectively. In the first case, for a given $k$, we consider a point $t$ on the unit interval to be part of the $i$-th interval with $i = \lceil 2^{kd}t\rceil$, and we define a mapping $\tau^{-}: (0,1] \rightarrow C$ to points in the unit cube by $\tau^{-}(t) = \lim_{k\rightarrow\infty} c(\lceil 2^{kd}t\rceil, k)$. In the second case, we consider $t$ to be part of the $i$-th interval with $i = \lfloor 2^{kd}t\rfloor + 1$, and we define a mapping $\tau^{+}: [0,1) \rightarrow C$ by $\tau^{+}(t) = \lim_{k\rightarrow\infty} c(\lfloor 2^{kd}t\rfloor + 1, k)$.

\subsection{Defining self-similar traversals by signed permutations}\label{sec:definitionbypermutations}
To define the mappings $\tau^{-}$ and $\tau^{+}$, all we need to do is to specify, for each $i \in \{1,...,2^d\}$, the transformation $\gamma_i$, the location of $C_i$ (or, to the same effect, $\rho_i$), and the orientation function $\chi_i$. This can be done in a graphical way, as explained above, but this approach is not suitable for automatic processing of traversals in software (or for four- and higher-dimensional traversals, for that matter). For that purpose, we need a numeric notation system. The numeric systems used in this article is based on ideas from Bos as incorporated in our work on hyperorthogonal well-folded Hilbert curves~\cite{hyperorthogonal}, adapted to suit the broader class of curves discussed in the present article. I will now explain this notation system.

We specify the base pattern by indicating, for each of the subcubes $C_i$ with $1 < i \leq 2^d$, where it lies relative to the previous subcube $C_{i-1}$. Let $c_i$ be the centre point of $C_i$; the position of $C_i$ relative to $C_{i-1}$ can then be expressed by the vector $v_i = c_i - c_{i-1}$. We use square brackets to index the elements of a vector, so $v_i$ is a column vector with elements $v_i[1], v_i[2], ... v_i[d]$. However, in our notation system, we specify $v_i$ in a more compact way, namely by a set of numbers $V_i \subset \{-1,...,-d\} \cup \{1,...,d\}$ such that $v_i[j] < 0$ if and only if $-j \in V_i$; $v_i[j] > 0$ if and only if $j \in V_i$; and $v_i[j] = 0$ if and only if $j, -j \notin V_i$. For an example, see Figure~\ref{fig:basepatternnotation}. Note how $V_i = \{ j \}$ can be interpreted as: move forward along the $j$-th coordinate axis to get from $C_{i-1}$ to $C_i$, while $V_i = \{ -j \}$ means: move back along the $j$-th coordinate axis, and $V_i = \{ j_1, j_2 \}$ indicates a diagonal move, simultaneously moving forward in coordinates $j_1$ and $j_2$.

\begin{figure}
\centering
\includegraphics{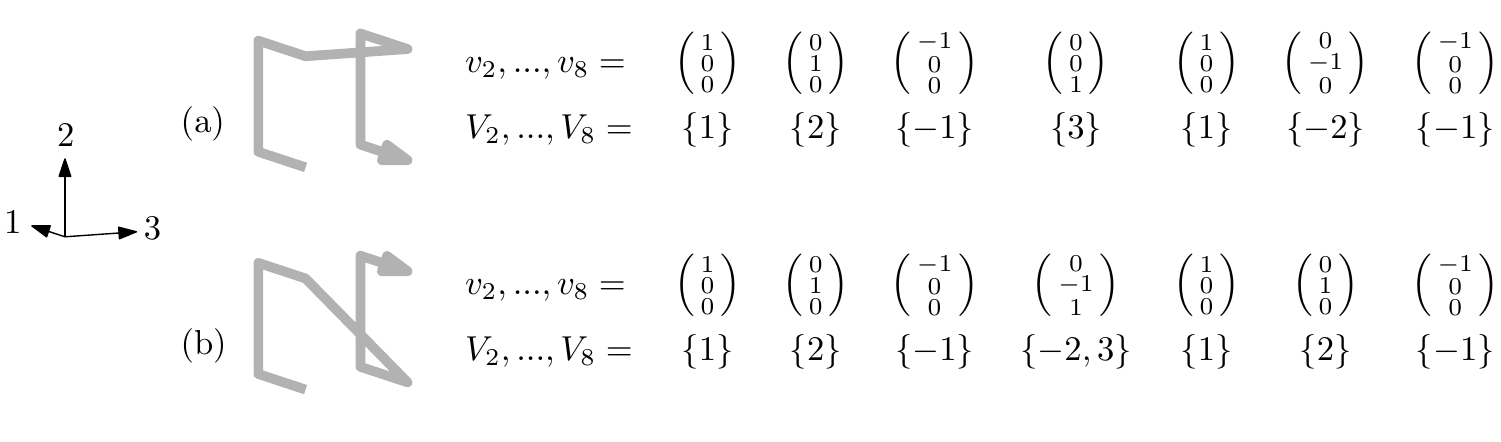}
\caption{Two examples of a base pattern, the corresponding vectors $v_2,...,v_8$ and ``moves'' $V_2,...,V_8$.}
\label{fig:basepatternnotation}
\end{figure}

Assume the unit cube is centered at the origin. Each transformation $\gamma_i: C \rightarrow C$ is a symmetry of the unit cube and can be interpreted as a matrix $M_i$ such that $\gamma_i(x) = M_i x$, where $x$ is a point given as a column vector of its coordinates. Each row and each column of $M_i$ contains exactly one non-zero entry, which is either $1$ or $-1$. We specify such a matrix by a signed permutation of row indices, that is, a sequence of numbers $\Pi_i = \pi_i[1],...,\pi_i[d]$ whose absolute values are a permutation of $\{1,...,d\}$ and which corresponds to the matrix in the following way: the non-zero entry of column $j$ is in row $|\pi_i[j]|$ and has the sign of $\pi_i[j]$. We write the sequence $\pi_i[1],...,\pi_i[d]$ between $[$ and $\}$ to specify a forward traversal ($\chi_i(t) = t$), whereas we write the sequence $\pi_i[1],...,\pi_i[d]$ between $\{$ and $]$ to specify a reverse traversal ($\chi_i(t) = 1-t$). For example, the traversal from Figure~\ref{fig:needthedots}ac has the following permutations, in order from $C_1$ to $C_8$:\[
[ 3, 2, 1 \}, [ 3, 1, 2 \}, \{ 3, 1, -2 ], [ -2, -1, 3 \}, \{ -2, -1, -3 ], [ -3, 1, -2 \}, \{ -3, 1, 2 ], \{ 2, -3, 1 ].
\]
Note how our notation facilitates mapping the base pattern to the order in which the suboctants of $C_i$ are visited. For example, if $j$ is positive, a move $\{ j \}$, forward along the $j$-th coordinates axis, translates to a move $\{ \pi_i[j] \}$ within $C_i$. If we define $\pi_i[-j] = -\pi_i[j]$, then the translation also works for negative values of $j$.

A complete self-similar traversal order is now specified by listing the signed and directed permutations $\Pi_1,...,\Pi_{2^d}$, with, between each pair of consecutive permutations $\Pi_{i-1}$ and $\Pi_i$, the set $V_i$ that gives the location of $C_i$ relative to $C_{i-1}$. Depending on lay-out requirements, we may omit commas and/or we may write the numbers of a set $V_i$ or a signed permutation $\Pi_i$ below each other instead of from left to right; we will also omit braces around $V_i$. Thus we get the following description of the traversal from Figure~\ref{fig:needthedots}ac:\[
\descr{\fwd[3\\2\\1]\edge[1]\fwd[3\\1\\2]\edge[2]\fwd[3\\1\\\m2]\edge[\m1]\fwd[\m2\\\m1\\3]%
\edge[3]%
\rev[\m2\\\m1\\\m3]\edge[1]\rev[\m3\\1\\\m2]\edge[\m2]\rev[\m3\\1\\2]\edge[\m1]\rev[2\\\m3\\1]}.
\]
Note that we do not specify the location of $C_1$ explicitly, but it can be derived from the sets $V_2,...,V_{2^d}$: $C_1$ is on the low side with respect to coordinate $j$ if and only if $j$ appears in any set $V_i$ before $-j$ does, that is, if there is a set $V_i$ such that $-j \notin V_2,...,V_i$ and $j \in V_i$.

\subsection{Self-similar space-filling curves}\label{sec:selfsimilarcurves}
If a traversal has the property that consecutive segments of the unit interval are always matched to subcubes that touch each other, then, as $k$ increases, the up to two subcubes corresponding to the segments that share a point $t \in [0,1]$ must shrink to the same point $p \in C$. For $t \in (0,1)$, we thus have $\tau^{-}(t) = \tau^{+}(t)$. Moreover, the functions $\tau^{-}$ and $\tau^{+}$ are continuous. The traversal thus follows a \emph{space-filling curve} $\tau: [0,1] \rightarrow [0,1]^d$ given by $\tau(0) = \tau^{+}(0)$, $\tau(t) = \tau^{-}(t) = \tau^{+}(t)$ for $0 < t < 1$, and $\tau(1) = \tau^{-}(1)$. By construction, this curve is self-similar: for each $i \in \{1,...,2^d\}$ and $t \in [0,1]$ we have $\tau((i-1+t)\cdot 2^{-d}) = \rho_i(\gamma_i(\tau(\chi_i(t))))$. Moreover, the mapping is measure-preserving: for any set of points $S \subset [0,1]$ with one-dimensional Lebesgue measure $z$, the image $\bigcup_{x\in S} \tau(x)$ of $S$ has two-dimensional Lebesgue measure $z$.

Recall that in our graphical notation, we used the first-order and the second-order approximating curve. In general, we define the $k$-th-order approximating curve $A_k$ of a space-filling curve $\tau$ as the polygonal curve that connects the centre points of the $2^{kd}$ subcubes in a regular grid in the order in which they appear along $\tau$. In fact, the space-filling curve $\tau$ is equal to the limit of $A_k$ as $k$ goes to infinity. The first-order approximating curve $A_1$ is easy to draw, given a description of the curve in our numerical notation: the sets $V_i$ explicitly specify the directions of the successive edges of $A_1$ (see Figure~\ref{fig:basepatternnotation}). The $2^d-1$ edges of $A_2$ within any octant $C_i$ are also easy to draw, as their directions are obtained by applying the signed permutation $\Pi_i$ to the sets $V_2,...,V_{2^d}$. Note, however, that in $A_2$, the edges \emph{between} the octants do not necessarily have the same directions as in $A_1$. For example, axis-parallel edges may become diagonal, as in Figure~\ref{fig:graphicalnotation}b. Therefore one cannot obtain the edges of $A_2$ by taking the sequence of alternating permutations and edges of $A_1$ that define the curve and merely substituting transformations of $A_1$ for the permutations. This is in contrast to the properties of approximating curves in our work on hyperorthogonal well-folded curves~\cite{hyperorthogonal}, where the specific properties of the curves under study ensured that edges keep their orientation from one approximating curve to the next.

\subsection{Variations}\label{sec:othertraversals}
Traversals can also be defined based on other shapes than squares or cubes, or based on subdivision into fewer or more than $2^d$ parts. Examples in two dimensions include the triangle-based Sierpi\'nski curve~\cite{Sagan}, the fractal-based Gosper flowsnake curve~\cite{flowsnake}, and the nine-part Peano curve~\cite{Peano}. Such traversals are beyond the scope of this article, although our notation system is powerful enough to describe some of them (see Figure~\ref{fig:trianglecurves} for some examples).

\begin{figure}
\centering
\hbox to\hsize{%
\rlap{(a)}\vbox{\hsize=0.3\hsize
\includegraphics[width=0.8\hsize]{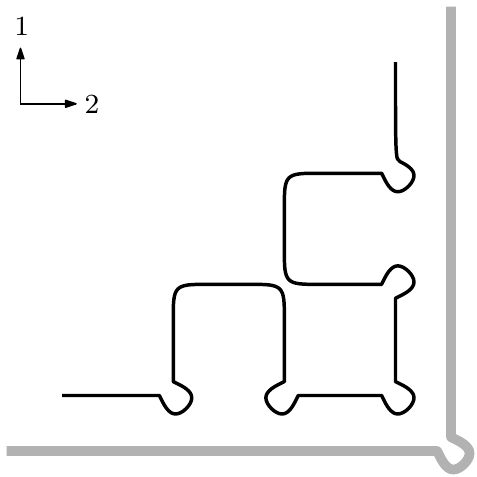}\\
$\descr{\fwd[1\\2]\edge[2]\rev[1\\\m2]\fwd[\m2\\1]\edge[1]\rev[\m2\\\m1]}$
}\hfill
\rlap{(b)}\vbox{\hsize=0.65\hsize
\includegraphics[width=0.9\hsize]{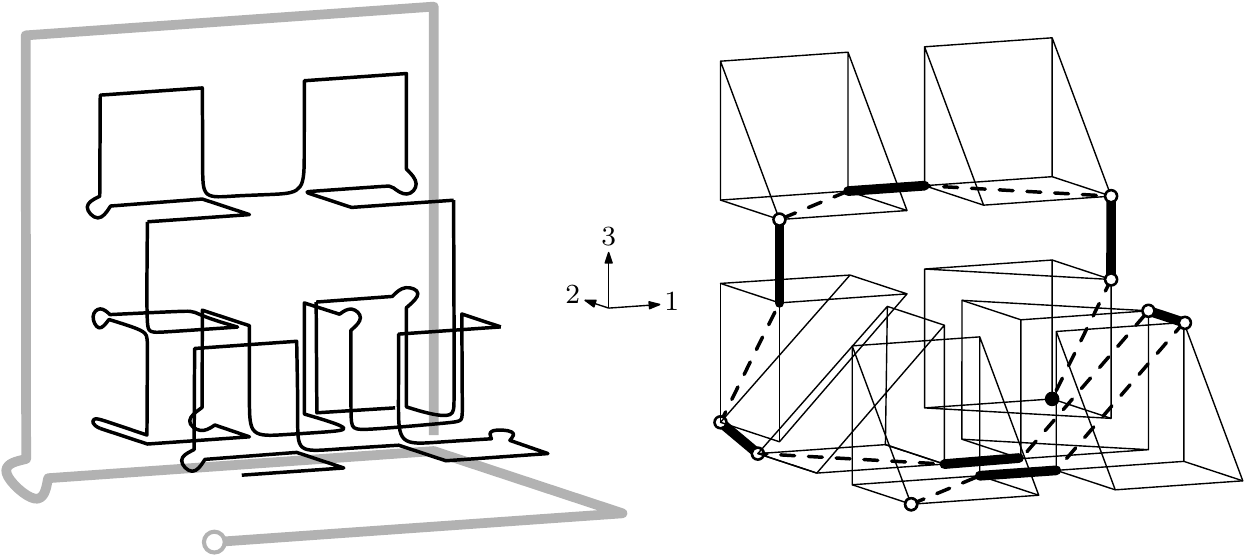}\\
$\descr{\fwd[1\\2\\3]\edge[1]\rev[\m1\\\m3\\\m2]\edge[2]\fwd[\m3\\\m1\\2]\edge[\m1]\rev[\m2\\1\\3]\fwd[\m2\\3\\1]\edge[3]\fwd[1\\2\\3]\edge[1]\rev[\m1\\2\\3]\edge[\m3]\fwd[\m3\\2\\\m1]}$
}}
\caption{(a) Definition of the Sierpi\'nski curve, which fills an isosceles right triangle. Note how the approximating curves visit vertices multiple times, since each visit corresponds to filling only half of the corresponding square in the underlying grid. In numeric notation, staying at a vertex for a second visit is indicated by an empty move between two signed permutations.
(b) Definition of a novel curve that fills the extrusion of an isosceles right triangle. It is the only eight parts' self-similar face-continuous (see Section~\ref{sec:generalproperties}) curve that fills this shape. Here, too, the approximating curves visit vertices multiple times. The diagram on the right, included for completeness, is explained in Section~\ref{sec:inventory}.}
\label{fig:trianglecurves}
\end{figure}

Non-self-similar traversals may be constructed from a set of multiple traversals in which each subcube is traversed by a scaled-down, rotated, reflected and/or reversed copy of a traversal from the given set; examples in two dimensions include Wierum's $\beta$- and $\Omega$-curves~\cite{betaomega} and the $AR^2W^2$-curve~\cite{asano}, which is constructed from a set of four curves~\cite{arrwwid,boxquality}. Non-self-similar traversals may also be constructed by concatenating rotated, reflected and/or reversed copies of a self-similar traversal: an example is Moore's cyclic variation of the Hilbert curve~\cite{moore}. The graphical and numerical notation systems described above suffice for self-similar traversals, but would have to be extended or adapted to be able to describe non-self-similar traversals. For the graphical notation, I describe such extensions in other work~\cite{inventory,arrwwid,boxquality}; I omit such extensions here, because in the next section, we will restrict the scope of this article to self-similar curves.

\section{Properties of Hilbert curves}\label{sec:properties}

\subsection{Essential properties of three-dimensional Hilbert curves}\label{sec:essentialproperties}

Within this publication, we restrict the discussion to traversals $\tau$ that are:\begin{itemize}
\item \emph{octant-based}: each of the $2^d$ subcubes of the unit cube is the image under $\tau$ of a consecutive interval within $[0,1]$;
\item \emph{self-similar}: the traversal $\tau$ restricted to any of the $2^d$ subcubes can be obtained by a linear transformation from the complete traversal $\tau$, as described in the previous section;
\item \emph{continuous}: this implies that if we consider a regular grid of $2^{kd}$ subcubes of the unit cube in the order which they are traversed by $\tau$, for any integer $k$, then consecutive subcubes in the traversal always touch each other.
\end{itemize}
In Section~\ref{sec:twodimensional} we will see that in two dimensions, these three properties constitute a minimal set of properties that uniquely defines the two-dimensional Hilbert curve. Therefore, one could say that any three-dimensional curve that fulfills these properties must be a three-dimensional Hilbert curve. This is indeed the approach which we will take in this article: we will call the properties of being octant-based, self-similar and continuous the three essential properties of Hilbert curves, and henceforth, we will consider a traversal to be a \emph{Hilbert curve} if and only if it has these three properties. This choice is justified in more detail in Section~\ref{sec:justification}.

The two-dimensional Hilbert curve also has other interesting, non-defining properties, which we might want to see in three-dimensional curves as well, for example to meet requirements of applications, to facilitate generalizations to even more dimensions, or simply to avoid confusion. Unfortunately, we can always think of a combination of properties of the two-dimensional curve that cannot be realized in three dimensions. Without the context of a particular application, we cannot decide a priori which properties to prefer at the expense of others.
Therefore, in this article, I will regard all additional properties to be optional. In Section~\ref{sec:optionalproperties} below we discuss a number of such properties and how to generalize them to three or more dimensions.

\subsection{Optional properties of three-dimensional Hilbert curves}\label{sec:optionalproperties}

Below is a list of non-defining properties of the two-dimensional Hilbert curve, stated in a dimension-independent way. The listed properties may be useful in higher dimensions as well. In Sections \ref{sec:inventory} and~\ref{sec:observations} we discuss what three-dimensional Hilbert curves have some of these properties.

\subsubsection{General properties}\label{sec:generalproperties}

\paragraph{Face-continuity.}
We say a space-filling curve is \emph{face-continuous}\footnote{Bader~\cite{Bader} uses the term \emph{face-connected}. I prefer \emph{face-continuous} because I find \emph{face-connected} easy to confuse with my definition of \emph{facet-gated} (see Section~\ref{sec:gateproperties}).} if, for any section of the curve, the interior of the region filled by that section is connected. In other words, for any $0 \leq a < b \leq 1$, the interior of the set $\bigcup_{t=a}^b \tau(t)$ must be connected. Concretely, for the case of $d$-dimensional Hilbert curves, this means that, at any level of recursion, cubes that are consecutive along the curve must share a $(d-1)$-dimensional face (hence the name), or equivalently, all edges of the approximating curves $A_1,A_2,...$, as defined in Section~\ref{sec:notation}, are axis-parallel.
Face-continuity thus generalizes the property of two-dimensional Hilbert curves that consecutive squares always share an edge.

Face-continuity may be considered instrumental in achieving good locality-preserving properties---see Section~\ref{sec:localityproperties}. However, requiring face-continuity also severely restricts the combinatorial possibilities for assembling a cube-filling curve from similar curves in each of eight octants. Under certain circumstances, better properties might be achieved by trading face-continuity for combinatorial flexibility.

\paragraph{Hyperorthogonality.}
Recall that a $d$-dimensional Hilbert curve can be described by a series of approximating polygonal curves $A_k$, whose edges connect the centres of consecutive cubes along the curve in a grid of $2^{dk}$ subcubes of the unit cube. We can identify the \emph{unsigned orientation} of an edge or a line $e$ by an unordered pair of antipodal points on the unit sphere, such that $e$ is parallel to the line through these points. We say that a $d$-dimensional Hilbert curve is \emph{hyperorthogonal} if and only if, for all positive integers $k$ and for all $n \in \{0,...,d-2\}$, the unsigned orientations of each sequence of $2^n$ consecutive edges of $A_k$ are those of exactly $n+1$ different axes of the Cartesian coordinate system~\cite{hyperorthogonal}.

Hyperorthogonality can be understood as a stronger (more restrictive) generalization of the two-dimensional Hilbert curve's property that consecutive squares always share an edge. This property of the two-dimensional curve can also be phrased as: each edge between the centres of consecutive squares must be parallel to an axis of the coordinate system. This is exactly what hyperorthogonality requires in the case $n = 0$, and this case is what hyperorthogonality boils down to if $d=2$. In three dimensions, hyperorthogonality requires the same (and thus, face-continuity), and adds the case $n = 1$: any pair of consecutive edges of an approximating curve must be orthogonal to each other. As we will see in Section~\ref{sec:observationslocality}, hyperorthogonal three-dimensional Hilbert curves have good locality-preserving properties, and Bos and I found that, for a certain metric of locality-preservation, this generalizes to higher dimensions~\cite{hyperorthogonal}.

\paragraph{Symmetry.}
A traversal order $\tau$ is \emph{symmetric} if there is an isometric transformation $\gamma$ such that $\tau^{+}(t) = \gamma(\tau^{-}(1-t))$ for all $t \in [0,1)$, and $\tau^{-}(t) = \gamma(\tau^{+}(1-t))$ for all $t \in (0,1]$. For a continuous traversal order, this is equivalent to $\tau(t) = \gamma(\tau(1-t))$ for all $t \in [0,1]$, and hence, $\tau(t) = \gamma(\tau(1-t)) = \gamma(\gamma(\tau(t))$. Thus, the curve $\tau$ is equal to its own reverse under the transformation $\gamma$, which must, in general, be a rotary reflection that is its own inverse.
Symmetry can have advantages for the implementation of efficient algorithms operating on the curve, since it allows the algorithm designer to choose between geometric transformations or reversing the direction, whatever is easiest to implement.

\paragraph{Metasymmetry.}
We say a traversal is \emph{metasymmetric} if there is a (not necessarily symmetric) linear transformation that maps the first half of the curve to the second half, and each half is metasymmetric itself. The property of being metasymmetric can be understood as a stronger (more restrictive) generalization of the two-dimensional Hilbert curve's symmetry and self-similarity: symmetry implies that sections of the curve of length $1/2$ are similar to each other; self-similarity implies that sections of length $1/2^d$ are similar to each other; metasymmetry requires for all positive integers $n$ that sections of the curve of length $1/2^n$ are similar to each other. Note, however, that, in deviation from the definition of plain symmetry, we do not require the similarities to be captured by symmetric transformations, that is, transformations that are their own inverse. Neither the two-dimensional Hilbert curve, nor any three-dimensional Hilbert curve, would fulfill a stronger definition of metasymmetry that requires each half of the curve to be fully symmetric in itself, that is, consisting of two quarters that can be mapped onto each other by a transformation that is its own inverse.

\paragraph{Palindromy.}
Consider an octant-wise traversal of the cube, and an \emph{interior facet}, that is, a facet $F$ between two octants $C_i$ and $C_j$, where $i < j$. For any $k \geq 1$, define $K = 4^{k-1}$ and consider $F$ subdivided into a regular grid of $K$ squares. Let $F_{i,1},...,F_{i,K}$ be these squares in the order in which the traversal visits the adjacent subcubes of $C_i$, and let $F_{j,1},...,F_{j,K}$ be the same squares in the order in which the traversal visits the adjacent subcubes of $C_j$. We say a traversal is \emph{facet-palindromic} if, for each interior facet $F$ between two octants $C_i$ and $C_j$ (note that there are twelve such facets), and for each level $k$, we have $F_{i,t} = F_{j,K+1-t}$. In other words, for any interior facet $F$, the order in which $F$ is traversed the second time around (during the traversal of $C_j$) is exactly the opposite of the order in which $F$ is traversed the first time around (during the traversal of $C_i$).

Palindromy is a property that allows simple and elegant implementations of finite element methods that use only stacks for storage of intermediate results---the so-called stack-and-stream method~\cite{Bader}. The two-dimensional Hilbert curve is facet-palindromic (with respect to the four edges between the quadrants). A three-dimensional facet-palindromic octant-wise continuous traversal is not known. When we consider the second-order approximating curves of the Hilbert curves in Figure~\ref{fig:imposters}, these curves appear to be facet-palindromic.\footnote{Thus these curves demonstrate that Bader's arguments (\cite{Bader}, p229) for the non-existence of palindromic three-dimensional Hilbert curves are inconclusive with respect to the definition of palindromy used here.} Unfortunately, the third-order approximating curves show violations of palindromy.

\begin{figure}[b]
\centering
\includegraphics[width=\hsize]{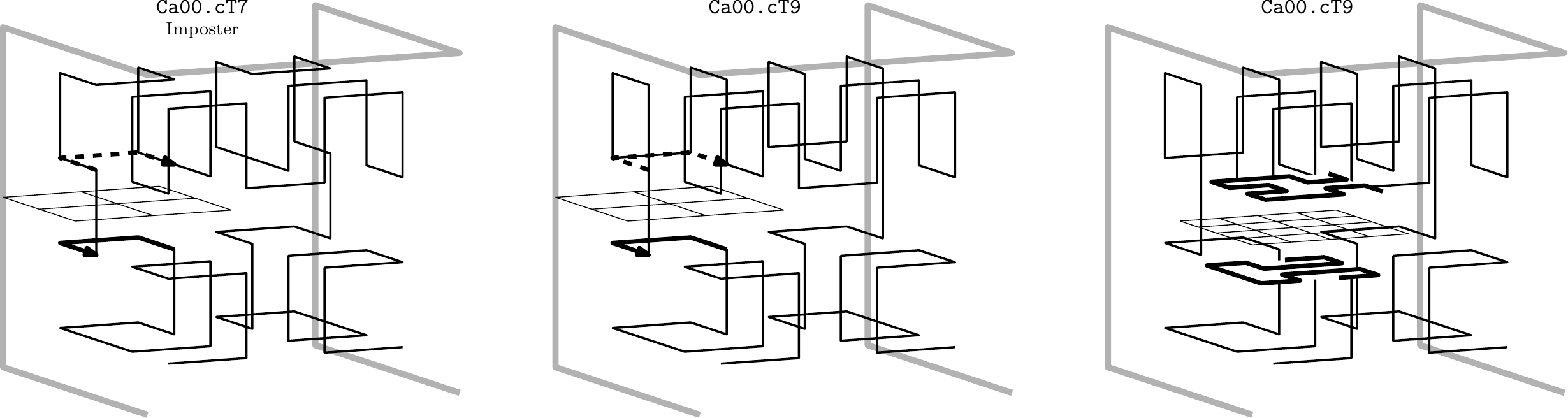}
\caption{Two curves, which we will later learn to identify as \curvename{Ca00.cT7} and \curvename{Ca00.cT9}, that seem palindromic at first sight. For example, consider the interior facet shared by the second and the third octant, as indicated in the figures. We see that the second time, the order in which we visit the four subsquares of this facet (dashed arrow) is exactly the opposite of the order in which we visit those subsquares the first time around (solid arrow). The reader may verify that also on the other eleven interior facets between the octants, the four quadrants are visited in the exact opposite order the second time around. However, if we expand the recursion and consider the subdivision of facets into sixteen squares, we find that the traversal orders from below and from above do not match on the facet between the first and the fourth octant. The right figure illustrates this for \curvename{Ca00.cT9}; for the other curve, \curvename{Ca00.cT7}, the situation is similar.}
\label{fig:imposters}
\end{figure}

\paragraph{Maximum facet-harmony.}
We say a $d$-dimensional traversal $\tau$ \emph{harmonizes} with an $n$-dimensional traversal $\tau'$ on a given $n$-dimensional subset $F$ of the unit cube, if $\tau$ restricted to the points of $F$ constitutes an isometric copy of $\tau'$. On all one-dimensional faces (edges) of the square, the two-dimensional Hilbert curve harmonizes with the unique and trivial one-dimensional Hilbert curve: the one-dimensional Hilbert curve traverses a line segment from one end to the other, and the two-dimensional curve visits the points on each edge of the square in order from one vertex to the other. Unfortunately, no three-dimensional Hilbert curve can harmonize with the two-dimensional Hilbert curve on each side of the cube (for a proof, see Appendix~\ref{apx:nofullharmony}), but it is possible to get five sides (and all edges) right, as we see in Figure~\ref{fig:allharmonious}. Therefore we say that a three-dimensional Hilbert curve has \emph{maximum facet-harmony} if it harmonizes with the two-dimensional Hilbert curve on five sides.

Note that harmony cannot be verified by only looking at the order in which the second-level subcubes are traversed and this may sometimes be misleading: one needs to make sure that the $(d-1)$-dimensional Hilbert order on the facets is maintained also when the grid is refined recursively. For example, Figure~\ref{fig:imposter-harmony} shows a three-dimensional Hilbert curve whose second-order approximating curve matches the two-dimensional Hilbert curve on five sides, but in recursion, harmony with the two-dimensional Hilbert curve is maintained on only on one of these sides.

\begin{figure}
\centering
\includegraphics[width=.9\hsize]{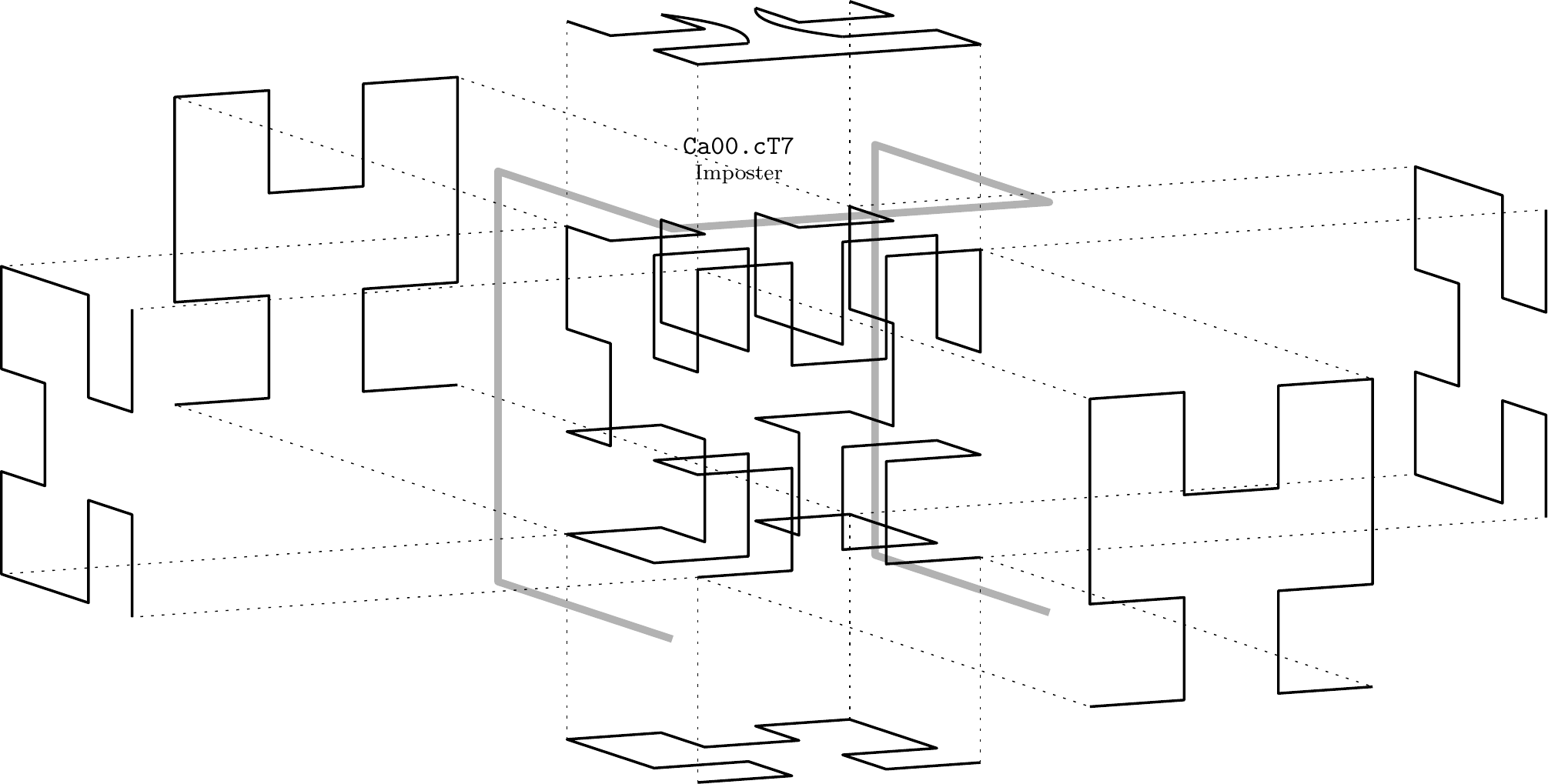}
\caption{This figure shows the second-order approximating curve of the three-dimensional Hilbert curve which we will later learn to identify as \curvename{Ca00.cT7}. The curve seems to harmonize with the two-dimensional Hilbert curve on five sides: all but the top facet. However, the fourth quadrant of the left facet is an image of the broken top facet, and thus, the third-order approximating curve of the left facet will not match the two-dimensional Hilbert curve. Similarly, in the third-order approximating curve, the harmony with the right facet is broken, and in the fourth-order approximating curve, the front and bottom facets will not match the two-dimensional Hilbert curve anymore either. Only on the back facet of the cube, this ``Imposter'' curve will actually harmonize with the two-dimensional Hilbert curve.}
\label{fig:imposter-harmony}
\end{figure}

Interest in harmonization properties arose from an application to the construction of R-trees, where it was desirable to use a traversal of the four-dimensional cube that, for points on a certain two-dimensional face of the cube, would degenerate to a two-dimensional Hilbert curve~\cite{jea}.

\paragraph{Full interior-diagonal harmony.}
A $d$-dimensional Hilbert curve has full interior-diagonal harmony if it harmonizes with the trivial one-dimensional Hilbert curve on all $2^{d-1}$ interior diagonals. Specifically, a three-dimensional Hilbert curve $\tau$ has full interior-diagonal harmony if, for each of the four interior diagonals, $\tau$ visits the points on the diagonal in order from one end to the other.

\paragraph{Well-foldedness.}
Let $G(d)$ denote the $d$-dimensional well-folded approximating curve, defined as follows: $G(0)$ is a single vertex, and $G(d)$, for $d > 0$, is the concatenation of $G(d-1)$, an edge in the direction of the $d$-th coordinate axis, and the reverse of $G(d-1)$. For example, $G(3)$ is the curve shown in Figure~\ref{fig:basepatternnotation}a. A Hilbert curve is \emph{well-folded}~\cite{hyperorthogonal} if its first-order approximating curve is $G(d)$ (modulo rotation, reflection and/or reversal). Note that the successive orientations of the edges in $G(d)$ indicate exactly which bits change when proceeding from one number to the next in the $d$ bits' binary reflected Gray code.

The regular structure of $G(d)$ provides a good basis for defining a family of Hilbert curves for any number of dimensions. Moreover, it can be instrumental in efficient computations with the curve.

One way to exploit well-foldedness is in the computation of an inverse of $\tau$, as demonstrated before by Bos and myself~\cite{hyperorthogonal}. An inverse of $\tau$ is a mapping $\tau^{-1}: [0,1]^d \rightarrow [0,1]$ such that $\tau(\tau^{-1}(x)) = x$). Such a mapping can be used to order points along the curve. To compute the order, one can maintain an interval $T$ for any point $p$ such that $\tau^{-1}(p) \in T$. Initially, one sets $T$ equal to $[0,1]$. Well-foldedness makes it possible to narrow down $T$ in steps: each step inspects only one bit of one coordinate of $p$ and then halves the size of $T$. To determine the order in which different points appear along the curve, one narrows down their corresponding intervals just enough so that they become disjoint and their order can be determined.

Another way to exploit well-foldedness is demonstrated by Lawder's algorithm~\cite{Lawder} to compute $\tau(t)$ for a given $t$ and vice versa, when $\tau$ is Butz's $d$-dimensional Hilbert curve. Lawder's algorithm exploits the properties of the binary reflected Gray code when using bitwise exclusive-or operations to translate between one-dimensional and $d$-dimensional coordinates in binary representation.

\subsubsection{Properties regarding specific points}\label{sec:gateproperties}

\paragraph{Being vertex-gated.}
For a traversal $\tau$, we call $\tau^{+}(0)$ and $\tau^{-}(1)$ the \emph{entrance gate} and the \emph{exit gate} of the traversal. A gate is a \emph{vertex gate}, an \emph{edge gate}, or a \emph{facet gate}, respectively, if, among all faces of the unit cube, the lowest-dimensional face that contains the gate is a vertex, an edge, or a $(d-1)$-dimensional facet. The two-dimensional Hilbert curve is \emph{vertex-gated}: both of its gates are vertex gates. An edge-gated variant has appeared in the literature and was found to have better locality-preserving properties according to some metrics, but that curve is not self-similar~\cite{boxquality,hungershoefer,betaomega,yoon}. In three dimensions, we may consider the possibilities of \emph{vertex-gated, edge-gated, and facet-gated} curves (where both gates are vertex gates, edge gates, or facet gates, respectively), and \emph{vertex-edge-gated, vertex-facet-gated, and edge-facet-gated} curves (where the two gates have the two different types mentioned). I am not aware of any advantages of disadvantages of specific gate types for any practical purpose, but, as we will see later, case distinctions by gate type will be very useful in analysing what three-dimensional Hilbert curves exist and what other properties they have.

\paragraph{Being edge-crossing.}
We say a traversal is \emph{edge-, facet-, or cube-crossing}, respectively, if, among all faces of the unit cube, the lowest-dimensional face that contains \emph{both} gates is an edge, a $(d-1)$-dimensional facet, or the full cube, respectively. Similar to gate types, the ``crossing type'' may not be interesting by itself, but distinctions by crossing type will be instrumental in obtaining the results in this article.

\paragraph{Being centred.}
We say a curve $\tau$ is \emph{centred} if $\tau(1/2)$, the point half-way along the curve, is the centre of the $d$-dimensional cube. 

\subsubsection{Properties of the transformations within the octants}\label{sec:octantproperties}

\paragraph{Preserving order.}
We say a self-similar traversal is \emph{order-preserving} if it can be defined without reversals, that is, $\chi_i(t) = t$ for all $i \in \{1,...,2^d\}$.

Order-preserving curves are arguably less complicated to understand and use (but not necessarily more efficient) than curves that contain reversals. Existing literature on space-filling curves tends to allow (use) or disallow reversal without discussing it. Alber and Niedermeier only considered order-preserving curves in their work on higher-dimensional Hilbert curves~\cite{Alber}. Asano et al.~\cite{asano} and Wierum~\cite{betaomega} implicitly used reversal in the description of their (non-self-similar) two-dimensional quadrant-based curves.

Note that if a traversal is symmetric, the reversed curve cannot be distinguished from a suitably rotated and/or reflected, non-reversed copy. Therefore one can choose to define the transformations in the octants with only the symmetries of the cube and no reversals. Thus, symmetric traversals are always order-preserving.

\paragraph{Isotropy.}
We say a face-continuous Hilbert curve, that is, a Hilbert curve whose approximating curves $A_k$ have only axis-parallel edges, is \emph{edge-isotropic} if, in the limit as $k$ goes to infinity, there is an equal number of edges of $A_k$ parallel to each axis~\cite{jaffer}.
We say a, not necessarily face-continuous, traversal is \emph{pattern-isotropic} if, in the limit as $k$ goes to infinity, each transformation of the base pattern, modulo reversal, occurs equally often among the transformations in the $2^{kd}$ subcubes of the unit cube. (Clearly, for face-continuous curves, pattern-isotropy implies edge-isotropy.)

Note that we do not take the direction in which the pattern is traversed into account. For edge-isotropy, this would not make a difference: as $k$ goes to infinity, the net amount of travel in the direction of each axis, relative to the total length of $A_k$, approaches zero; therefore, parallel to each axis, there must be an equal number of edges in each direction. For pattern-isotropy, if we would take the direction into account, the two-dimensional Hilbert curve would not qualify. For example, the two-dimensional Hilbert curve traverses some squares from the bottom left to the top left corner, but never from the top left corner to the bottom left corner.

Isotropy, like fairness~\cite{mokbel}, may be instrumental in ensuring that the performance of applications that order objects along a space-filling curve does not depend on the orientation of patterns in the data, since an isotropic or fair space-filling curve does not favour any particular orientation. Moon et al.~\cite{moon} proved that the two-dimensional Hilbert curve, along with certain generalizations to higher dimensions, is edge-isotropic.

Note that we have not defined edge-isotropy for non-face-continuous curves. This would require dealing with a number of subtleties\footnote{In general, edges in approximating curves can have 13 different (unsigned) orientations: 3 orientations parallel to the coordinate axes; 6 orientations parallel to facet diagonals; and 4 orientations parallel to interior diagonals. What conditions would we impose on the relations between the frequency of edges in each orientation? One solution could be to consider the three groups of edges separately, depending on whether edges are parallel to edges, facet diagonals or interior diagonals of the unit cube. Taking the direction of the traversal into account can now make a difference. Furthermore, as observed in Section~\ref{sec:selfsimilarcurves}, edges may change orientation from one level or refinement to the next.} and it is not a priori clear what it is the most meaningful way to do so.

\paragraph{Shifting coordinates.}
We call a signed permutation $\Pi_i$ that encodes a transformation $\gamma_i$ a \emph{shift} if the permutation, without the signs, is either the identity permutation or a rotation in the permutation-sense of the word. In other words, $P_i$ is a shift if and only if, for all $j \in \{1,...,d\}$, we have $|\pi_i[j]| = |\pi_i[d]| + j \pmod d$. We say a self-similar traversal is coordinate-shifting if it can be defined in such a way that, for all $i \in \{1,...,2^d\}$, the signed permutation $\Pi_i$ that defines $\gamma_i$ is a shift.

Implementations of higher-dimensional Hilbert curves, such as Butz's~\cite{Butz,Lawder,MooreCode}, often exploit this property to avoid having to code for arbitrary permutations of the coordinates.

\paragraph{Standing}
We say a self-similar traversal is \emph{standing} if it can be defined in such a way that, for fixed $m, n \in \{1,...,d\}$ and for all $i \in \{1,...,2^d\}$, the signed permutation that encodes the transformations $\gamma_i$, without the signs, is either the identity permutation or swaps only the $m$-th and the $n$-th coordinate. Note that in two dimensions, any traversal is, trivially, both coordinate-shifting and standing, but in three or more dimensions, these two properties are mutually exclusive.

Similar to coordinate-shifting traversals, standing traversals may be easier to employ efficiently because an implementation does not need to be capable of handling all $d!$ possible permutations of the coordinate axes. The term ``standing'' derives from the fact that such curves can be drawn in a way that keeps the third coordinate vertical.\footnote{Thus, in the approximating curves, similarities between sections remain recognizable if edges in the horizontal $mn$-plane are drawn in a different style as compared to edges that travel in the third dimension, for example, gangways versus stairs, as in Figure~\ref{fig:building}.}

\begin{figure}
\centering
\includegraphics[width=\hsize]{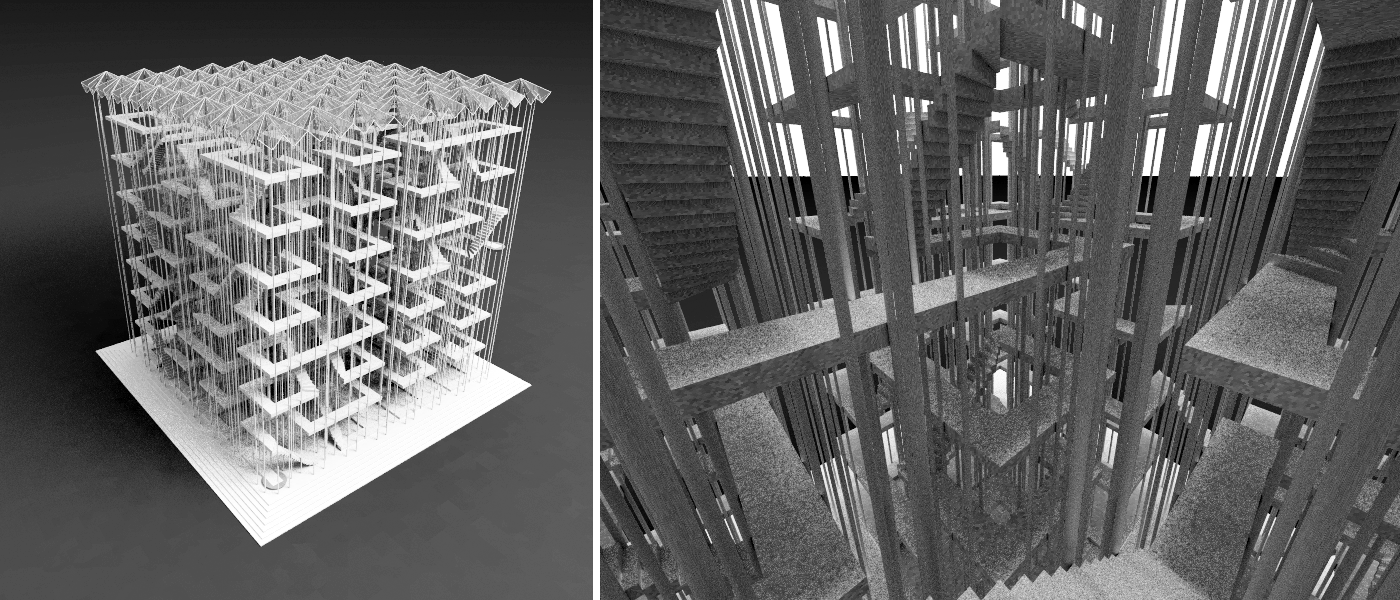}
\caption{An impression of a standing three-dimensional Hilbert curve (which we will later identify as \curvename{Cd00.ct.7h}) with the edges of the third-order approximating curve drawn as gangways and stairs. Left: view from the outside. Right: view inside. Colour pictures and more examples are available from the author's website at http://spacefillingcurves.net/.}
\label{fig:building}
\end{figure}

\subsection{Locality-preserving properties}\label{sec:localityproperties}
\label{sec:measures}

The space-filling curves discussed in this article are, by construction, \emph{measure-preserving}: the $d$-dimensional volume of the image of an interval $[a,b]$ under a traversal $\tau$ is equal to the length of the interval, that is, $b-a$.
Such space-filling curves tend to have \emph{locality-preserving} properties: points that are close to each other along the traversal, that is, in the domain of $\tau$, tend to be close to each other in $d$-dimensional space, that is, in the image of $\tau$, and vice versa. Many authors have worked on quantifying the locality-preserving properties of space-filling curves in general, and the Hilbert curve and its generalizations to higher dimensions in particular.

More specifically, some authors have studied bounds on the (worst-case or average) distance between two points in $d$-dimensional space as a function of their distance along the curve~\cite{chochia,faloutsos,gotsman,niedermeier,niedermeier-manhattan}. These studies have been motivated by, among others, applications to load balancing in parallel computing. Other metrics consider the shapes of curve sections: I and other researchers have tried to calculate bounds on the (worst-case or average) perimeter, diameter, or bounding-box size of sections of the curve as a function of the volume of the curve section~\cite{boxquality,hungershoefer,betaomega}, again motivated by applications to load balancing or to the organization of spatial data in external memory.

To define such metrics more precisely, we need the following definitions. Given two points $p$ and $q$ in the unit cube, let $\delta_i(p,q)$ be the $L_i$-distance between $p$ and $q$. Given a Hilbert curve $\tau$ and two points $a$ and $b$ in the unit interval, let $C(a,b) = \bigcup_{t=a}^b \tau(t)$ be the set of points that appear on the curve between $\tau(a)$ and $\tau(b)$. Given a set $S$ of $d$-dimensional points, let $\mathrm{vol}(S)$, $\mathrm{diam}_i(S)$, $\bbox(S)$, $\mathrm{bball}_i(S)$, and $\mathrm{surface}(S)$ be the volume, $L_i$-diameter, the minimum axis-parallel bounding box, the minimum bounding $L_i$-ball, and the $(d-1)$-dimensional measure of the boundary of the set $S$, respectively. We can now define the following quality measures of a $d$-dimensional space-filling curve, where in each case, the maximum is taken over all pairs $a,b \in [0,1]$ with $a \leq b$:\begin{itemize}
\item \emph{$L_i$-dilation} or $\mathrm{WL}_i$ (for $i \in \{1,2,\infty\}$): the maximum of $\delta_i(\tau(a),\tau(b))^d/(b-a)$;
\item \emph{$L_i$-diameter ratio} or $\mathrm{WD}_i$ (for $i \in \{1,2,\infty\}$): the maximum of $\mathrm{diam}_i(C(a,b))^d/(b-a)$;
\item \emph{$L_i$-bounding ball ratio} or $\mathrm{WBB}_i$ (for $i \in \{1,2,\infty\}$): the maximum of $\mathrm{vol}(\mathrm{bball}_i(C(a,b)))/(b-a)$;
\item \emph{surface ratio} or \WS: the maximum of $(\mathrm{surface}(C(a,b))/2d)^{d/(d-1)}/(b-a)$;
\item \emph{bounding-box volume ratio} or \WBV: the maximum of $\mathrm{vol}(\mathrm{bbox}(C(a,b)))/(b-a)$;
\item \emph{bounding-box surface ratio} or \WBS: the maximum of $(\mathrm{surface}(\mathrm{bbox}(C(a,b)))/2d)^{d/(d-1)}/(b-a)$.
\end{itemize}
In fact, the $L_i$-dilation and the $L_i$-diameter ratio of a space-filling curve are equal for any~$i$, and
the $L_\infty$-diameter ratio and the $L_\infty$-bounding ball ratio are always equal as well (for proofs, see Appendix~\ref{apx:boundingball}). I conjecture that the same holds for the $L_2$-diameter ratio and the $L_2$-bounding ball ratio, but I can prove this only for two-dimensional space-filling curves (see Appendix~\ref{apx:boundingball}) and I have not found a proof for three-dimensional space-filling curves.

In a previous publication on two-dimensional space-filling curves we described algorithms to compute bounds on $\mathrm{WL}_i$, \WBV, and \WBS\ for any given curve~\cite{boxquality}. We have also implemented higher-dimensional versions of these algorithms, including an algorithm to compute \WS~\cite{sasburg}, and used these algorithms to analyse the curves discussed in the next sections of this article. I will present the results in Section~\ref{sec:observationslocality}. Note, however, that it is not really clear how meaningful differences between curves on metrics of locality-preservation are, as the metrics tend to be the result of formalizing a much simplified account of what may be relevant for applications. Moreover, in practice, metrics that consider averages rather than worst cases may be more relevant, but average-case metrics are non-trivial to define~\cite{boxquality} and tend to be much more difficult to compute efficiently and accurately for large numbers of curves~\cite{sasburg}. Nevertheless, if we can establish that a possible three-dimensional Hilbert curve is particularly good (or bad) according to some metric of locality-preservation, then, it is, of course, an interesting curve to study: we may want to inspect such curves to see what qualitative properties of their structure cause it to perform so well (or badly) according to these metrics.

Other types of locality-preservation metrics studied in the literature include bounds on the average distance between points along the curve as a function of their distance in $d$-dimensional space~\cite{fishburn,mitchison,wierumpathlength,yoon}. However, non-trivial worst-case bounds are not possible in this case: there will always be pairs of points that are very close to each other in $d$-dimensional space but very far apart along the curve~\cite{gotsman}. Mokbel et al.\ define metrics that capture to what extent a traversal differs from sorting points in ascending order by one coordinate, and how these differences are distributed over the $d$ coordinates~\cite{mokbel}. One may also consider the number of contiguous sections of the curve that are needed to adequately cover any given query window in the unit cube~\cite{asano,arrwwid,moon,xu}. As I established through Observation~3 and Theorem~9 in my previous work on this topic~\cite{arrwwid}, if the query window is a cube, seven or eight sections of any three-dimensional Hilbert curve are sufficient and in the worst-case necessary for an approximate cover. An exact cover requires an unbounded number of curve sections in the worst case, unless one assumes the query range to be aligned with the grid of $2^{kd}$ subcubes at a particular depth~$k$~\cite{moon,xu}. Either way, it is questionable whether these worst-case metrics of cover quality capture the differences between the curves within the scope of this article well. Attempts at average-case analysis~\cite{arrwwid,moon,xu} suggest that what really matters are the orientations of the edges of the approximating curves: axis-parallel edges, modelling face-continuous curves, are good; curves with diagonal edges may be less good.

\subsection{Justification of the essential properties}\label{sec:justification}
In this section I will further justify the choice of octant-based self-similarity as the property that distinguishes three-dimensional Hilbert curves from other space-filling curves. In other words, this section is about why I use the label ``Hilbert curve'' in the way I do. The reader who is convinced already that the octant-based self-similar space-filling curves are a category of space-filling curves worth studying and who does not care too much about what to call them, may prefer to skip this section.

I considered three ways of generalizing the definition of Hilbert's space-filling curve to three dimensions: (i) face-continuous octant-based space-filling curves; (ii) vertex-gated, face-continuous, octant-based space-filling curves; (iii) self-similar octant-based space-filling curves.

\paragraph{(i) Face-continuous octant-based space-filling curves}
In the article in which Hilbert presents his continuous traversal of a square, Hilbert describes it as following a recursive subdivision into quadrants, and writes that each square along the curve should share an edge with the previous square\footnote{``die Reihenfolge der Quadrate [ist] so zu w\"ahlen [], dass jedes folgende Quadrat sich mit einer Seite an das vorhergehende anlehnt.''~\cite{Hilbert}}. Face-continuity is a possible generalization of the latter condition to higher dimensions. However, note that it is not enough to unambiguously define the two-dimensional Hilbert curve as we know it. If all we require is that the curve be face-continuous and quadrant-based, then, in every refinement step, we can choose \emph{any} of the subsquares of the starting square from the previous level as our new starting square. Wierum's $\beta\Omega$-curve~\cite{betaomega} would qualify as a Hilbert curve, along with an infinite number of other curves, in two-dimensional space already.

\paragraph{(ii) Vertex-gated, face-continuous, octant-based space-filling curves}
To disambiguate the definition of Hilbert's two-dimensional curve, we could add the condition that the curve be vertex-gated. Thus, the two-dimensional Hilbert curve is uniquely defined (see Theorem~\ref{thm:uniqueqfcvg} in Section~\ref{sec:twodimensional}). As we will discuss in Section~\ref{sec:howmanyin3D}, in three dimensions, infinitely many curves would qualify.

\paragraph{(iii) Self-similar octant-based space-filling curves}
Another way to disambiguate the definition of Hilbert's two-dimensional curve is to require that the curve be self-similar. This condition, together with the requirement that the traversal is quadrant-based and that each square touches the previous square in at least a \emph{vertex}, is enough to uniquely determine the two-dimensional Hilbert curve (see Theorem~\ref{thm:uniqueqss} in Section~\ref{sec:twodimensional}). The main message of Peano's and Hilbert's publications was that, surprisingly at the time, there are \emph{continuous} surjective mappings from one- to higher-dimensional space. Assuming that Hilbert indeed intended to define a \emph{self-similar} traversal that visits the square quadrant by quadrant, Hilbert had to include a condition that would narrow the scope to the only \emph{continuous} traversal of this type. For that purpose, in the two-dimensional setting, it did not matter whether he required that each square share at least a vertex with the previous one, or an edge: there is only one solution. However, in three dimensions it makes a difference. Given the context, we may understand the shared-edges condition merely as Hilbert's instruction to ensure continuity at all, not specifically face-continuity, and therefore we generalize it to higher dimensions by requiring that consecutive cubes always share at least one vertex.

\bigskip
Given these three options, in this article we choose the third one: we define a three-dimensional Hilbert curve as a self-similar, continuous, octant-by-octant traversal. The restriction to self-similar curves ensures compact descriptions that allow for efficient analysis of the curves and effective use in software. By avoiding the other options' restrictions to face-continuous (and possibly vertex-gated) curves, we can discover interesting curves that we would have missed otherwise.

\paragraph{The one-dimensional Hilbert curve} With the essential properties as we define them, the one-dimensional Hilbert curve is also well-defined as the only self-similar, continuous, half-by-half traversal (modulo reversal): it is simply the curve $\tau_1: [0,1] \rightarrow [0,1]$ defined by $\tau_1(t) = t$, traversing the unit line segment from one end to the other.

\section{Necessary and sufficient conditions in two dimensions}
\label{sec:twodimensional}

In this section we prove that the two-dimensional Hilbert curve is a) the only quadrant-wise self-similar space-filling curve, and b) the only quadrant-wise face-continuous vertex-gated space-filling curve.

\begin{theorem}\label{thm:uniqueqss}
The quadrant-wise self-similar square-filling curve is unique.
\end{theorem}
\begin{proof}
To prove the theorem, we consider all combinations of gates $\tau(0)$ and $\tau(1)$ that could be considered:\begin{itemize}
\item[(i)] vertex gates at opposite ends of the same edge;
\item[(ii)] vertex gates at opposite ends of a diagonal;
\item[(iii)] one edge gate and one vertex gate at the end of the same edge;
\item[(iv)] one edge gate and one vertex gate that does not lie on the same edge;
\item[(v)] two edge gates on adjacent edges;
\item[(vi)] two edge gates on opposite edges.
\end{itemize}
We analyze these cases one by one. In all cases, we try to follow the curve through the four quadrants, assuming, without loss of generality, that we start in the lower left quadrant. The various cases are illustrated in Figure~\ref{fig:nootherqss}.

\begin{figure}
\centering
\includegraphics[width=\hsize]{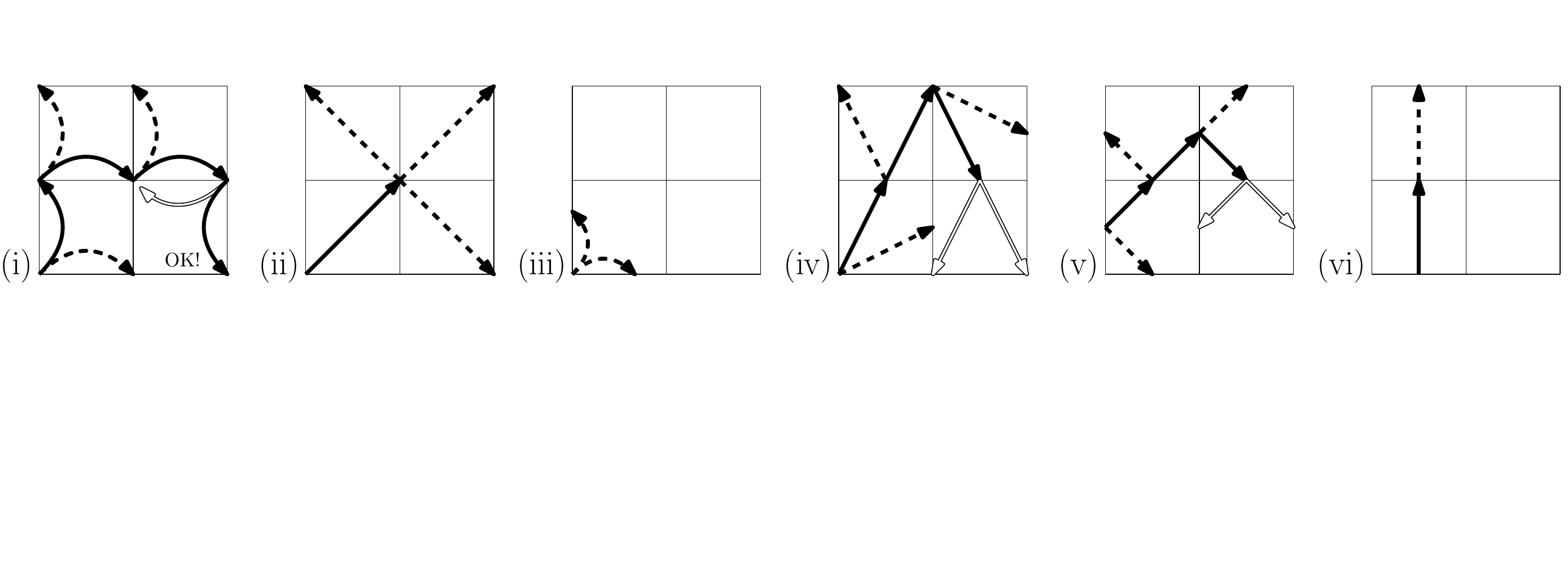}
\caption{This figure shows, for each of the six possible combinations of gate types for a square-filling curve, what sequences of gates between the quadrants we could realize assuming that the curve is quadrantwise self-similar. Solid arrows indicate a feasible sequence. Dashed arrows indicate dead ends leading to a point where we cannot connect to the next quadrant (either because the point is not incident on any other quadrant, or because it is only incident on the quadrant that must be the last to be visited but other unvisited quadrants remain). Hollow arrows lead to an exit gate of the fourth quadrant that is not consistent with the exit gate of the complete curve, under the assumptions of the case.}
\label{fig:nootherqss}
\end{figure}

\begin{itemize}
\item[(i)] \emph{Vertex gates at opposite ends of the same edge.}\\
Without loss of generality, assume the gates are located in the lower left and the lower right quadrant, so the lower left quadrant is the first to be traversed, and the lower right quadrant is the last to be traversed. We enter the lower left quadrant in the lower left corner, so we must leave it either through its lower right corner (in the middle of the bottom edge of the unit square) or through its upper left corner (in the middle of the left edge of the unit square). In the first case we would immediately enter the lower right quadrant, but this contradicts the assumption that this is the last quadrant to be traversed. So the only eligible case is the second case: we leave the lower left quadrant through its upper left corner in the middle of the left edge of the unit square. There we enter the upper left quadrant, which we must then leave through its lower right corner (the centre point of the unit square) in order to be able to connect to the third quadrant. This must then be the upper right quadrant (since the lower right quadrant must be the last), which we leave in its lower right corner in the middle of the right edge of the unit square, where we connect to the lower right quadrant.
Thus, the locations of the entrance and exit gates of all quadrants are unambiguously determined by the locations of the entrance and the exit gate of the unit square. By induction, it follows that the complete curve is uniquely determined by the choice of the edge that contains the gates---leading to four curves that are all equal modulo isometric transformations.
\item[(ii)] \emph{Vertex gates at opposite ends of a diagonal.}\\
We enter the lower left quadrant in the lower left corner, and leave it at its upper right corner, which is the centre of the unit square. Now, no matter which quadrant we traverse next, we must enter it at the centre of the unit square and leave it at a corner of the unit square. But there, there is no third quadrant to enter. Hence, with vertex gates at opposite ends of a diagonal, we cannot construct a self-similar curve.
\item[(iii)] \emph{One vertex gate and one edge gate on an incident edge.}\\
We enter the lower left quadrant in the lower left corner. Then we must leave it in the interior of either its bottom or its left edge. But there is no second quadrant to enter there. Hence, with this combination of gates, we cannot construct a self-similar curve.
\item[(iv)] \emph{One vertex gate and one edge gate on a non-incident edge.}\\
We enter the lower left quadrant in the lower left corner. Without loss of generality, assume we leave it through its top edge. Then the second quadrant must be the upper left quadrant, which we enter at its bottom edge, and leave at its top left or top right corner. At the top left corner, there is no third quadrant to connect to, so we must leave the second quadrant at its top right corner, and enter the upper right quadrant there. We leave through the bottom edge, entering the lower right and last quadrant, which we must then leave either at its bottom left or its bottom right vertex. But those points lie on an edge of the unit square that is incident to the entrance gate in the lower left corner, which contradicts the conditions of this case. Hence, with this combination of gates, we cannot construct a self-similar curve.
\item[(v)] \emph{Two edge gates on adjacent edges.}\\
If the gates lie on adjacent, that is, orthogonal edges, then the orientations of the edges containing the gates must alternate as we follow the curve from the entrance gate of the first octant to the exit gate of the last octant. But then the exit gate of the last octant lies on an edge of the same orientation as the entrance gate of the first octant, which contradicts the assumption that the gates of the unit cube lie on edges of different orientations. Hence, with this combination of gates, we cannot construct a self-similar curve.
\item[(vi)] \emph{Two edge gates on opposite edges.}\\
If the gates lie on opposite, that is, parallel edges, no curve can cross both the horizontal and the vertical centre line of the cube to reach the top right quadrant. Hence, with this combination of gates, we cannot construct a self-similar curve.
\end{itemize}
Thus, the only case that works out, is case (i), and it does so in a unique way.
\end{proof}

The conditions of Theorem~\ref{thm:uniqueqss} constitute a minimal set that uniquely defines the Hilbert curve. If we drop any of the conditions, there are other traversals that fulfill the remaining conditions: Peano's curve \cite{Peano} is a self-similar square-filling curve that is not quadrant-based; Wierum's $\beta\Omega$-curve \cite{betaomega} is a quadrant-based square-filling curve that is not self-similar; the Z-order traversal \cite{Morton} constitutes a quadrant-wise self-similar square-filling traversal that is not a curve (it is discontinuous)\footnote{Alternatively, the Z-order traversal can be modelled as a curve by including straight line segments to bridge the discontinuities, as in Lebesgue's space-filling curve \cite{Lebesgue}. However, then it does not comply with our framework of quadrant-based traversals as described in Section~\ref{sec:notation}. In other words, we can model Z-order as a curve but then it is not measure-preserving and not strictly quadrant-by-quadrant, or we can model Z-order as a measure-preserving quadrant-by-quadrant traversal, but then it is discontinuous.}.

\begin{theorem}\label{thm:uniqueqfcvg}
The quadrant-based face-continuous vertex-gated space-filling curve is unique.
\end{theorem}
\begin{proof}
A face-continuous quadrant-wise space-filling curve must traverse the squares of any grid of $2^k$ times $2^k$ squares one by one, such that each square (except the first) shares an edge with the previous square. Otherwise the curve would contain a section that fills two squares that are consecutive along the curve but do not share an edge: such a section would have a disconnected interior, and thus the curve would not be face-continuous. In particular, this means that $A_1$ must have the familiar $\Pi$-shape (see Figure~\ref{fig:nootherqfcvg}a), modulo reflection and rotation.

\begin{figure}
\centering
\includegraphics[width=\hsize]{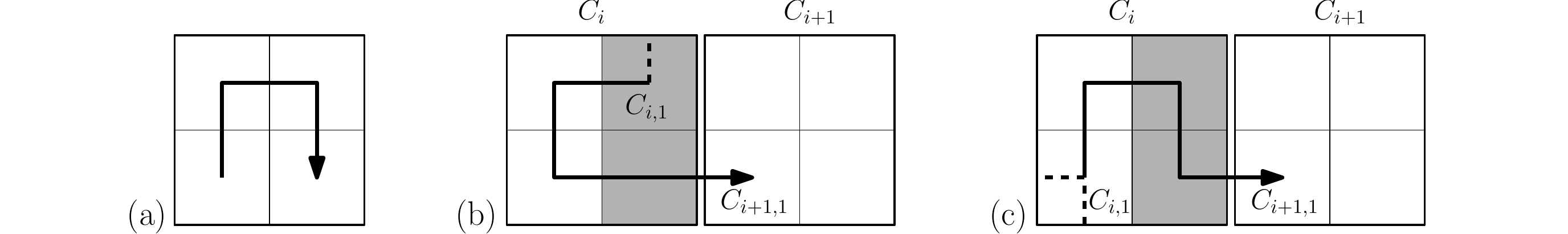}
\caption{(a) The base pattern of any quadrant-by-quadrant, face-continuous curve. (b,c) The course of $A_k$ within a $(k-1)$-level subsquare $C_i$ is uniquely determined by the $k$-level square $C_{i,1}$ where the curve enters $C_i$, and the location of the next $(k-1)$-level square $C_{i+1}$. The shaded squares are $X$ and $Y$, that is, the subsquares of $C_i$ that are adjacent to $C_{i+1}$.}
\label{fig:nootherqfcvg}
\end{figure}

Furthermore, given the order in which the squares of a $2^{k-1}$ by $2^{k-1}$ grid are traversed, for $k \geq 2$, the $k$-th order approximating curve $A_k$ is uniquely determined as follows. Let $C_1,...,C_{4^{k-1}}$ be the squares of the $2^{k-1}$ by $2^{k-1}$ grid, in the order in which they are visited. The first subsquare of $C_1$ must be the one in the corner of the unit cube. Now let, for any $i$, the first subsquare of $C_i$ be given as $C_{i,1}$, and let the two subsquares of $C_i$ that touch $C_{i+1}$ be labelled $X$ and $Y$. If $C_{i,1}$ is $X$ or $Y$, then we must put the $\Pi$-pattern inside $C_i$ such that it starts at $C_{i,1}$ and ends at $Y$ or $X$, respectively, to be able to make the connection to $C_{i+1}$ (see Figure~\ref{fig:nootherqfcvg}b). Otherwise, we must put the $\Pi$-pattern in $C_i$ such that it starts at $C_{i,1}$ and ends at the unique square out of $X$ and $Y$ that shares an edge with $C_{i,1}$ (Figure~\ref{fig:nootherqfcvg}c). The first subsquare of $C_{i+1}$ must now be the one that shares an edge with the last subsquare of $C_i$. Thus, the course of $A_k$ through $C_1,...,C_{4^{k-1}-1}$ follows by induction. The rotation or reflection of the $\Pi$-pattern inside $C_{4^{k-1}}$ follows from the requirement that we end in the corner of the unit cube. This is always possible, since we enter $C_{4^{k-1}}$ in another subsquare, which, by a simple parity argument, can be shown to be adjacent to the corner subsquare.
\end{proof}

The conditions of Theorem~\ref{thm:uniqueqfcvg} constitute a minimal set that uniquely defines the Hilbert curve. If we drop any of the conditions, there are other traversals that fulfill the remaining conditions: Peano's curve \cite{Peano} is a face-continuous vertex-gated space-filling curve that is not quadrant-based; the AR$^2$W$^2$-curve \cite{asano} is a quadrant-based vertex-gated space-filling curve that is not face-continuous; and Wierum's $\beta\Omega$-curve \cite{betaomega} is a quadrant-based face-continuous space-filling curve that is not vertex-gated.

\section{A naming scheme for three-dimensional Hilbert curves}\label{sec:scheme}

\subsection{A five-stage approach that highlights symmetries}\label{sec:fivestages}

A three-dimensional self-similar, octant-based traversal is defined by the order in which the octants are visited, together with the transformations in each of the octants. As we saw in Section \ref{sec:definitionbypermutations}, each such traversal is defined by a string of at most 45 numbers from $\{1,-1,2,-2,3,-3\}$, 8 opening brackets (rectangular or curly) and 8 complementary brackets (rectangular or curly). Thus, there can be only a finite number of such traversals. However, enumerating them by simply trying all possible permutations of the octants, together with all possible combinations of reflections, rotations and reversals in the octants (note that there are $2^d \cdot d! \cdot 2 = 96$ choices per octant), would be infeasible. Moreover, it would be quite difficult to recognize such things as symmetric curves or pairs of curves that share the same sequence of gates between octants.

Therefore, in this section we set up an alternative approach. We will distinguish five levels of detail in the description of the curves. On the coarsest level, we specify a \emph{partition}, that is, which octants lie in the first half of the traversal and which octants lie in the second half. On the next level, we specify the \emph{base pattern}, that is, in which order the octants in each half are visited. On the third level, we specify the locations of the entrance and exit gates of the curve. On the fourth level, we specify the \emph{gate sequence}, that is, the locations of the entrance gates (first points visited) and exit gates (last points visited) of all octants. On the fifth level, we specify the remaining details of the transformations of the curve within the octants. This generalizes the approach taken by Alber and Niedermeier~\cite{Alber}, who effectively consider the second, fourth and fifth level of detail. However, they encoded only the 920 face-continuous, vertex-gated, order-preserving curves\footnote{Alber and Niedermeier counted 1\,536 curves with these properties, since they counted some curves twice which we consider to be equivalent: see Footnote~\ref{fn:alber920} in Section~\ref{sec:howmanyin3D} for details.}, whereas we will encode 10,694,807 different curves, as we will see later. We need to make the descriptions of each level much more powerful than in their work to allow us to describe curves that are not face-continuous, not vertex-gated and/or not order-preserving effectively.

Recall that we consider curves that can be transformed into each other by rotation, reflection, translation, scaling or reversal to be equivalent. We will set up our naming scheme such that we give a unique name to exactly one curve of each equivalence class.

Our five-stage approach allows us to enumerate curves by generating them with increasing amount of detail. We will see that when we get to the third level, there are only a limited number (a few hundred) possibilities. For each choice for a third-level specification, we can generate all feasible gate sequences, making sure that every triple of an octant with its entrance gate and its exit gate can be obtained by at least one transformation of the unit cube with its entrance and exit gates, and each octant's entrance gate matches the previous octant's exit gate. Then, for each possible gate sequence, we can enumerate all options for filling in the remaining details of the curve: these options consist in all combinations of independent choices for which transformation to use in which octant, out of all transformations that map the gates of the unit cube to the gates of the octants.

\subsection{The general format}\label{sec:generalformat}

In general, our curve names follow the pattern $Pcmn.gh.st.op.qr$. In this pattern, $P$ is a uppercase letter specifying the partition, $m$ and $n$ are hexadecimal digits specifying the order of octants within each half; $c$ is a lowercase letter specifying how the two halves fit together, and thus, $Pcmn$ encodes the base pattern. Next are two letters $g$ and $h$, specifying the location of the gates $\tau(0)$ and $\tau(1)$, and two hexadecimal digits $s$ and $t$, specifying the locations of the gates between the octants. Thus, $Pcmn.gh.st$ encodes a gate sequence. The remaining digits, $op.qr$, specify the transformations within the octants. Many curves have shorter names: depending on the gate sequence, the number of digits actually used to specify the transformations within the octants can be zero, two or four.

The first symbol in each pair ($m,g,s,o,q$) concerns the first half in the curve; the second symbol in each pair ($n,h,t,p,r$) concerns the second half. This is implemented such that a name that follows the pattern $Pcmm.gg.ss.oo.qq$ describes a curve of which the first and the second half are the same, modulo a transformation which is effectively encoded by $c$. As we will see later, the second-last pair of digits, $oo$, is redundant for such curves and therefore, for a symmetric curve, we may use a condensed form of the name that follows the pattern $Pcmm.gsq$.

Hexadecimal digits (represented by $s,t,o,p,q$ and $r$ in the aforementioned patterns) are used to specify gates or transformations within one half of the curve. Typically these hexadecimal digits arise from encoding one bit of information per octant, which results in a number that is actually better read as a binary number rather than a hexadecimal number. Therefore, for hexadecimal digits we use symbols that are reminiscent of their binary equivalents, as displayed in Table~\ref{tab:hexadecimal}. The four bits of a hexadecimal digit, in order from most significant to least significant (first encoded octant to last encoded octant) are represented by the absence (0) or the presence (1) of, respectively, a vertical stroke on the right, a high horizontal stroke, a horizontal stroke in the centre, and a low horizontal stroke. If the first bit is zero, there is a vertical stroke on the left. In practice, we approximate the shapes thus composed by standard letters and digits as shown in Table~\ref{tab:hexadecimal}.

\begin{table}
\caption{Symbols for hexadecimal digits}\label{tab:hexadecimal}
\def\zero{\rlap{\vrule height1.5ex width0.2ex}}
\def\rightone{\rlap{\kern0.8ex\vrule height1.5ex width0.2ex}}
\def\highone{\rlap{\vrule height1.5ex depth-1.3ex width1ex}}
\def\middleone{\rlap{\vrule height0.85ex depth-0.65ex width1ex}}
\def\lowone{\rlap{\vrule height0.2ex width1ex}}
\centering\begin{tabular}{|ccc@{  }>{\ttfamily}c|ccc@{  }>{\ttfamily}c|ccc@{  }>{\ttfamily}c|ccc@{  }>{\ttfamily}c|}
\hline
dec & bin & \multicolumn{2}{c|}{symbol} &
dec & bin & \multicolumn{2}{c|}{symbol} &
dec & bin & \multicolumn{2}{c|}{symbol} &
dec & bin & \multicolumn{2}{c|}{symbol} \\
\hline
\hline
0 & 0000 &  \zero                               & I & 4 & 0100 &  \zero\highone                       & T & 8 & 1000 &  \rightone                           & X & 12 & 1100 & \rightone\highone                   & 7 \\
1 & 0001 &  \zero\lowone                        & L & 5 & 0101 &  \zero\highone\lowone                & C & 9 & 1001 &  \rightone\lowone                    & J & 13 & 1101 & \rightone\highone\lowone            & Z \\
2 & 0010 &  \zero\middleone                     & h & 6 & 0110 &  \zero\highone\middleone             & P & 10 & 1010 & \rightone\middleone                 & 4 & 14 & 1110 & \rightone\highone\middleone         & 9 \\
3 & 0011 &  \zero\middleone\lowone              & b & 7 & 0111 &  \zero\highone\middleone\lowone      & E & 11 & 1011 & \rightone\middleone\lowone          & d & 15 & 1111 & \rightone\highone\middleone\lowone  & 3 \\
\hline
\end{tabular}
\end{table}

In the following sections, we will describe the details of our curve naming scheme level by level.

\subsection{First two levels: encoding the base pattern}\label{sec:basepatterns}
\label{sec:basepatternames}

A base pattern is identified by a string of four symbols $Pcmn$. I will first describe what values these symbols can have and how to interpret them. After that, we will discuss how each possible base pattern has a unique name.

\paragraph{Decoding a base pattern identifier}
The first symbol $P$ is one of $\{\curvename{C},\curvename{L},\curvename{N},\curvename{S},\curvename{X},\curvename{Y}\}$ and indicates which four octants are traversed by the first half of the curve. The six possibilities are described in Figure~\ref{fig:partitions} (first row) and Table~\ref{tab:partitions}. In the table, a vector \vtx{x\\y\\z} represents the octant that includes the unit cube vertex $(x/2,y/2,z/2)$, assuming a unit cube of volume 1, centered at the origin. Table~\ref{tab:partitions} and Figure~\ref{fig:partitions} (second row) also give a standard order in which these octants are traversed---we will see shortly how different traversal orders are encoded by the third symbol of the base pattern name.

\begin{figure}
\centering
\includegraphics{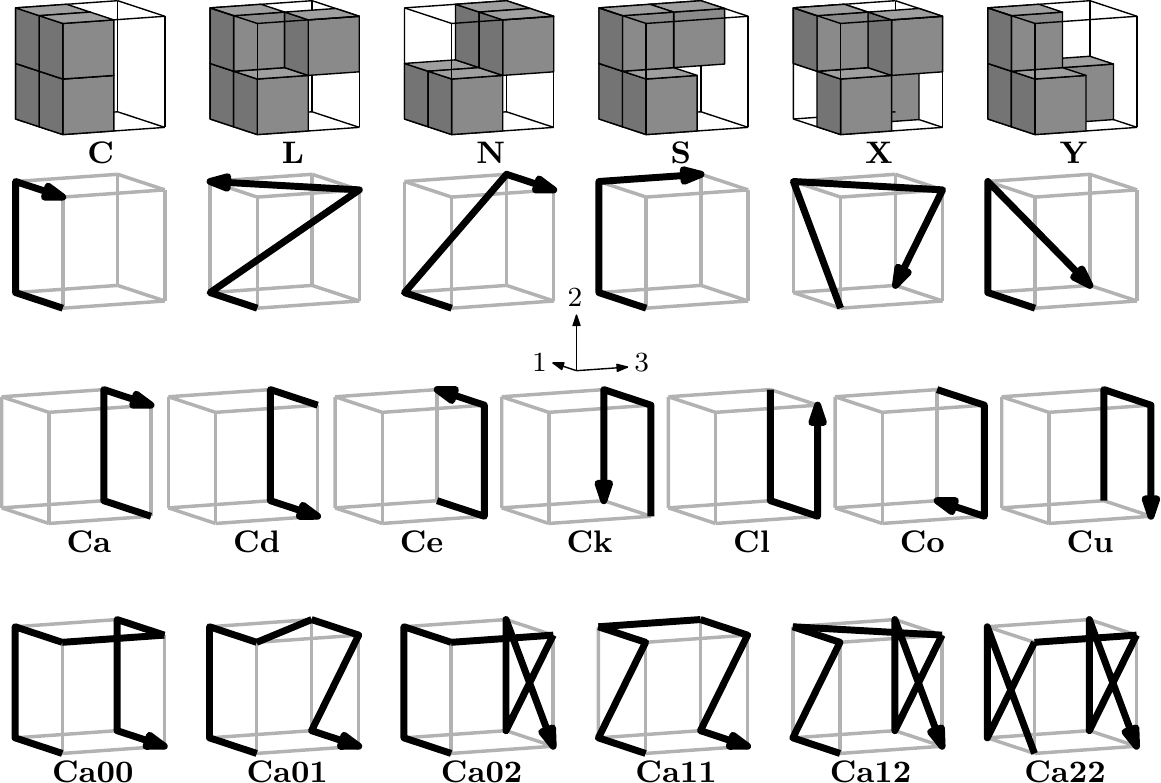}
\caption{Encoding of base patterns. First row: the six options for the set of octants in the first half of the traversal. Second row: the standard order in which these octants are traversed. Third row: the seven ways in which a standard order for the second half of a C-pattern can be obtained by a transformation of the first half. Fourth row: six examples of how the final octant order is obtained by permuting the traversal order in the first and the second half, and finally reversing the second half.}
\label{fig:partitions}
\end{figure}

\begin{table}
\caption{Possible combinations of octants visited by the first half of a traversal}\label{tab:partitions}
\centering\def\arraystretch{2}\begin{tabular}{|>{\ttfamily}cl|>{\ttfamily}cl|}
\hline
\textrm{symbol} & octants (w.l.o.g.) & \textrm{symbol} & octants (w.l.o.g.) \\\hline\hline
C   & \vtx{-1\\-1\\-1} \vtx{+1\\-1\\-1} \vtx{+1\\+1\\-1} \vtx{-1\\+1\\-1} & S   & \vtx{-1\\-1\\-1} \vtx{+1\\-1\\-1} \vtx{+1\\+1\\-1} \vtx{+1\\+1\\+1} \\
L   & \vtx{-1\\-1\\-1} \vtx{+1\\-1\\-1} \vtx{-1\\+1\\+1} \vtx{+1\\+1\\-1} & X   & \vtx{-1\\-1\\-1} \vtx{+1\\+1\\-1} \vtx{-1\\+1\\+1} \vtx{+1\\-1\\+1} \\
N   & \vtx{-1\\-1\\-1} \vtx{+1\\-1\\-1} \vtx{+1\\+1\\+1} \vtx{-1\\+1\\+1} & Y   & \vtx{-1\\-1\\-1} \vtx{+1\\-1\\-1} \vtx{+1\\+1\\-1} \vtx{+1\\-1\\+1} \\
\hline
\end{tabular}
\end{table}

\begin{table}
\caption{Encoding of the symmetries of the unit cube}\label{tab:transformations}
\centering\begin{tabular}{|@{\,}>{\ttfamily}rll@{\,}|}
\hline
\textrm{symbol} & \multicolumn{2}{l|}{symmetry} \\
\hline
\hline
a & [ 1, 2,-3] & reflection in plane orthogonal to 3rd coordinate axis \\
b & [ 1,-2, 3] & reflection in plane orthogonal to 2nd coordinate axis \\
c & [-1, 2, 3] & reflection in plane orthogonal to 1st coordinate axis \\
d & [ 1,-2,-3] & 180 degrees' rotation around line parallel to 1st coordinate axis \\
e & [-1, 2,-3] & 180 degrees' rotation around line parallel to 2nd coordinate axis \\
g & [-1, 3, 2] & 180 degrees' rotation around line through edge midpoints $(0,-1,-1)$ and $(0,1,1)$\\
h & [-1,-3,-2] & 180 degrees' rotation around line through edge midpoints $(0,-1,1)$ and $(0,1,-1)$\\
i & [ 3,-2, 1] & 180 degrees' rotation around line through edge midpoints $(-1,0,-1)$ and $(1,0,1)$\\
j & [-3,-2,-1] & 180 degrees' rotation around line through edge midpoints $(-1,0,1)$ and $(1,0,-1)$\\
k & [ 2, 1,-3] & 180 degrees' rotation around line through edge midpoints $(-1,-1,0)$ and $(1,1,0)$\\
l & [-2,-1,-3] & 180 degrees' rotation around line through edge midpoints $(-1,1,0)$ and $(1,-1,0)$\\
o & [-1,-2,-3] & point reflection with respect to the centre of the cube \\
q & [ 1, 3,-2] & 90 degrees' rotation around line parallel to 1st coordinate axis \\
r & [ 3, 2,-1] & 90 degrees' rotation around line parallel to 2nd coordinate axis \\
s & [ 2,-1, 3] & 90 degrees' rotation around line parallel to 3rd coordinate axis \\
u & [ 2,-1,-3] & \curvename{a} combined with \curvename{s} \\
w & [-3,-1,-2] & reflection combined with 120 degrees' rotation around interior diagonal \\
x & [ 3,-1, 2] & reflection combined with 120 degrees' rotation around interior diagonal \\
y & [-3, 1, 2] & reflection combined with 120 degrees' rotation around interior diagonal \\
z & [ 3, 1,-2] & reflection combined with 120 degrees' rotation around interior diagonal \\
\hline
\end{tabular}
\end{table}

The second symbol $c$ is a lower-case letter that specifies a rotary reflection that is a symmetry of the unit cube and maps the set of octants in the first half of the traversal to the set of octants in the second half of the traversal. Thus, we also get a standard order for the traversal of the octants in the second half of the curve. The possible values for $c$ and the corresponding symmetries are listed in Table~\ref{tab:transformations}: each symmetry is given as a signed permutation (as described in Section~\ref{sec:definitionbypermutations}) enclosed in square brackets. The third row of Figure~\ref{fig:partitions} gives an example of the resulting traversal orders for the octants in the second half of a traversal.

The third symbol, $m$, in the name of a base pattern indicates the permutation of the octants within the first half of the pattern. The default is that the first four octants, as encoded by the first symbol, are visited in the order indicated in the last column of Table~\ref{tab:partitions} and the second row of Figure~\ref{tab:partitions}. Other orders are indicated by a pentadecimal number according to Table~\ref{tab:permutations}. For example, if a name of a base pattern starts with \curvename{Se9}, then the first four octants are the 2nd, 3rd, 4th, and 1st from those listed with S in Table~\ref{tab:partitions}, so, the first four octants, in order, are those with vertex coordinates $\frac12(+1,-1,-1)$, $\frac12(+1,+1,-1)$, $\frac12(+1,+1,+1)$, and $\frac12(-1,-1,-1)$. For more examples, refer to the fourth row of Figure~\ref{fig:partitions}.

\begin{table}
\caption{Encoding of permutations of octants in one half of a traversal}\label{tab:permutations}
\centering\begin{tabular}{|>{\ttfamily}rl|>{\ttfamily}rl|>{\ttfamily}rl|}
\hline
\textrm{symbol} & traversal order & \textrm{symbol} & traversal order & \textrm{symbol} & traversal order \\
\hline\hline
0 & 1st 2nd 3rd 4th & 6 & 2nd 1st 3rd 4th & c & 3rd 1st 2nd 4th \\
1 & 1st 2nd 4th 3rd & 7 & 2nd 1st 4th 3rd & d & 3rd 1st 4th 2nd \\
2 & 1st 3rd 2nd 4th & 8 & 2nd 3rd 1st 4th & e & 3rd 2nd 1st 4th \\
3 & 1st 3rd 4th 2nd & 9 & 2nd 3rd 4th 1st & & \\
4 & 1st 4th 2nd 3rd & a & 2nd 4th 1st 3rd & & \\
5 & 1st 4th 3rd 2nd & b & 2nd 4th 3rd 1st & & \\
\hline
\end{tabular}
\end{table}

The fourth symbol describing the base pattern indicates the permutation of the octants within the second half of the traversal, in reverse order, so from the eighth back to the fifth octant of the complete base pattern. The default is that the eighth octant back to the fifth, in order, are the ones corresponding to the octants listed in Table~\ref{tab:partitions}, in order, under the transformation indicated by the second symbol of the base pattern name. In other words, the traversal order for the second half is obtained by taking the ordered set of octants specified by $P$, applying the transformation specified by $c$, followed by the permutation specified by $n$, and finally reversing the order.

The fourth row of Figure~\ref{fig:partitions} shows several examples. Note that the approximating curve of base pattern \curvename{Ca00} is the curve $G(3)$ (see Section~\ref{sec:optionalproperties}), and thus, \curvename{Ca00} identifies the base pattern of well-folded curves.

\paragraph{Selecting a unique name for each base pattern}
The previous discussion of how to decode a base pattern name may raise two questions. First, for some base patterns there may be multiple ways to encode them: how do we select a unique name for a pattern, so that patterns that are equivalent modulo reflections, rotations and/or reversal get the same name? Second, can we give a name to \emph{each} possible base pattern in this way?

To deal with the first question, we restrict the names of base patterns to those that are implicitly listed in Table~\ref{tab:octantorders}. A base pattern name is valid if it meets the following three conditions:\begin{itemize}
\item the transformation $c$ should be listed in the second column in the row for partition $P$;
\item the permutations $m$ and $n$ should be listed in the third column in the row for partition $P$;
\item if $c$ is lexicographically smaller than '\curvename{p}', then $m \leq n$.
\end{itemize}

\begin{table}
\caption{All possible base patterns by name. A name identifies a base pattern if and only if it consists of a symbol from the first columm, a symbol from the second column, a symbol from the third column and another symbol from the third column of the same row, subject to the following restriction: if the second symbol denotes a symmetric transformation, then the third and the fourth symbol should be in lexicographical order. In effect, this means that names with the third and fourth symbol out of order must start with \curvename{Cu} or \curvename{Yz}\strut.}\label{tab:octantorders}
\centering\begin{tabular}{|>{\ttfamily}l>{\ttfamily}l>{\ttfamily}lrr|}
\hline
\textrm{partition} & \textrm{transformations} & \textrm{permutations} & \multicolumn{2}{c|}{number of patterns}\\
 &&& symmetric & asymmetric \\
\hline
\hline
C    & adeklou           & 012               & 18           &  27 \\
L    & al                & 012345678cde      & 24           & 132 \\
N    & ae                & 012345            & 12           &  30 \\
S    & ei                & 0123456789ab      & 24           & 132 \\
X    & abcghijkloqrswxyz & 0                 & 10           &   7 \\
Y    & hikoz             & 0236              & 16           &  40 \\
\hline
\multicolumn{3}{|r}{total}                    &104           & 368 \\
\hline
\end{tabular}
\end{table}

To understand why this gives us a unique name for each equivalence class of base patterns, the following observations are helpful.

First, up to rotation and reflection, there are indeed exactly six possibilities for how octants can be divided between the first half and the second half of the traversal, as listed in Table~\ref{tab:partitions}. The possibilities can easily be analysed by distinguishing between three cases: (i) there is a plane that separates the first half from the other (type \curvename{C}); (ii) there is a plane that separates three octants in the first half from the fourth (types \curvename{L}, \curvename{S}, and \curvename{Y}); (iii) any axis-parallel plane through the centre of the cube has two octants from each half on each side (types \curvename{N} and \curvename{X}). Henceforth, we assume any base pattern or traversal is rotated and/or reflected such that the first four octants have the coordinates as indicated in the table, where coordinates \vtx{x\\y\\z} indicate the octant that includes the unit cube vertex $(x/2,y/2,z/2)$.

Second, each of the sets identified by \curvename{C}, \curvename{L}, \curvename{N}, \curvename{S}, \curvename{X} and \curvename{Y} has certain symmetries in itself. This limits the number of permutations we need to encode with the third symbol of the base pattern name. For example, for \curvename{S}-patterns, we only need to encode permutations that start with the first or the second octant in the set---if we would want to start with the third or the fourth octant, we would instead apply the rotary reflection $[-3,-2,-1]$ to the whole pattern, so that we swap the first and the second octant with the fourth and the third octant, respectively. Thus, for \curvename{S}-patterns, the third symbol can be restricted to the range $\{\curvename{0},...,\curvename{b}\}$. By a similar argument, the values of the fourth symbol are restricted in the same way: any permutation outside the given range can always be obtained by combining a permutation within the given range with a rotary reflection of the second half of the pattern---in effect changing the choice of the transformation encoded by the second symbol. As indicated in Table~\ref{tab:partitions}, the required set of permutations is different for each partition, because it depends on the geometric arrangement of the octants within one half.

Third, if the transformation specified by the second symbol $c$ is symmetric (that is, equal to its own inverse), then applying it to the base pattern $Pcmn$ and reversing the order results in the base pattern $Pcnm$. Thus, if $n \neq m$, then these are two different names for the same equivalence class of base patterns. In that case we choose the lexicographically smallest name. Hence the third condition on base pattern names---note that among the transformations listed in Table~\ref{tab:transformations}, the symmetric transformations are exactly those with a symbol lexicographically smaller than '\curvename{p}'.

Fourth, with the second symbol we only need to be able to specify a limited subset of the 48 symmetries of the unit cube. This is because many symmetries of the unit cube do not map any of the sets \curvename{C}, \curvename{L}, \curvename{N}, \curvename{S}, \curvename{X} or \curvename{Y} to their complement, or they are redundant, because we could use another transformation to describe a reversed and/or reflected version of the same base pattern.
This is why only seven, not eight different transformations are applied to partition \curvename{C}.

Table~\ref{tab:octantorders} also lists the numbers of symmetric base patterns and the numbers of asymmetric base patterns per row. Note that a pattern is symmetric if and only if the transformation that maps the first half to the second half is symmetric (that is, $c \in \{\curvename{a},...,\curvename{o}\}$), and the first and the second half are permuted in the same way (that is, $m = n$).

Appendix~\ref{apx:verification} explains an automated approach to verifying that the table is correct and complete, providing exactly one name for each equivalence class of base patterns. Figure~\ref{fig:realizablebasepatterns} in Appendix~\ref{apx:schemes} shows all base patterns that we found to be realizable by three-dimensional Hilbert curves.

\subsection{Third level: encoding the entrance and exit gates}\label{sec:gates}
In our naming scheme, the encoding of the base pattern is followed by two symbols that encode the entrance and the exit gate, respectively.
These symbols are given in Table~\ref{tab:gatesymbols}. Note that the interpretation of the exit gate symbol is subject to the transformation that maps the octants in the first half of the order to the octants in the second half (see Figure~\ref{fig:gates} for an example).

\begin{table}
\caption{Gate symbols}\label{tab:gatesymbols}
\centering\begin{tabular}{|>{\ttfamily}lll|}
\hline
\textrm{symbol} & intuition & gate location \\
\hline
\hline
c      & corner & at vertex \\
r      & radial & on interior of edge parallel to 1st coordinate axis \\
v      & vertical & on interior of edge parallel to 2nd coordinate axis \\
t      & transverse & on interior of edge parallel to 3rd coordinate axis \\
f      & front  & on face orthogonal to 1st coordinate axis \\
g      & ground & on face orthogonal to 2nd coordinate axis \\
s      & side   & on face orthogonal to 3rd coordinate axis \\
\hline
\end{tabular}
\end{table}

\begin{figure}
\centering
\includegraphics{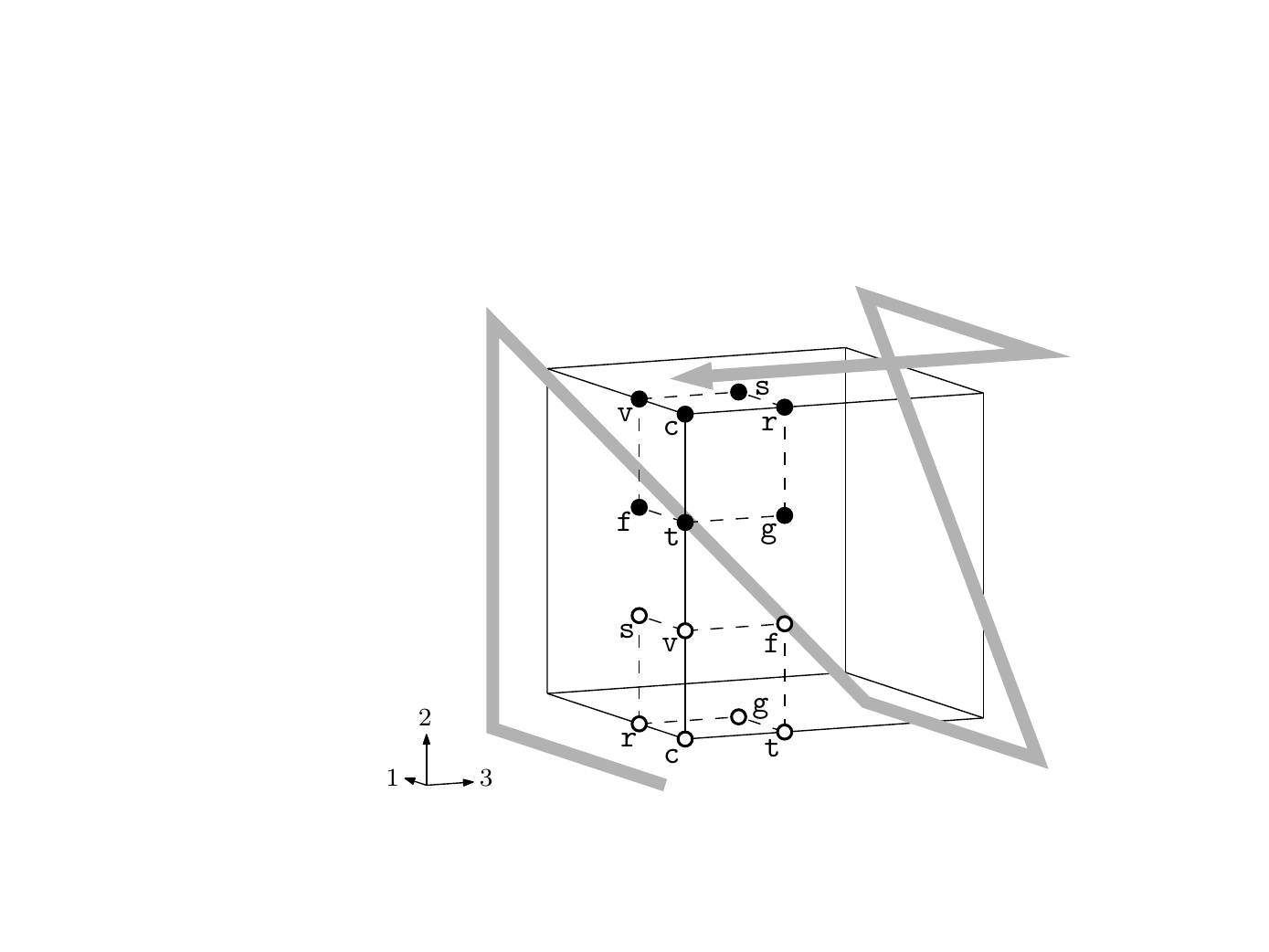}
\caption{Encoding of possible gate locations for the entrance gate (open dots) and the exit gate (closed dots) of a curve with base pattern \curvename{Yz00}.}
\label{fig:gates}
\end{figure}

Observe that we only encode the topological location of the gates, but not where exactly the edge and facet gates are located on the respective edges and facets. In the following section, we will find that the realizable combinations of entrance and exit gates for the unit cube are actually fairly limited, enough so that the topological location of the gates is sufficient to determine their exact location:\begin{itemize}
\item edge and facet gates in vertex-edge or vertex-facet gated curves are located at the midpoint of the edge or facet, respectively (Theorems~\ref{thm:vertex-edge-gates} and \ref{thm:vertex-facet-gates});
\item edge gates in edge-gated curves are located at distance 1/3 from the closest endpoint of the edge (Lemma~\ref{lem:edge-edge-gates-distance});
\item facet gates in facet-gated curves are located at distance 1/3 from two adjacent edges of the facet (Corollary~\ref{cor:facet-facet-gates-distance});
\item edge-facet-gated curves do not exist (Theorem~\ref{thm:edge-facet-gates}).
\end{itemize}

If a curve is symmetric in the base pattern but not in the locations of the gates, then we encode the curve in the direction that results in the name that comes first in lexicographical order. For example, we use \curvename{Ca00.gs}, not \curvename{Ca00.sg}.

\section{Inventory of three-dimensional Hilbert curves}\label{sec:inventory}

To set up the fourth and fifth levels of our naming scheme, we need to distinguish a number of cases depending on the locations of the entrance and exit gates, similar to the proof of Theorem~\ref{thm:uniqueqss}. In the next subsections we will discuss vertex-gated, vertex-edge-gated, vertex-facet-gated, edge-gated, edge-facet-gated, and facet-gated curves, respectively. In each subsection we will establish where exactly the entrance and exit gates must be located relative to each other, so that we can enumerate the corresponding gate sequences and the curves that implement them efficiently, and define an appropriate encoding.

\paragraph{Results}
In Section~\ref{sec:vertexgated} we find that there are 178 gate sequences for vertex-gated curves. The symmetric placement of the gates makes it possible to reflect and/or reverse the subcurves within the eight octants independently from each other. Thus, each gate sequence is realized by up to $(2 \cdot 2)^8 = 65\,536$ curves, some of which are equivalent, leading to a total of 10,691,008 different vertex-gated curves. These include all face-continuous curves save one, all order-preserving curves, all symmetric curves, and many well-folded curves. The vast possibilities in shaping the octants allow us to construct a hyperorthogonal curve, a maximally facet-harmonious curve, and many coordinate-shifting curves. However, none of the vertex-gated curves are centred or standing.

The vertex-edge-gated curves, discussed in Section~\ref{sec:vertexedgegated}, are much more flexible with respect to the placement of the connecting gates between the octants. They allow 2\,758 gate sequences, but the positions of the gates completely determine the transformations within the octants, so there is only one curve per gate sequence. Among them we find centred curves, well-folded curves, and standing curves. However, none of the vertex-edge-gated curves are face-continuous, order-preserving, symmetric, or coordinate-shifting.

There is only one facet-gated curve (Section~\ref{sec:facetgated}). It is face-continuous, well-folded, hyperorthogonal, and pattern-isotropic. It is of particular interest because it can be generalized to higher dimensions~\cite{hyperorthogonal}, and was found to have excellent locality-preserving properties (Section~\ref{sec:observationslocality}). One could view the curve as a three-dimensional version of the two-dimensional facet-gated (but non-self-similar) $\beta\Omega$-curve~\cite{betaomega}, which was shown to have slightly better locality-preserving properties than the original Hilbert curve~\cite{boxquality,betaomega,wierumpathlength,yoon}. However, in three and more dimensions the construction is actually simpler, namely self-similar, using only generalized $\beta$-shapes and no $\Omega$-shapes.

The vertex-facet-gated curves are few (Section~\ref{sec:vertexfacetgated}: 1,024 curves from 4 gate sequences), the edge-gated curves even fewer (Section~\ref{sec:edgegated}: 16 curves), and I have not been able to identify any particularly interesting curves among them. Edge-facet-gated curves cannot be realized (Section~\ref{sec:edgefacetgated}).

In total, we find that there are 10,694,807 different three-dimensional Hilbert curves.

\paragraph{Terminology and notation}
With each combination of entrance and exit gates, the fourth level, the gate sequence, is encoded with two hexadecimal digits after the third level. The first digit is for the first half of the curve; the second digit is for the second half of the curve as seen from the other end---that is, the second digit is for the first half of the reversed curve under the inverse of the transformation encoded by the second symbol of the base pattern name. In each case the two halves will be encoded in the same way. Therefore, in the following subsections, we will only describe how to encode the gates in the first half of the curve as a four digits' binary number, which is then written as a single hexadecimal digit. Typically, but not always, we use one bit per octant, in order of decreasing significance as we traverse the octants starting from the entrance gate. So the two hexadecimal digits, if written as eight binary digits, from left to right, encode the first, second, third, fourth, eighth, seventh, sixth and fifth octant, in that order.

In the cases in which the fifth level, concrete curves, needs one or two further pairs of digits, the same principles apply: the first digit is for the first half of the curve, and we use one bit per octant, starting with the most significant bit for the octant at the entrance or exit gate.

The following subsections include a number of figures that sketch three-dimensional Hilbert curves. For each curve there are two diagrams (see, for example, Figure~\ref{fig:edgecrossingcurves}). The drawing on the left is an annotated drawing of the second-order approximating curve and defines the space-filling curve in the way described in Section~\ref{sec:definitionbyfigure}. The drawing on the right is an exploded view of the eight octants that clarifies the gate sequence. In the exploded views, bold solid lines connect the exit gate of an octant with the entrance gate of the next octant---note that in the unexploded reality, these points coincide. Dashed lines connect the entrance gate of an octant with the exit gate of the same octant. For asymmetric curves, open dots mark the octant gates that correspond to the entrance gate of the entire curve: an open dot thus marks the exit or the entrance gate of an octant, depending on whether the transformation that maps the entire curve to the curve within the octant involves reversal or not.

Note that the curves are not always drawn with the same orientation of the coordinate system: some of the curves are rotated and/or reflected for a better view. For example, Figures \ref{fig:facetcrossingcurves}d and~\ref{fig:vertexedgegatedcurves}c show curves with the same base pattern \curvename{Cd00}, but one is rotated and reflected with respect to the other. To clarify this, each figure includes a drawing that indicates the positive direction of each coordinate axis.

Throughout the following subsections, given a traversal $\tau$, we use $\gate i$ to denote $\tau(i/8)$, that is, the exit gate of the $i$-th octant and the entrance gate of the $(i+1)$-st octant.

\subsection{Vertex-gated curves}\label{sec:vertexgated}

\begin{theorem}\label{thm:vertex-vertex-gates}
The gates of vertex-gated three-dimensional Hilbert curves are located at opposite ends of either an edge or a facet-diagonal of the cube.
\end{theorem}
\begin{proof}
The theorem only excludes the possibility of gates at opposite ends of an interior diagonal. The proof that this is not possible is completely analogous to case (ii) in the proof of Theorem~\ref{thm:uniqueqss}.
\end{proof}

In other words, vertex-gated curves are either edge-crossing or facet-crossing, but not cube-crossing. Whether the curve is edge-crossing or facet-crossing, depends on whether the first and the last octant of the base pattern lie along the same edge of the unit cube, or only on the same facet.

\subsubsection{Curves that are vertex-gated and edge-crossing}\label{sec:edgecrossing}

\paragraph{Gate sequences}
With vertex-gated, edge-crossing curves, the locations of the gates are restricted by the fact that they have to lie relative close to each other: subsequent gates must lie at opposite ends of an octant edge. In particular, $\gate 0$ is at a corner of the unit cube, $\gate 1$ must be in the middle of the edge of the unit cube that is shared by the first and the second octant, and $\gate 2$ must be in the middle of the facet of the unit cube that is shared by the second and third (and, necessarily, also the first) octant. For the first half of the curve, this leaves only $\gate 3$ and $\gate 4$ to be specified. The gate $\gate 3$ must lie either in the centre of the cube, or at the midpoint of the edge of the unit cube that is shared by the third and fourth octant. The gate $\gate 4$ must lie at the midpoint of a facet of the unit cube that is shared by the fourth and the fifth octant.

We keep the encoding of the first half of the traversal independent of the second half (which contains the fifth octant), and choose the following solution for the first half: the first bit encodes $\gate 3$ (0 if in the centre; 1 if on an edge), the remaining three bits in order of decreasing significance give the three coordinates of $\gate 4$ in order of increasing index (for each coordinate: 0 if in the centre, 1 if not).

For example, consider the gate sequence of the curve in Figure~\ref{fig:edgecrossingcurves}g, identified by $Pcmn.gh.st = \curvename{Si00.cc.LT}$. First consider the first half of the curve. The gate $\gate 3$ between the third and the fourth octant is the centre of the cube, so the first bit is a zero. The gate $\gate 4$ between fourth and the fifth octant has coordinates $(0,0,1)$, so the next three bits are 001, and therefore $s = 0001 = \curvename{L}$ (see Table~\ref{tab:hexadecimal}). The gate $\gate 5$ between the third-last (sixth) and the fourth-last (fifth) octant is also the centre of the cube, so the first bit of $t$ is a zero. The gate $\gate 4$ has coordinates $(0,0,1)$; under the inverse of transformation $c = \curvename{i} = [3,-2,1]$ the coordinates are $(1,0,0)$, and therefore the next three bits of $t$ are 100. Thus $t = 0100 = \curvename{T}$.

Gate sequences for vertex-gated, edge-crossing curves can easily be enumerated by exhaustive search: it turns out there are 29 such gate sequences, which realize 14 different base patterns. Among these gate sequences are \curvename{Ca00.cc.hh} (Hil$^3_1$.A(b) from Alber and Niedermeier~\cite{Alber}), \curvename{Ca00.cc.h4} (Hil$^3_1$.A(d)), \curvename{Ca00.cc.TT} (Hil$^3_1$.A(c)), \curvename{Ca00.cc.44} (Hil$^3_1$.A(a)), and \curvename{Si00.cc.LT} (Hil$^3_1$.B(a),  reversed, and Hil$^3_1$.B(b)). From the 29 gate sequences, 10 sequences are symmetric and 19 sequences are asymmetric. A full list is given in Appendix~\ref{apx:schemes}, Table~\ref{tab:vertexgated}.

\paragraph{Curves}
To fully specify a vertex-gated curve, we need to specify how exactly the curve is transformed within each octant. For each octant $C_i$, the locations of the gates are given by the gate sequence. This still leaves two binary choices per octant. The first choice is whether to use a forward or a reverse copy of the curve, that is, whether to map the entrance gate of $\tau$ to $\gate {i-1}$ and the exit gate to $\gate i$, or the other way around. The second choice is whether to use only rotation, scaling, translation and/or reversal to map $\tau$ to $\tau[(i-1)/8,i/8]$, or to employ also reflection. The difference between the two options for the second choice is a reflection in the diagonal plane that bisects $C_i$ and contains both gates $\gate {i-1}$ and $\gate i$.

We encode these choices by two pairs of hexadecimal digits following the name of the gate sequence. The general pattern of a curve name is then $Pcmn.\curvename{cc}.st.op.qr$, where $Pcmn.\curvename{cc}.st$ identifies the gate sequence, $o$ and $p$ specify the reversals in the first and the second half of the curve, respectively, and $q$ and $r$ specify the reflections in the first and the second half, respectively. Reversals and reflections are specified using one bit per octant. For the first half of the curve, 0 means non-reversed/non-reflected; 1~means reversed/reflected. For the second half of the curve, the meaning of 0 and 1 is modified according to the transformation that is used in the base pattern encoding, see Table~\ref{tab:ccencoding}.

\begin{table}
\caption{Interpretation of the reversal and reflection bits in the names of vertex-gated curves.}
\begin{tabular}{|l|cccc|}
\hline
& \multicolumn{2}{c}{reversal encoding} & \multicolumn{2}{c|}{reflection encoding}\\
& 0 & 1 & 0 & 1 \\
\hline
\hline
first half                             & non-reversed & reversed     & non-reflected & reflected     \\
second half, $c$ maintains orientation & reversed     & non-reversed & non-reflected & reflected     \\
second half, $c$ induces reflection    & reversed     & non-reversed & reflected     & non-reflected \\
\hline
\end{tabular}
\label{tab:ccencoding}
\end{table}

For example, consider the curve \curvename{Ca00.cc.44.hh.db} in Figure~\ref{fig:edgecrossingcurves}c. The first reversal digit is '\curvename{h}', binary 0010 (see Table~\ref{tab:hexadecimal}), so among the first four octants, only the third octant is reversed---in correspondence with the placement of the open dots in the figure. The second reversal digit is '\curvename{h}', binary 0010, but for the second half of the curve, the meaning of 0 and 1 is flipped, as is the order of the bits. Thus, among the last four octants, the third-\emph{last} octant is the only one that is \emph{not} reversed. The digit '\curvename{d}', binary 1011, tells us that that the curve is reflected in the first, third and fourth octant, but not in the second octant. The digit '\curvename{b}', binary 0011, should be interpreted with the meaning of 0 and 1 flipped (because the transformation $\curvename{a} = [1,2,-3]$ induces a reflection) and counting from the back, so the last and the second-last octant (corresponding to the zeros in 0011) are reflected, and the third-last and fourth-last octant (the ones in 0011) are not.

Figure~\ref{fig:edgecrossingcurves} shows several other examples of vertex-gated, edge-crossing curves. At first sight, most of the chosen examples may look very similar, but they each have unique properties that are not shared by the other curves.

\paragraph{Assigning unique names to curves}
If the gate sequence is symmetric then there is some redundancy in the encoding scheme that we just defined. Let $Pcmm.\curvename{cc}.ss$ denote a symmetric gate sequence for vertex-gated, edge-crossing curves. To avoid an unnecessary case distinction, note that Table~\ref{tab:vertexgated} in Appendix~\ref{apx:schemes} tells us that $c$ must denote the transformation $\curvename{a} = [1,2,-3]$, which induces a reflection. Let $\overline{x}$ denote the bitwise complement of a hexadecimal digit $x$.

First, observe that $P\curvename{a}mm.\curvename{cc}.ss.op.qr$ and $P\curvename{a}mm.\curvename{cc}.ss.\overline{p}\overline{o}.\overline{r}\overline{q}$ now encode curves that are reversed and reflected copies of each other, and should therefore get the same name. In such cases we will only use the lexicographically smallest name (interpreting digits by value, not by symbol).

Second, note that for symmetric curves, the two binary choices per octant do not give us four, but only two options per octant, because the curve equals its reversed reflection. Therefore we will not use reversal in the first half of the curve and not use non-reversal in the second half of the curve, and always give symmetric curves names of the form $P\curvename{a}mm.\curvename{cc}.ss.\curvename{II}.qq$ (recall that $\curvename{I} = 0000$). Alternatively, we may use the shorthand $P\curvename{a}mm.\curvename{c}sq$.

It remains to filter out the names that employ reversal to encode a symmetric curve. Consider a vertex-gated, edge-crossing curve that has a symmetric gate sequence and name $P\curvename{a}mm.\curvename{cc}.ss.op.qr$. Since transformation \curvename{a} induces a reflection, reversing the curve within an octant, that is, mapping $\tau(0)$ to $\gate i$ and $\tau(1)$ to $\gate {i-1}$, rather than the other way around, results in a reflection of the base pattern. Therefore, if we write \xor<x,y> to denote the digit that is obtained by applying the bitwise exclusive-or operation to the digits $x$ and $y$ in binary representation, then $P\curvename{a}mm.\curvename{cc}.ss.op.qr$ and $P\curvename{a}mm.\curvename{cc}.ss.\curvename{II}.\xor<o,q>\xor<p,r>$ are the same with respect to the order in which the second-level subcubes are traversed. If $\xor<o,q> = \xor<p,r>$, then that order is symmetric, so the whole curve is symmetric, and $P\curvename{a}mm.\curvename{cc}.ss.\curvename{II}.\xor<o,q>\xor<p,q>$ is its name. Conversely, $P\curvename{a}mm.\curvename{cc}.ss.op.qr$ names an asymmetric curve if and only if
$P\curvename{a}mm.\curvename{cc}.ss.op.qr$ is lexicographically smaller than $P\curvename{a}mm.\curvename{cc}.ss.\overline{p}\overline{o}.\overline{r}\overline{q}$ \emph{and} \xor<o,q> differs from \xor<p,r>.

\paragraph{Counting the curves}
As explained above, given the gate sequence, we need to make two binary choices per octant, which makes 16 choices for the full curve. Thus we get $2^{16} = 65\,536$ curves per gate sequence. However, if the gate sequence is symmetric, then the number of different curves is smaller. Symmetric curves have names ending in II followed by two identical digits: thus there are 16 such curves for each symmetric gate sequence. Asymmetric curves from symmetric gate sequences have 16 choices for each of the digits, except that we need to exclude the choice for the last digit that results in $\xor<o,q> = \xor<p,r>$, and we need to divide by two because we only keep the lexicographically smallest version out of $P\curvename{a}mm.\curvename{cc}.ss.op.qr$ and $P\curvename{a}mm.\curvename{cc}.ss.\overline{p}\overline{o}.\overline{r}\overline{q}$. Thus there are $16 \cdot 16 \cdot 16 \cdot 15 / 2 = 30\,720$ asymmetric curves from each symmetric gate sequence. In total we have 30\,736 curves for each of 10 symmetric gate sequences and 65\,536 different curves for each of 19 asymmetric gate sequences: 1\,552\,544 different edge-crossing vertex-gated curves in total.

\begin{figure}
\noindent
\raisebox{1ex}{(a)}\includegraphics[width=0.48\hsize,page=1]{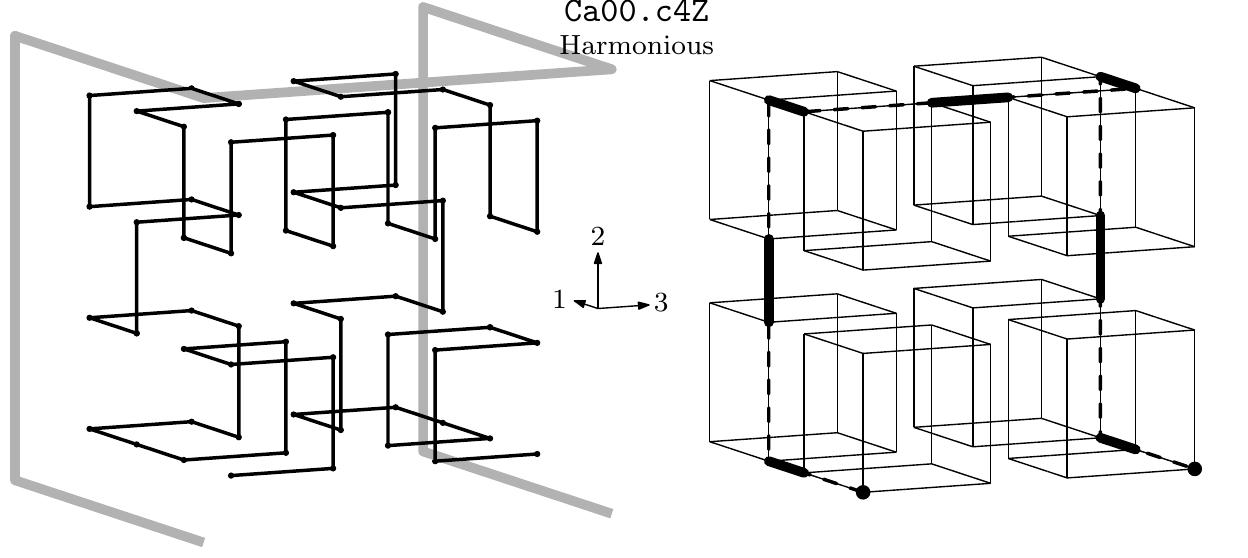}\hfill \raisebox{1ex}{(b)}\includegraphics[width=0.48\hsize,page=2]{curves.pdf}\\[1.8ex]
\raisebox{1ex}{(c)}\includegraphics[width=0.48\hsize,page=6]{curves.pdf}\hfill \raisebox{1ex}{(d)}\includegraphics[width=0.48\hsize,page=3]{curves.pdf}\\[1.8ex]
\raisebox{1ex}{(e)}\includegraphics[width=0.48\hsize,page=4]{curves.pdf}\hfill \raisebox{1ex}{(f)}\includegraphics[width=0.48\hsize,page=18]{curves.pdf}\\[1.8ex]
\raisebox{1ex}{(g)}\includegraphics[width=0.48\hsize,page=23]{curves.pdf}\hfill
\raisebox{1ex}{(h)}\includegraphics[width=0.48\hsize,page=5]{curves.pdf}
\caption{Examples of vertex-gated, edge-crossing curves.
(a)~The Harmonious Hilbert curve, with maximum facet-harmony.
(b)~Butz's coordinate-shifting curve.
(c)~Alfa, the vertex-gated hyperorthogonal curve, which has excellent locality-preserving properties, in particular with respect to bounding boxes and $L_\infty$-dilation (see Section~\ref{sec:observationslocality}).
(d)~The Sasburg curve, a metasymmetric curve with optimal score on the worst-case surface ratio.
(e)~A realization of another symmetric gate sequence for face-continuous curves. An asymmetric gate sequence \curvename{Ca00.cc.h4} for face-continuous curves can be constructed by combining the left half of the gate sequence in Figure e~with the right half of the sequence in Figure~b.
(f)~The Imposter curve, whose approximating curve $A_2$ wrongfully suggests maximum facet-harmony and palindromy.
(g)~A curve with full interior-diagonal harmony.
(h)~Another metasymmetric curve.}
\label{fig:edgecrossingcurves}
\end{figure}

\subsubsection{Curves that are vertex-gated and facet-crossing}\label{sec:facetcrossing}

\begin{figure}
\noindent
\raisebox{1ex}{(a)}\includegraphics[width=0.48\hsize,page=8]{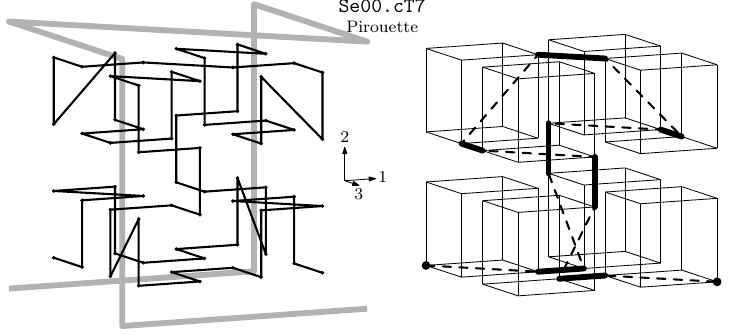}\hfill \raisebox{1ex}{(b)}\includegraphics[width=0.48\hsize,page=9]{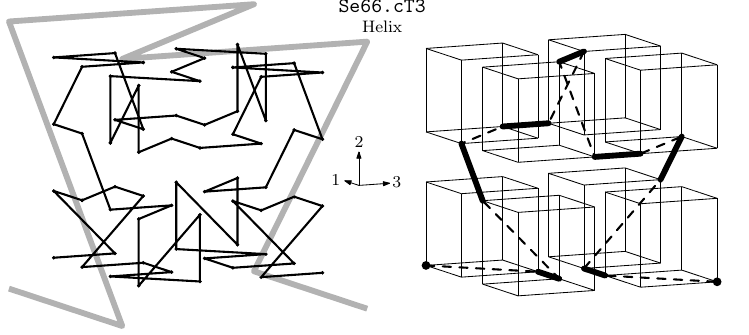}\\[2ex]
\raisebox{1ex}{(c)}\includegraphics[width=0.48\hsize,page=7]{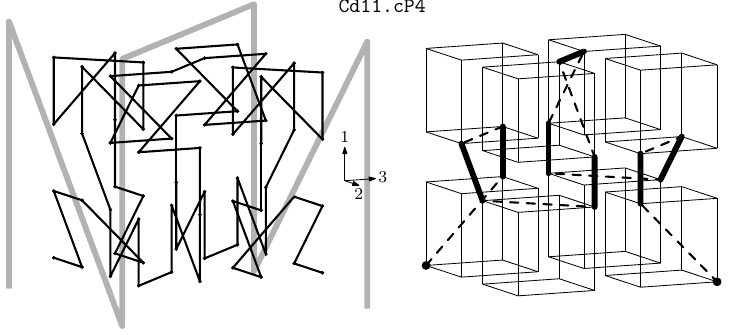}\hfill \raisebox{1ex}{(d)}\includegraphics[width=0.48\hsize,page=24]{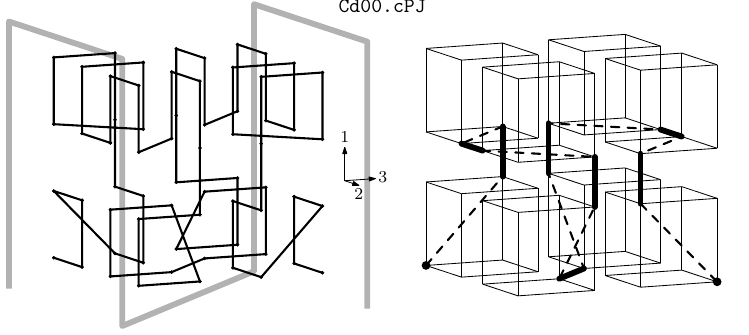}\\[2ex]
\raisebox{1ex}{(e)}\includegraphics[width=0.48\hsize,page=11]{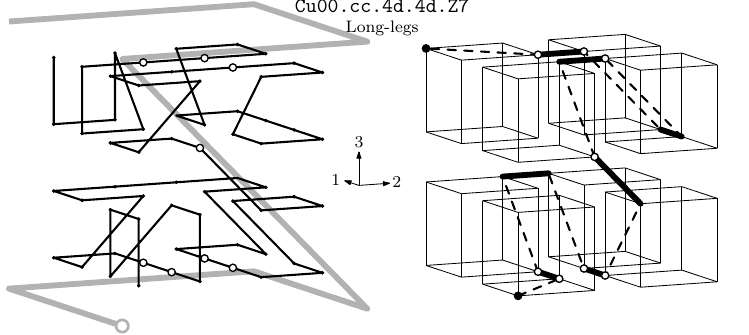}\hfill \raisebox{1ex}{(f)}\includegraphics[width=0.48\hsize,page=20]{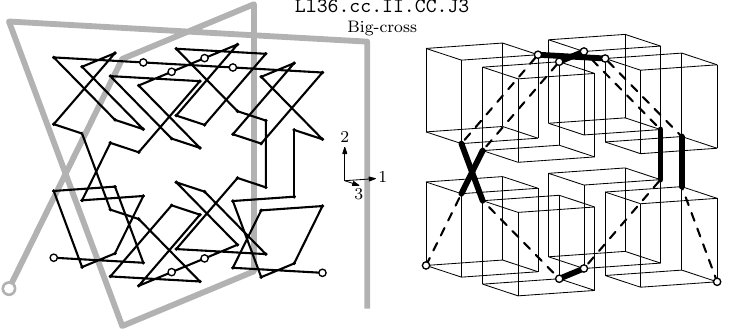}\\[2ex]
\raisebox{1ex}{(g)}\includegraphics[width=0.48\hsize,page=10]{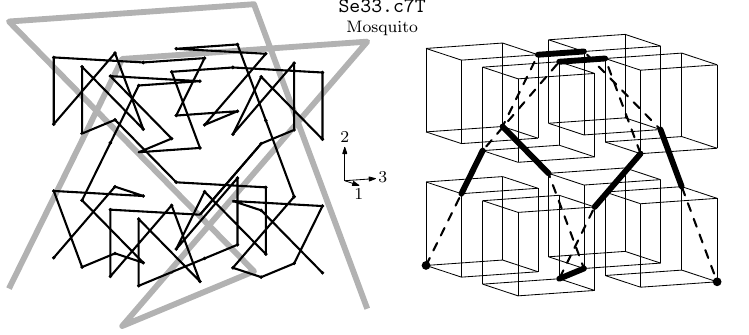}\hfill
\raisebox{1ex}{(h)}\includegraphics[width=0.48\hsize,page=21]{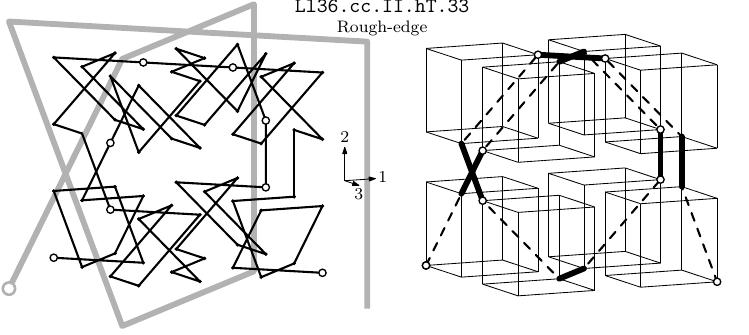}
\caption{Examples of vertex-gated, facet-crossing curves.
(a,b,c,d) Metasymmetric curves with full interior-diagonal harmony.
(e) An attempt to maximize the number of times four consecutive subcubes in the $4\times 4\times 4$ grid lie in a row.
(f) An attempt to maximize the number of times four consecutive subcubes in the $4\times 4\times 4$ grid lie on a diagonal.
(g) A crazy-looking, but symmetric, curve.
(h) An attempt to maximize the worst-case surface ratio (for example, the curve contains four consecutive octants that do not share any octant facets).}
\label{fig:facetcrossingcurves}
\end{figure}

\paragraph{Gate sequences}
The entrance gate $\gate 0$ must lie at a corner of the unit cube, so the coordinates of $\gate 0$ sum up to $\frac12 \pmod 1$. By induction, one can now show that all gates $\gate i$ for $0 \leq i \leq 8$ must have the same coordinate sum modulo 1, and thus, they must all be corners or facet midpoints of the unit cube. Moreover, since the gates $\gate i$ for $1 \leq i \leq 7$ must connect two octants, they cannot be corners of the unit cube, so each gate $\gate i$, for $1 \leq i \leq 7$, must be the midpoint of a facet of the unit cube that is shared by the $i$-th octant and the next.

Note that none of the base patterns from Table~\ref{tab:octantorders} start with a pair of octants that only differ in the third coordinate\footnote{The first halves of the \curvename{C}, \curvename{L}, \curvename{N}, and \curvename{X} patterns do not even contain \emph{any} pair of octants that only differ in the third coordinate. In the first half of an \curvename{S} pattern, the third and fourth octant in the default order differ only in the third coordinate, but no base pattern has a permutation of \curvename{S} that starts with those two octants. In the first half of a \curvename{Y} pattern, the second and fourth octant in the default order differ only in the third coordinate, but no base pattern has a permutation of \curvename{Y} that starts with those two octants.}. Therefore, for $\gate 1$, there is never a choice between the midpoint of a facet orthogonal to the first coordinate axis and the midpoint of a facet orthogonal to the second coordinate axis: if there is a choice, it is between the midpoint of a facet orthogonal to the third coordinate axis and the midpoint of a facet orthogonal to another coordinate axis. Therefore it suffices to encode whether $\gate 1$ lies on a facet that is orthogonal to the third coordinate axis. We implement this by letting the first bit encode the third coordinate of $\gate 1$ (0 if zero, 1 if non-zero).

For the remaining gates, note that $\gate i$ must be one of the two unit cube facet midpoints that are incident on the $i$-th octant and differ from $\gate {i-1}$. The three facet midpoints incident on the $i$-th octant can be distinguished by their non-zero coordinate. Thus we encode the location of $\gate i$ (for $2 \leq i \leq 4$) as follows. If the third coordinate of $\gate {i-1}$ is zero, then the $i$-th bit encodes the third coordinate of $\gate i$ (0 if zero, 1 if non-zero) to distinguish between the midpoint of a facet orthogonal to the third coordinate axis and the midpoint (not $g_{i-1}$) of a facet orthogonal to another coordinate axis. If the third coordinate of $\gate {i-1}$ is non-zero, then the $i$-th bit encodes the second coordinate of $\gate i$, to distinguish between the midpoint of a facet orthogonal to the first coordinate axis and the midpoint of a facet orthogonal to the second coordinate axis.

Gate sequences for vertex-gated, facet-crossing curves can easily be enumerated by exhaustive search: it turns out there are 149 such gate sequences, which realize 54 different base patterns. From the 149 gate sequences, 18 sequences are symmetric and 131 sequences are asymmetric. A full list is given in Appendix~\ref{apx:schemes}, Table~\ref{tab:vertexgated}.

\paragraph{Curves}
Given a gate sequence, vertex-gated, facet-crossing curves are specified in the same way as vertex-gated, edge-crossing curves. Note, however, that there is a subtle difference in the conditions for asymmetric curves from symmetric gate sequences. As one can see in Table~\ref{tab:vertexgated}, all symmetric, vertex-gated, facet-crossing gate sequences are symmetric by transformation \curvename{d} or \curvename{e} from Table~\ref{tab:transformations}, which is a rotation without reflection. In this respect the facet-crossing gate sequences differ from the edge-crossing gate sequences, where reversal of a symmetric gate sequence induced reflection of the base pattern. Consequently, names of asymmetric curves are a bit easier to recognize. If $Pcmm.\curvename{cc}.ss$ is a symmetric gate sequence for vertex-gated, facet-crossing curves, then $Pcmm.\curvename{cc}.ss.op.qr$ names an asymmetric curve if and only if $Pcmm.\curvename{cc}.ss.op.qr$ is lexicographically smaller than $Pcmm.\curvename{cc}.ss.\overline{p}\overline{o}.rq$ \emph{and} the last two digits are simply different.

Some examples of vertex-gated, facet-crossing curves are shown in Figure~\ref{fig:edgecrossingcurves}.

\paragraph{Counting the curves}
The curves can be counted as with vertex-gated, edge-crossing gate sequences. We have 30\,736 curves for each of 18 symmetric gate sequences and 65\,536 different curves for each of 131 asymmetric sequences: 9\,138\,464 different facet-crossing vertex-gated curves in total.

\subsection{Vertex-edge-gated curves}\label{sec:vertexedgegated}

We first establish where the gates lie relative to each other:

\begin{theorem}\label{thm:vertex-edge-gates}
If one gate of a three-dimensional Hilbert curve is at a vertex and the other gate lies in the interior of an edge, then the edge gate lies exactly in the middle of the edge.
Moreover, the two gates lie on a common facet of the unit cube, but not on a common edge of the unit cube.
\end{theorem}
\begin{proof}
Assume the entrance gate is at a vertex and the exit gate is in the interior of an edge (the reverse configuration is analogous). This implies that the traversal must be reversed in every second octant to be able to match the gates between the octants. Thus the traversal ends at a \emph{vertex} of the eighth octant that lies in the interior of an edge of the unit cube---thus it actually lies in the middle of that edge. We can now distinguish three possible locations for the exit gate.

(i) The exit gate lies on one of the three edges incident on the entrance gate. This is impossible by the same argument as for case (iii) in the proof of Theorem~\ref{thm:uniqueqss} (after traversing the first octant within a cube, it would be impossible to connect to the second octant).

\begin{figure}
\centering\includegraphics[page=2]{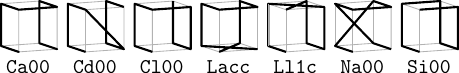}
\caption{Base patterns in which no octant is separated from its predecessor and its successor by a single axis-parallel plane (\curvename{Lacc} drawn upside down).}
\label{fig:wideturnpatterns}
\end{figure}

(ii) The exit gate lies on one of the three edges that are not on any of the unit cube facets that contain the entrance gate. This implies that one can never enter and leave an octant across the same axis-parallel centre plane of the cube. We will call this the \emph{no-turns condition}; it is equivalent to the condition that no pair of consecutive edges of $A_1$ makes an acute angle. The reader may now verify that the only base patterns that comply with the no-turns condition are \curvename{Ca00}, \curvename{Cd00}, \curvename{Cl00}, \curvename{Lacc}, \curvename{Ll1c}, \curvename{Na00}, and \curvename{Si00}\footnote{Half base patterns that comply are \curvename{C}*\curvename{0}*, \curvename{L}*\curvename{1}*, \curvename{L}*\curvename{c}*, \curvename{N}*\curvename{0}*, \curvename{N}*\curvename{5}*, and \curvename{S}*\curvename{0}*, but \curvename{N}*\curvename{5}* cannot be extended to a complete base pattern that complies with the no-turns condition.}, see Figure~\ref{fig:wideturnpatterns}.

Note that for any of these patterns to be realizable under the conditions of case (ii), the exit gate must be a vertex of the last octant that lies in the middle of an edge of the unit cube that is not on a common unit cube facet with the the first octant. In other words, if the first octant is the lower left front octant, then the exit gate must lie in the middle of the top right edge, the top back edge, or the right back edge. With \curvename{Ca00}, \curvename{Lacc}, \curvename{Ll1c}, \curvename{Na00} and \curvename{Si00}, this is not possible, because the first and the last octant are adjacent, and the last octant does not have any vertices on the top right, top back, or right back edge. With \curvename{Cd00}, the last octant has one such vertex (the middle of the top right edge), but it cannot be used as an exit gate since it lies on the octant facet that is shared by the second-last and the last octant and which must therefore contain the last octant's entrance gate (see Figure~\ref{fig:nocubecrossingvertexedgegated}a). Finally, the reader may verify that the only feasible gate sequence for the beginning of \curvename{Cl00} would put $\gate 2$ in the middle of the left back edge and $\gate 4$ in the middle of the front face. But thus, $\gate 4$, the fifth octants' entrance gate, lies on the octant facet that is shared by the fifth and the sixth octant and must, therefore, also contain the fifth octants' exit gate (see Figure~\ref{fig:nocubecrossingvertexedgegated}b). Thus, none of the base patterns that comply with the no-turns condition can be realized under the conditions of case (ii).

\begin{figure}
\centering
\includegraphics{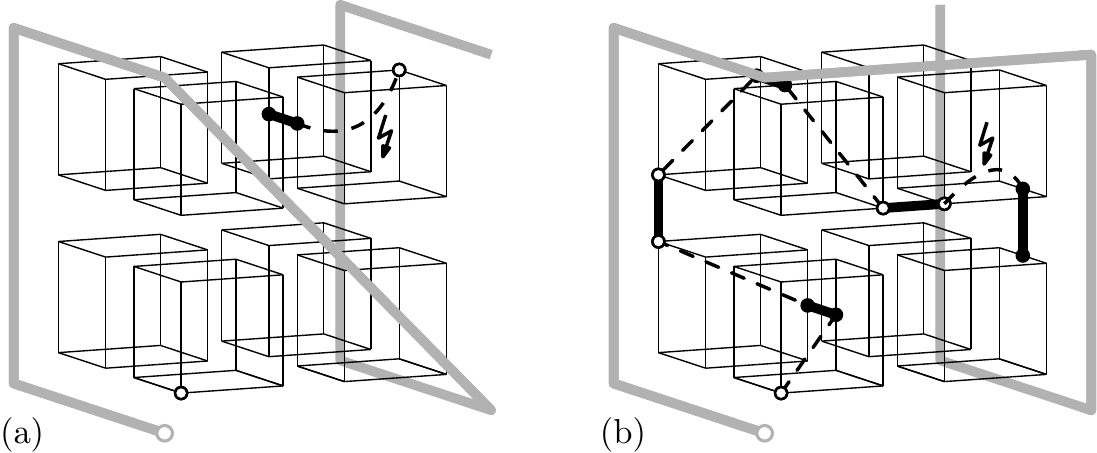}
\caption{(a) No cube-crossing vertex-edge-gated curve with base pattern \curvename{Cd00} could connect the last octant's entrance gate to the full curve's exit gate. (b) No cube-crossing vertex-edge-gated curve with base pattern \curvename{Cl00} could connect the fifth octant's entrance gate to the six octant's entrance gate.}
\label{fig:nocubecrossingvertexedgegated}
\end{figure}

(iii) The exit gate lies on a common unit cube facet, but not on a common unit cube edge with the entrance gate. This is the only possibility that remains.
\end{proof}

\paragraph{Gate sequences}
Theorem~\ref{thm:vertex-edge-gates} leaves only one possibility for the positions of the entrance gate and the exit gate relative to each other. Given only the location of the entrance gate $\gate {i-1}$ of an octant $C_i$, there are up to six different transformations of the complete curve that could map it to the curve within the octant, such that one of the gates $\gate 0$ or $\gate 8$ is mapped to $\gate {i-1}$. Correspondingly, there are up to six different possibilities for the location of the exit gate $g_i$ of $C_i$. However, the reader may verify that any particular facet of $C_i$ contains at most two of the possible locations for the exit gate $\gate i$, so at most two possible locations for $\gate i$ can be adjacent to the next octant. In fact, given the location of the entrance gate on an octant and the octant facet (if any) shared with the next octant, we only need to know whether the traversal within the octant is reflected to fully determine the location of the next gate.

Therefore, gate sequences are encoded with one bit per octant that simply indicates whether the traversal within the octant is reflected (1) or not (0). Note that thus, unfortunately, the encoding of the gate sequence in the first half of the curve is not independent from the second half: it depends on where the exit gate $\gate 8$ is located (on which edge), and it depends on the location of the fifth octant (to encode~$\gate 4$). Recall that reflections in the second half of the curve are subject to the transformation encoded by the second symbol of the base pattern name. Consider, for example, the curve \curvename{La13.cv.II} in Figure~\ref{fig:vertexedgegatedcurves}e: its reflections are encoded by $\curvename{II} = 0000\,0000$, which means no reflections. However, the second half of the curve is still subject to the reflection \curvename{a}, so that all octants in the second half of the pattern do in fact contain a reflected copy of the curve.

The combination of vertex and edge gates allows great flexibility in how to connect things up. An exhaustive search brought up a wide variety of gate sequences, 2\,758 in number, which realize 112 different base patterns. The gate sequences are listed in Appendix~\ref{apx:schemes}, Tables \ref{tab:vertexedgegated1} and~\ref{tab:vertexedgegated2}: for each entry in the table, the first column gives a prefix (third-level description), and the second and third column give a number of possibilities for the first and the second symbol, respectively, of the gate sequence specification. Each combination of a prefix, one symbol from the second column, and one symbol from the third column, constitutes a gate sequence name for a vertex-edge-gated curve.

\paragraph{Curves}
The location of the gates $\gate 0, ..., \gate 8$ leaves no freedom with respect to the rotations, reflections and/or reversals of the traversals within the octants. Therefore, there is only one curve per gate sequence, and a gate sequence name suffices to identity a curve. Some examples are shown in Figure~\ref{fig:vertexedgegatedcurves}. Note the examples of centred, standing curves: as we will prove in Section~\ref{sec:observationsorientation} (Corollary~\ref{cor:standing}) and Appendix~\ref{apx:verifyobservations} (Theorem~\ref{thm:CimpliesVE}), these properties cannot be obtained with vertex-gated curves.

\begin{figure}
\noindent
\hbox to\hsize{\raisebox{1ex}{(a)}\includegraphics[width=0.48\hsize,page=12]{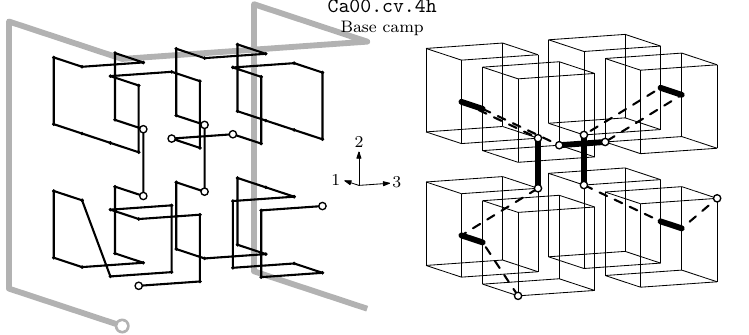}\hfill \raisebox{1ex}{(b)}\includegraphics[width=0.48\hsize,page=13]{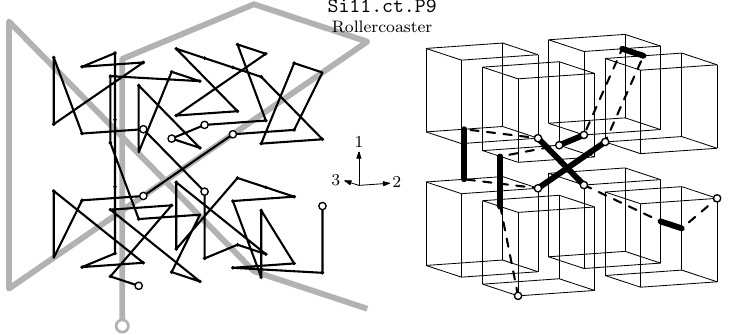}}\\[2ex]
\hbox to\hsize{\raisebox{1ex}{(c)}\includegraphics[width=0.48\hsize,page=17]{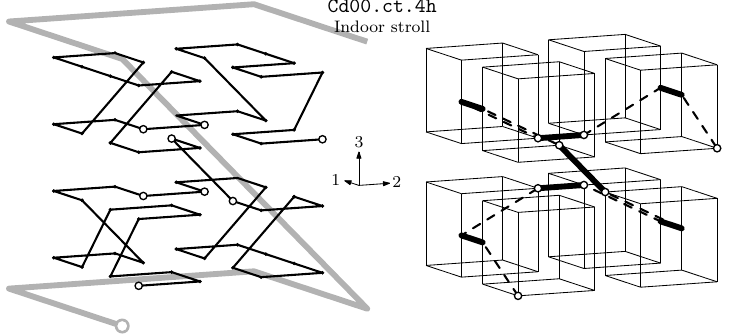}\hfill \raisebox{1ex}{(d)}\includegraphics[width=0.48\hsize,page=19]{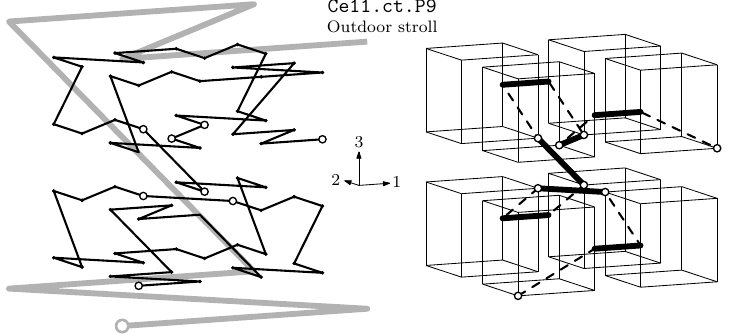}}\\[2ex]
\hbox to\hsize{\hfill \raisebox{1ex}{(e)}\includegraphics[width=0.48\hsize,page=22]{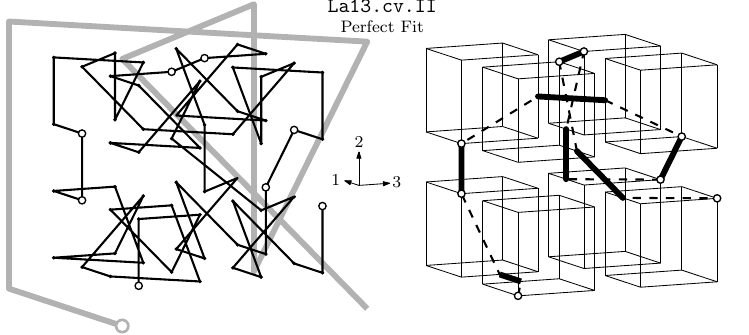}\hfill}
\caption{Examples of vertex-edge-gated curves.
(a) The centred, standing, well-folded curve that may have the most regular structure of any Hilbert curve.
(b) A crazy curve with many diagonal connections through the centre. It exhibits full harmony on the interior diagonals. Its worst-case locality-preserving properties are among the worst of all Hilbert curves.
(c,d) Two standing curves that would make for a relatively leisurely stroll if built as a three-dimensional labyrinth: there are no vertical edges, and a minimal number of sloped edges (only one per octant). These properties remain true in recursion.
(e) The only Hilbert curve that I found to be uniquely defined by its base pattern, \curvename{La13}: there is only one way to fit the octants together in this pattern.}
\label{fig:vertexedgegatedcurves}
\end{figure}

\subsection{Vertex-facet-gated curves}\label{sec:vertexfacetgated}

\begin{theorem}\label{thm:vertex-facet-gates}
If one gate of a three-dimensional Hilbert curve is at a vertex and the other lies in the interior of a facet, then the facet gate lies exactly in the middle of the facet. Moreover, the two gates do not lie on a common facet of the unit cube.
\end{theorem}
\begin{proof}
The proof is completely analogous to the beginning of proof Theorem~\ref{thm:vertex-edge-gates} (replacing edges by facets) up to and including case (i), leaving only the option of a facet gate in the middle of a facet that does not contain the entrance gate.
\end{proof}

\paragraph{Gate sequences}
Note that Theorem~\ref{thm:vertex-facet-gates} leaves only one possibility for the positions of the entrance and exit gates relative to each other.
As with the hypothetical vertex-edge-gated curves in case (ii) in the proof of Theorem~\ref{thm:vertex-edge-gates}, the no-turns condition applies. Therefore, only the base patterns \curvename{Ca00}, \curvename{Cd00}, \curvename{Cl00}, \curvename{Lacc}, \curvename{Ll1c}, \curvename{Na00}, and \curvename{Si00} need to be considered, from which \curvename{Ca00}, \curvename{Cd00}, \curvename{Na00}, and \curvename{Si00} can be eliminated right away because the last octant would have to have the entrance and the exit gate on the same octant facet. Traversing the octants one by one, starting from the vertex gate, forward and reverse copies of the curve must alternate so that they can connect at alternating vertex and facet gates. Thus, gates $\gate 0, \gate 2, \gate 4, \gate 6$ and $\gate 8$ will be at octant vertices and gates $\gate 1, \gate 3, \gate 5$ and $\gate 7$ will be in the centres of the octant facets that are shared by the octants they connect. This rules out the base patterns \curvename{Lacc} and \curvename{Ll1c}, because in these patterns, the seventh and the eight octant, who share $\gate 7$, do not share an octant facet.

This leaves only one possible base pattern, \curvename{Cl00}. The reader may now verify that only four gate sequences are possible, with gates $\gate 1, \gate 3, \gate 5$ and $\gate 7$ fixed as explained above, and $\gate 4$ in a fixed position on an edge of the unit cube. Only for $\gate 2$ and $\gate 6$ there is a choice: each of them can be either in the middle of an edge (\curvename{e}) or in the middle of a facet (\curvename{f}) of the unit cube. Thus we get the gate sequences \curvename{Cl00.cf.ee}, \curvename{Cl00.cf.ef}, \curvename{Cl00.cf.fe}, and \curvename{Cl00.cf.ff}.

\paragraph{Curves}
For each octant $C_i$, given the location of the gates $\gate {i-1}$ and $\gate i$, there is still freedom whether or not to reflect the traversal in the diagonal plane that contains both gates. This is encoded with one bit per octant: 0 if the transformation of the whole curve to the curve within the octant can be obtained without reflection; 1 if it requires a reflection. Thus, each gate sequence allows exactly $2^8 = 256$ different curves, and there are 1024 vertex-facet-gated curves in total. An example is shown in Figure~\ref{fig:boringcurves}a.

\begin{figure}
\noindent
\raisebox{1ex}{(a)}\includegraphics[width=0.48\hsize,page=14]{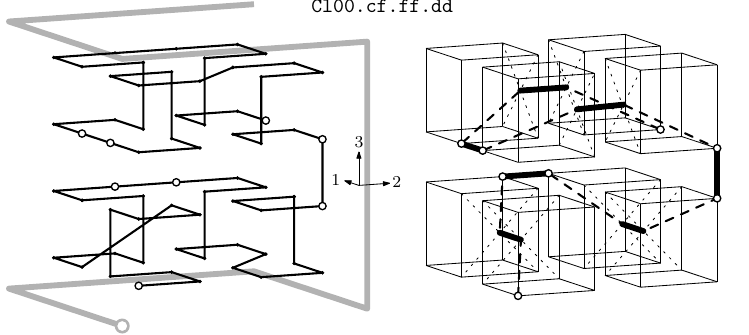}\hfill \raisebox{1ex}{(b)}\includegraphics[width=0.48\hsize,page=15]{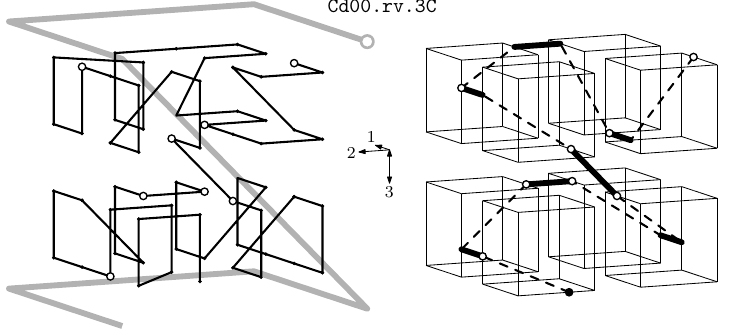}
\caption{Examples of: (a) a standing vertex-facet-gated curve. (b) an edge-gated curve.}
\label{fig:boringcurves}
\end{figure}

\subsection{Edge-gated curves}\label{sec:edgegated}

For completeness, I have to discuss the edge-gated curves. Unfortunately, whereas combining edge gates with vertex gates unlocked a world of curves with interesting properties that could not be achieved with vertex gates alone, having edge gates at both ends seems to be a severe restriction. The analysis of the possible locations of the gates in such configurations ultimately takes more space than the description of the curves themselves, which number only sixteen and in which I have not discovered any particularly interesting properties. The impatient reader is welcome to follow me to Section~\ref{sec:facetgated} on facet-gated curves immediately; the curious reader may consider consulting Section~\ref{sec:software} to learn about the software tools that can be used to explore the sixteen ugly ducklings to see if there is a swan in there after all.

We first determine where the gates lie on their respective edges in Lemma~\ref{lem:edge-edge-gates-distance} below. After that we determine where the entrance and the exit gate lie relative to each other in Theorem~\ref{thm:edge-edge-gates}.

\begin{lemma}\label{lem:edge-edge-gates-distance}
If both gates lie on the interiors of edges, then each gate lies at distance 1/3 to the closest vertex of the unit cube.
\end{lemma}
\begin{proof}
Let $a > 0$ and $z > 0$ be the distances of the entrance and the exit gate, respectively, to the closest vertex of the unit cube.
Let $a_0,...,a_8$ be the distance of $\gate 0,...,\gate 8$, respectively, to the closest octant vertex. If $a \neq z$, forward and reverse copies of the curve must alternate to match up, so we can distinguish two cases: (i) $a_0,a_2,a_4,a_6,a_8 = a/2$ and $a_1,a_3,a_5,a_7 = z/2$, or (ii) $a_0,a_2,a_4,a_6,a_8 = z/2$ and $a_1,a_3,a_5,a_7 = a/2$. If case (ii) applies, then case (i) applies to the reverse of the curve; if $a = z$, case (i) also applies; so we may assume that case (i) applies without loss of generality. In particular, we have $a_0 = a_8 = a/2$.

Since the curve does not have vertex gates, the gates $\gate 0$ and $\gate 8$ cannot lie at a vertex of the first or last octant, respectively, and therefore they cannot lie exactly in the middle of an edge of the unit cube. So we have $a < 1/2$ and $z < 1/2$. Now we can unambiguously define $u_0$ and $u_8$ as the vertices of the unit cube that are closest to $\gate 0$ and $\gate 8$, respectively, and we define $p_0,...,p_8$ as the octant vertices closest to $\gate 0,...,\gate 8$, respectively. Since $a_0 = a/2 \neq a$, the octant vertex $p_0$ cannot be $u_0$ but must be a midpoint of an edge of the unit cube; therefore the coordinates of $p_0$ sum up to 0, modulo~1. Between each pair of points $p_{i-1},p_i$, for $i \in \{1,...,8\}$, some coordinates may change by 1/2 and others remain equal. By the self-similarity of the curve, the number of coordinates that change is always the same and it is either even or odd, so between each pair of points $p_{i-1},p_i$, the sum of the coordinates changes either by 0 or by 1/2, modulo 1. Thus, summed over eight pairs $p_{i-1},p_i$, for $i \in \{1,...,8\}$, the coordinate sum changes by 0, modulo 1, and hence $p_8$, like $p_0$, must be a midpoint of an edge of the unit cube, not a vertex, that is, not $u_8$.

Now let $|xy|$ denote the distance between the points $x$ and $y$. We have $|u_0 p_0| = |u_0 \gate 0| + |\gate 0 p_0|$, so $1/2 = a + a_0 = a + a/2$ and, therefore, $a = 1/3$. Similarly, we have $|u_8 p_8| = |u_8 \gate 8| + |\gate 8 p_8|$, so $1/2 = z + a_0 = z + a/2 = z + 1/6$, and therefore, $z = 1/3 = a$.
\end{proof}

\begin{theorem}\label{thm:edge-edge-gates}
If both gates of a three-dimensional Hilbert curve lie on the interiors of edges, then the octants that contain the gates lie on a common facet of the unit cube, but not on a common edge of the unit cube; one of the gates lies on that facet shared by the first and the last octant, whereas the other gate lies on an edge orthogonal to the shared facet. Each gate lies at distance 1/3 to the closest vertex of the unit cube.
\end{theorem}
\begin{proof}
First note that the entrance and the exit gates cannot lie on parallel edges, since then, by induction, all gates $g_0,...,g_8$ would have to lie on parallel edges, and the traversal would never be able to cross the centre plane of the unit cube that is orthogonal to those edges. We now distinguish three cases for the possible locations of the first and the last octant relative to each other: they may (i) share an edge of the unit cube, or (ii) only a facet, or (iii) they lie opposite of each other on an interior diagonal.

\begin{figure}
\centering
\includegraphics[width=\hsize]{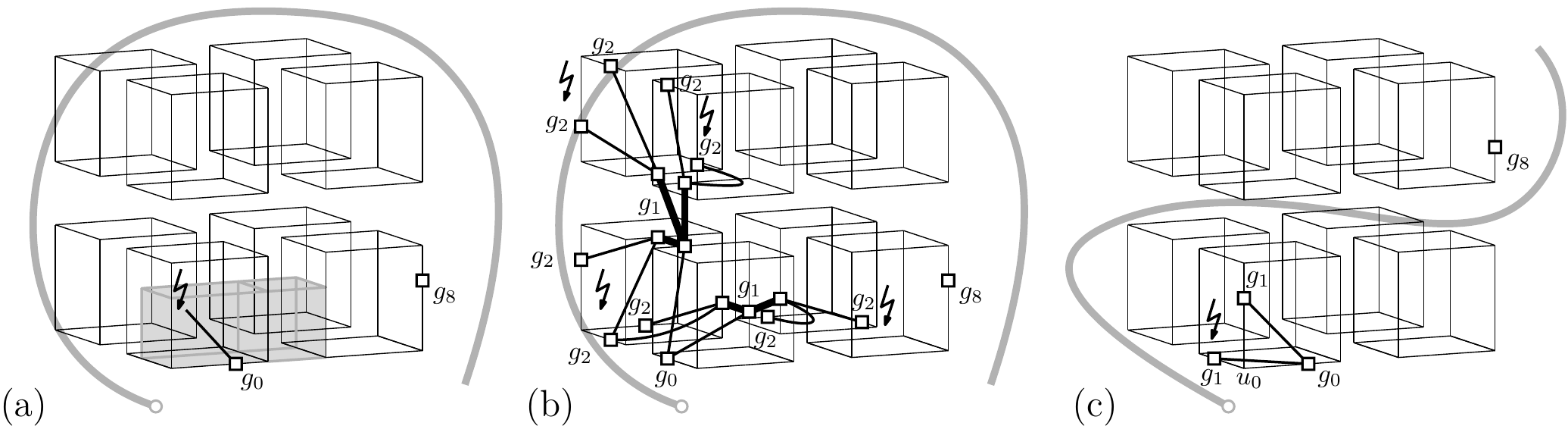}
\caption{Impossible gate combinations for edge-gated curves. (a) Case (i)(a): if $g_0$ lies on the unit cube edge shared by the first and the last octant, then $g_1$ must lie at $L_\infty$-distance at most $1/3$ from $g_0$, that is, within the shaded box. But that box does not reach any octant other than the first and the last. (b) Case (i)(b): showing the possible locations for $g_1$ and $g_2$. None of the possible locations for $g_2$ allows a connection to a third octant. (c) Case (ii)(a): showing the possible locations for $g_1$. None of the possible locations for $g_1$ allows a connection to a second octant.}
\label{fig:impossible-edge-gated-curves}
\end{figure}

Case (i): the first and last octant lie on the same edge of the unit cube. We distinguish two subcases: (a) the entrance or the exit gate lies on the shared edge of the unit cube, and (b) neither the entrance nor the exit gate lies on the shared edge. In case (a), suppose the entrance gate lies on the shared edge (the case of an exit gate on the shared edge is symmetric); see Figure~\ref{fig:impossible-edge-gated-curves}a. The $L_\infty$-distance between $g_0$ and $g_8$ would then be at most 2/3, and thus, the $L_\infty$-distance between $g_0$ and $g_1$ would have to be at most 1/3. But the $L_\infty$-distance between $g_0$ and the second octant is 1/2, so this is not possible. In case (b), illustrated by Figure~\ref{fig:impossible-edge-gated-curves}b, recall that the entrance and exit gate do not lie on parallel edges. A traversal within an octant that starts on an edge of the unit cube, at distance 1/3 from the corner of the unit cube, can now only end on an axis-parallel line through the centre of the unit cube at distance 1/3 from the centre, and vice versa. Thus $\gate 2$, like $\gate 0$, must lie on an edge of the unit cube, where no connection to the third octant is possible. So case (i) cannot occur.

Case (ii): the first octant and the last octant lie on the same facet of the unit cube. Since the entrance and exit gates cannot lie on parallel edges, at most one of them lies on an edge orthogonal to the unit cube facet shared by the first and the last octant. We distinguish two subcases: (a) both the entrance and the exit gate lie on the shared facet of the unit cube, and (b) only one of the gates lies on the shared facet. In case (a), following the curve through the first octant, it follows that $g_1$ lies on an edge of the unit cube, at distance 1/3 from $u_0$ (see Figure~\ref{fig:impossible-edge-gated-curves}c). So no connection to the second octant is possible, and case (a) cannot occur. Case (b) is entirely realizable, as we will see below and illustrate in Figure~\ref{fig:boringcurves}b.

Case (iii): the first and the last octant lie opposite of each other on an interior diagonal of the unit cube. In this case, consider the four layers of second-level subcubes, and without loss of generality, suppose the traversal starts in the bottom layer. With every traversal of an octant in the bottom half of the unit cube, one moves from the bottom layer to the second layer or vice versa, so after traversing all octants in the bottom half (possibly visiting octants in the top half in between), one ends in the bottom layer with no more octants to go to. However, the traversal must end in an octant that spans the top two layers, so this case is not realizable.
\end{proof}

\paragraph{Gate sequences}
Theorem \ref{thm:edge-edge-gates} leaves only one possibility for the positions of the gates $g_0$ and $g_8$ relative to each other. Furthermore, by the same arguments as in the proof of Theorem~\ref{thm:vertex-edge-gates}, case (ii), the fact that the gates do not lie on a common facet restricts the possible base patterns to \curvename{Ca00}, \curvename{Cd00}, \curvename{Cl00}, \curvename{Lacc}, \curvename{Ll1c}, \curvename{Na00}, and \curvename{Si00}. From these, only \curvename{Cd00} has the first and the last octant on a common unit cube facet, but no a common unit cube edge, as required by Theorem~\ref{thm:edge-edge-gates}. So all edge-gated curves have base pattern \curvename{Cd00}. Moreover, by Theorem~\ref{thm:edge-edge-gates}, one of the gates is on an edge orthogonal to the facet of the unit cube that is shared by the first and the last octant, so the names of all edge-gated curves start with either \curvename{Cd00.rt} or \curvename{Cd00.rv}.

\begin{figure}
\centering
\includegraphics[width=\hsize]{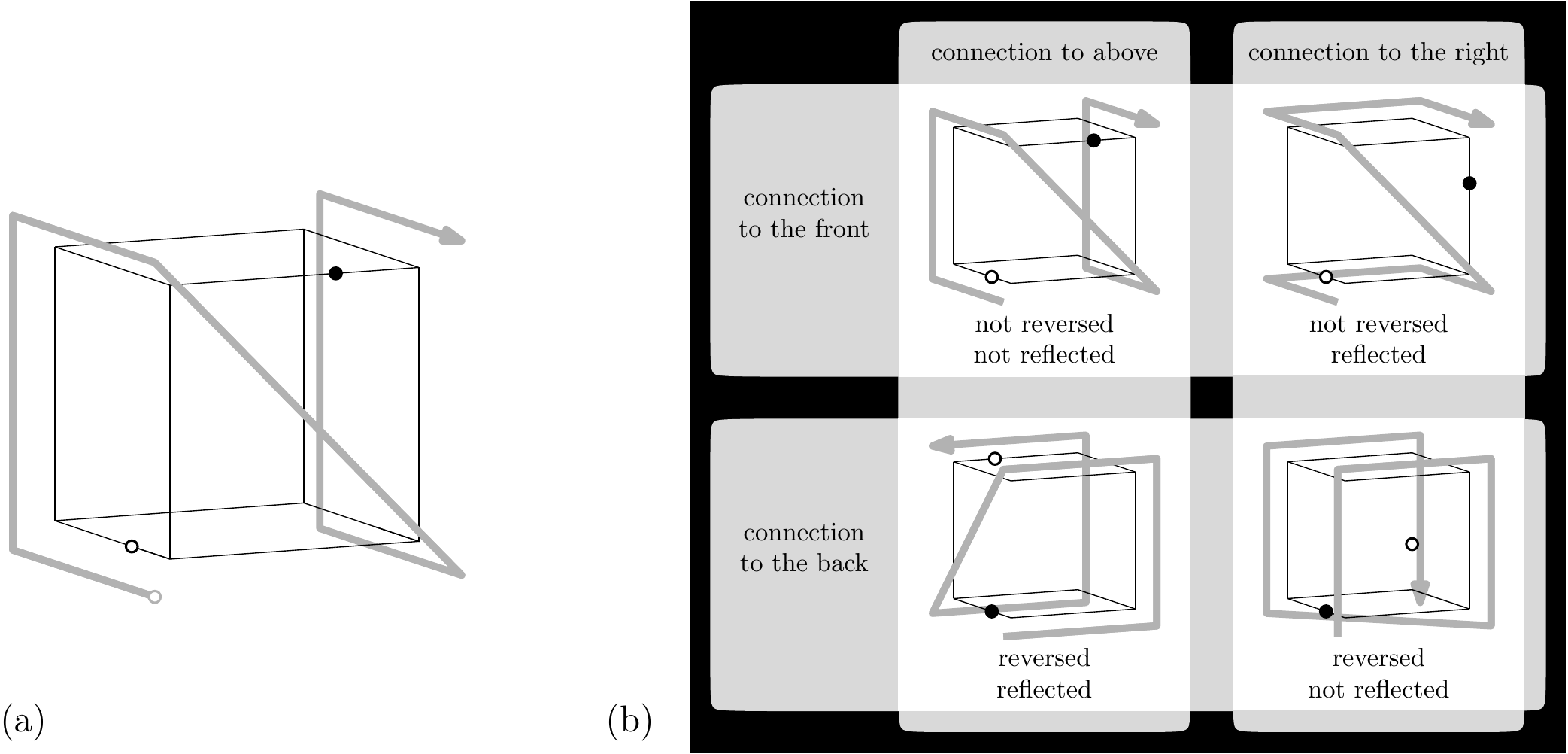}
\caption{(a) The location of the entrance and exit gates of a \curvename{Cd00.rt}-curve. (b) Example of how reflection in a given octant distinguishes between the two possible ways to reach the next octant. Assume (without loss of generality, modulo rotary reflections) that the entrance gate of the given octant is on the bottom left edge at 1/3 of the way from the front vertex to the back vertex. Then there are four possible locations for the exit gate, as shown in the four figures. On each of the top facet, right facet, front facet and back facet of the octant, we find exactly two candidate locations for the exit gate: one that is reached by a (possibly reversed) non-reflected traversal and one that is reached by a (possibly reversed) reflected traversal.}
\label{fig:edgegates}
\end{figure}

Given the location of the entrance gate of an octant and the octant facet (if any) that is shared with the next octant, we only need to know whether the traversal within the octant is reflected to fully determine the location of the next gate (Figure~\ref{fig:edgegates} illustrates this for \curvename{Cd00.rt}-curves; the situation for \curvename{Cd00.rv}-curves is similar). Therefore, gate sequences are encoded with one bit per octant that simply indicates whether the traversal within the octant is reflected (1) or not (0).

An exhaustive search brought up 16 gate sequences, as listed in Appendix~\ref{apx:schemes}, Table~\ref{tab:edgegated}.

\paragraph{Curves}
The location of the gates $\gate 0, ..., \gate 8$ leaves no freedom with respect to the rotations, reflections and/or reversals of the traversals within the octants. Therefore, there is only one curve per gate sequence, and a gate sequence name suffices to identity a curve. An example is shown in Figure~\ref{fig:boringcurves}b.

\subsection{Edge-facet-gated curves}\label{sec:edgefacetgated}

\begin{theorem}\label{thm:edge-facet-gates}
There is no three-dimensional Hilbert curve with one gate in the interior of an edge and one gate in the interior of a facet.
\end{theorem}
\begin{proof}
Suppose there would be a three-dimensional Hilbert curve with the entrance gate in the interior of an edge and the exit gate in the interior of a facet. This implies that the traversal must be reversed in every second octant to be able to match the gates between the octants. Thus the traversal ends on an \emph{edge} of the eighth octant that lies in the interior of a facet of the unit cube---thus the exit gate lies on an axis-parallel centreline of that facet, but not in the centre point of the facet. Call the facet that contains the exit gate the back face.

Then the first and the second octant must be connected back to back and they must have the same transformations, apart from a reflection in a plane parallel to their shared octant facet and possibly a reflection in a plane containing the back face centreline with the connecting gate. Thus, $\gate 0$ and $\gate 2$ must lie on parallel edges. Following the traversal through the subsequent octants, we find that $\gate 0$, $\gate 2$, $\gate 4$, $\gate 6$ and $\gate 8$ are all on parallel edges. Hence $\gate 0$, the entrance gate of the unit cube, is on an edge that is parallel with the facet that contains the exit gate.

Now rotate the curve such that the edge with the entrance gate is vertical. We now find that all gates $\gate 0,...,\gate 8$ lie in the interior of vertical octant edges or facets. But then the curve cannot connect the bottom half of the unit cube with the top half of the unit cube. Therefore it is impossible to realize a gate sequence with one gate in the interior of an edge and the other gate in the interior of a face.
\end{proof}

\subsection{Facet-gated curves}\label{sec:facetgated}

\begin{theorem}\label{thm:facet-facet-gates}
There is only one facet-gated three-dimensional Hilbert curve.
\end{theorem}
\begin{proof}
In a previous manuscript~\cite{inventory}, I gave a proof similar to that of Theorem~\ref{thm:edge-edge-gates}, but considerably more complicated to take the seemingly greater number of degrees of freedom in the locations of the gates into account. In contrast, here I will give a mostly independent and much shorter proof, building on the recent results from my work with Bos~\cite{hyperorthogonal}.

First observe that any facet-gated curve must be face-continuous, since consecutive subcubes in a grid must always share a subcube facet for their exit and entrance gate to match up. This limits the possible base patterns to those whose approximating curves $A_1$ only have axis-parallel edges: \curvename{Ca00}, \curvename{Cl00}, and \curvename{Si00} (see Figure~\ref{fig:wideturnpatterns}).

Moreover, the gates cannot be on parallel facets, otherwise, by induction, all of the gates $\gate 0, \gate 1, ...$ must be on parallel octant facets and the traversal could never cross the centre planes of the unit cube that are orthogonal to the facets that contain the gates. Hence, the gates are on non-parallel facets, and thus, no pair of consecutive edges in an approximating curve can be collinear. It follows that, in three dimensions, facet-gated curves are not only face-continuous, but even hyperorthogonal.

We now consider the three possible base patterns one by one.

For base pattern \curvename{Ca00}, which identifies well-folded curves, there is exactly one facet-gated hyperorthogonal curve, as proven by Bos and myself~\cite{hyperorthogonal}.

Base pattern \curvename{Cl00} cannot be realized with face-continuous curves, by the same arguments as for case (iii) of the proof of Theorem~\ref{thm:edge-edge-gates}.

\begin{figure}
\centering
\includegraphics[width=\hsize]{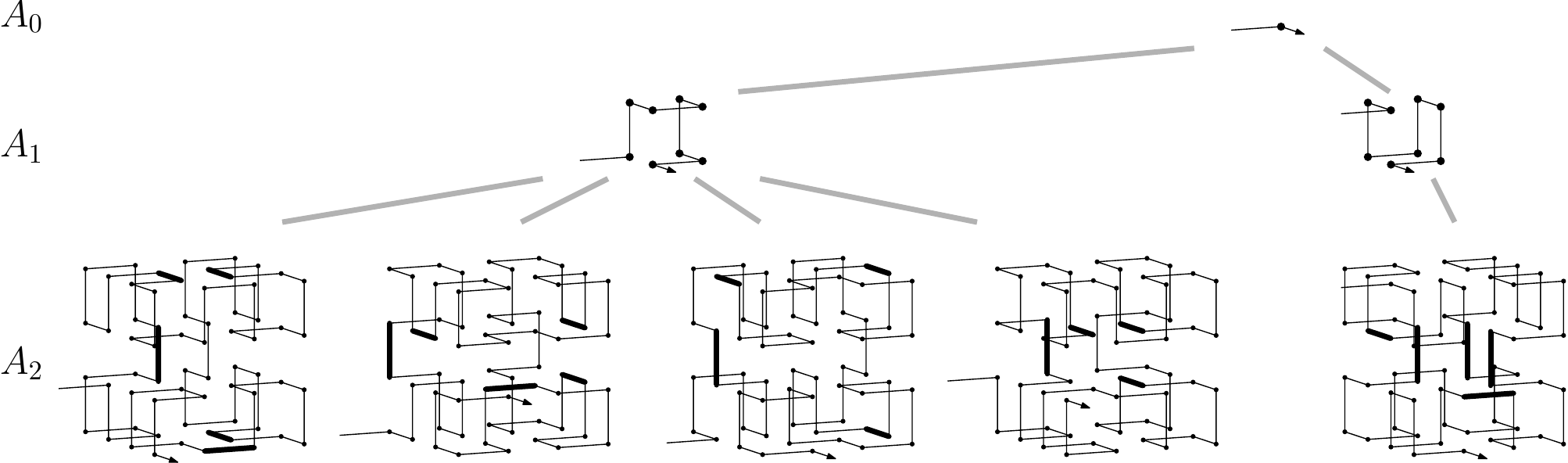}
\caption{Trying and failing to construct a facet-gated curve with base pattern \curvename{Si00}.}
\label{fig:impossiblefacetgatedSi00}
\end{figure}

We investigate base pattern \curvename{Si00} by trying to draw the approximating curves $A_0, A_1, A_2, A_3$, each extended with an entry edge and an exit edge orthogonal to the facets that contain the gates. Approximating curve $A_0$, consisting of an entry edge, one vertex, and an exit edge, is trivial. There are two different ways (modulo rotation, reflection and/or reversal) to construct a matching approximating curve $A_1$, using base pattern \curvename{Si00}, and without collinear edges, as shown in Figure~\ref{fig:impossiblefacetgatedSi00}. From these, we can construct five different curves $A_2$ while maintaining that the curve in each octant is similar to $A_1$. However, now one can see that we will not be able to construct matching curves $A_3$ while maintaining continuity: if we replace the curve in each octant by a copy of $A_2$, then the edges drawn fat in Figure~\ref{fig:impossiblefacetgatedSi00} will break. Hence, there are no facet-gated curves with base pattern \curvename{Si00}.

If follows that the only facet-gated three-dimensional Hilbert curve is the facet-gated three-dimensional hyperorthogonal well-folded curve which was described by Bos and myself~\cite{hyperorthogonal}. Bos and I also calculated the locations of the gates and thus we get:
\end{proof}

\begin{corollary}\label{cor:facet-facet-gates-distance}
If both gates lie on the interiors of facets, then each gate lies at distance 1/3 to the closest two edges of the facet of the unit cube that contains it.
\end{corollary}

The only facet-gated curve, as known from my work with Bos~\cite{hyperorthogonal}, is shown in Figure~\ref{fig:facetgatedcurve}. Its name is simply \curvename{Ca00.gs}. Since there is only one facet-gated curve, its name does not need to be more specific than this: there is no need to set up an encoding for different gate sequences.

As we will see in Section~\ref{sec:observationslocality}, the facet-gated hyperorthogonal curve has excellent locality-preserving properties: it is the unique best three-dimensional Hilbert curve with respect to the worst-case bounding-box surface and $L_2$-dilation measures. On each of the other metrics calculated, the curve is within 4\% from optimal.

\begin{figure}
\centering\noindent
\includegraphics[width=0.48\hsize,page=16]{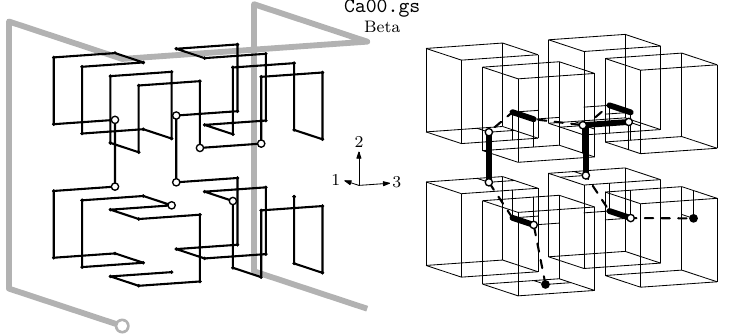}
\caption{The facet-gated curve.}
\label{fig:facetgatedcurve}
\end{figure}

\section{Observations on properties of three-dimensional Hilbert curves}\label{sec:observations}

In this section I will present some observations on the properties of the three-dimensional Hilbert curves that we can generate by enumerating their names according to the principles laid out in the previous sections. In Sections \ref{sec:observationsoptionalproperties} and~\ref{sec:observationsorientation}, we will see a number of results that allow us, or our software tool (see Section~\ref{sec:software}), to automatically recognize face-continuous, hyperorthogonal, maximally facet-harmonious, fully interior-diagonal-harmonious, pattern-isotropic or edge-isotropic curves. In Section~\ref{sec:observationslocality}, I present some results of calculations of locality-preserving properties of the generated curves.

Other striking observations about the generated curves include the following. No three-dimensional Hilbert curve follows partition~\curvename{X}. All symmetric curves are vertex-gated curves whose names start with \curvename{Ca}, \curvename{Cd}, \curvename{Ce}, \curvename{La}, \curvename{Ne} or \curvename{Se}. All centred curves are vertex-edge-gated. In fact, for these observations, one can give fairly compact proofs that do not depend on enumerating all curves and testing them. These proofs can be found in Appendix~\ref{apx:verifyobservations}, but they are without further consequence for the work presented in this article.

\subsection{General optional properties}\label{sec:observationsoptionalproperties}

Violations of face-continuity, hyperorthogonality, harmony and palindromy are easy to recognize, but it may not be straightforward to verify that a curve is entirely free of violations and therefore has the aforementioned properties. In this subsection we will discuss easy solutions for these properties and explain why the easy solutions suffice. At the end of this section, we discuss metasymmetry.

\begin{theorem}\label{thm:allfacecontinuous}
The face-continuous three-dimensional Hilbert curves are those whose names start with \curvename{Ca00.cc}, \curvename{Ca00.gs}, or \curvename{Si00.cc}.
\end{theorem}
\begin{proof}
A curve can only be face-continuous if each octant except the first shares an octant facet with the previous octant. There are only three base patterns that satisfy this requirement: \curvename{Ca00}, \curvename{Cl00}, and \curvename{Si00} (see Figure~\ref{fig:wideturnpatterns}).

With base patterns \curvename{Ca00} and \curvename{Si00}, the first and the last octant lie next to each other along an edge of the unit cube. By the theorems from Section~\ref{sec:inventory}, traversals for which this is the case may only be realized as vertex-gated, as vertex-edge-gated, or as facet-gated curves. We consider these case one by one.

Vertex-gated curves: note that if octants $C_i$ and $C_{i+1}$ share an octant facet $F$ and connect in a vertex gate $\gate i$, then the suboctants of $C_i$ and $C_{i+1}$, respectively, that connect in $\gate i$ must also share a suboctant facet, namely the appropriately sized subfacet of $F$ in the corner at $\gate i$. Hence, by induction, all vertex-gated curves whose names start with \curvename{Ca00.cc} or \curvename{Si00.cc} are face-continuous.

Vertex-edge-gated curves: consider the curve $A_2$ that sketches the traversal of the second-level subcubes. Let $v_{ijk}$ be the vertex of $A_2$ in the $i$-th layer, $j$-th row, $k$-th column, where $i, j, k \in \{1,...,4\}$. We define the \emph{parity} of a vertex $v_{ijk}$ as the parity of $i + j + k$. Note that $A_2$ has an odd number of edges, and, if the curve is face-continuous, each edge connects vertices of different parity. Thus, the first and the last vertex of $A_2$ must have different parity. From Theorem~\ref{thm:vertex-edge-gates} we know that, without loss of generality, we may number layers, rows and columns such that $A_2$ starts at $v_{111}$ and ends at $v_{124}$ or $v_{134}$. Given that the base pattern \curvename{Ca00} or \curvename{Si00} ends with on octant that lies next to the first octant along an edge of the unit cube, $A_2$ actually has to end at $v_{124}$, but $v_{124}$ has the same parity as $v_{111}$. This contradicts the conditions of a face-continuous traversal. Hence, no vertex-edge-gated curve is face-continuous.

Facet-gated curves: these are necessarily face-continuous, because consecutive octants must always share the octant facet that contains the gate between them. In particular, the only facet-gated curve, \curvename{Ca00.gs}, is face-continuous.

With base pattern \curvename{Cl00}, the first and the last octant are opposite of each other on an interior diagonal of the unit cube. By the theorems from Section~\ref{sec:inventory}, this can only be realized by vertex-facet-gated curves. However, starting from the vertex gate, we find that a face-continuous traversal is not possible by the same argument as in case (iii) of the proof of Theorem~\ref{thm:edge-edge-gates}.
\end{proof}

\begin{finding}\label{fnd:hyperorthogonal}
The only hyperorthogonal three-dimensional Hilbert curves are \curvename{Ca00.cc.44.hh.db} and \curvename{Ca00.gs}.
\end{finding}
\begin{howfound}
\curvename{Ca00.cc.44.hh.db} and \curvename{Ca00.gs} were proven to be hyperorthogonal in earlier work~\cite{hyperorthogonal}. All other face-continuous three-dimensional Hilbert curves were found to have a pair of consecutive collinear edges in the third-order approximating curve $A_3$, thus violating the conditions of hyperorthogonality.
\end{howfound}

\begin{finding}\label{fnd:harmonious}
The only three-dimensional Hilbert curve with maximum facet-harmony is \curvename{Ca00.c4Z}.
\end{finding}
\begin{howfound}
One can easily verify by induction that \curvename{Ca00.c4Z} (see Figure~\ref{fig:edgecrossingcurves}a) is consistent with the two-dimensional Hilbert curve on all facets except the back facet. The crucial observation to use is the following: the transformations that map the curve as a whole to the curves within the octants are such that the back facet of the unit cube is mapped to octant facets that lie either in the interior of the unit cube, or on the back. Thus, the violations of two-dimensional Hilbert order that show up on the back facets, do not show up on any of the other facets in recursion. I found that for all other three-dimensional Hilbert curves, inconsistency with the two-dimensional Hilbert curve can be established for at least two facets by inspecting the approximating curve $A_3$. (In fact, for all curves other than the ``Imposter'' \curvename{Ca00.cT7}, such violations are already visible in $A_2$).
\end{howfound}

The following lemma is straightforward to prove:
\begin{lemma}\label{lem:decidediagonalharmony}
A curve has full interior-diagonal harmony if and only if on each interior diagonal, the second-level subcubes are visited in order.
\end{lemma}


\begin{finding}
There is no palindromic three-dimensional Hilbert curve.
\end{finding}
\begin{howfound}
All three-dimensional Hilbert curves were found to show violations of palindromic conditions in $A_3$.
(In fact, for all curves other than the ``Imposter'' \curvename{Ca00.cT7} and \curvename{Ca00.cT9}, such violations are already visible in $A_2$.)
\end{howfound}

Let $[X]$ denote any character from the string $X$.

\begin{finding}\label{fnd:metasymmetry}
The only metasymmetric three-dimensional Hilbert curves are the curves
$\curvename{Ca00.cT}[\curvename{bC47}]$, $\curvename{Ca11.cT}[\curvename{IPJ3}]$, $\curvename{Cd00.cP}[\curvename{IPJ3}]$, $\curvename{Cd11.cP}[\curvename{bC47}]$, $\curvename{Se00.cT}[\curvename{bPJ7}]$, and $\curvename{Se66.cT}[\curvename{IC43}]$.
\end{finding}
Metasymmetry is an unpleasant property in the following sense. On the one hand, I do not have an easy conclusive argument, short enough to present here, why the 24 curves of Finding~\ref{fnd:metasymmetry} are indeed the only metasymmetric curves. On the other hand, the amount of work involved in tediously verifying the finding by hand is still small enough that it does not warrant the effort involved in implementing an automatic check. For the reader who wishes to verify Finding~\ref{fnd:metasymmetry} by hand, some hints are given in Appendix~\ref{apx:verifyobservations}.

\subsection{Orientation properties}\label{sec:observationsorientation}

\paragraph{Pattern-isotropy}
To be able to analyse isotropy and other orientation properties of the curves, we introduce the following notation.

For a given three-dimensional Hilbert curve $\tau$, let $\Gamma^k(\tau)$ be the set of symmetries of the unit cube that map $\tau$ to the curve (or its reverse) within at least one $k$-th-level subcube (modulo scaling, translation, and reversal). So $\gamma \in \Gamma^1(\tau)$ if and only if there is a first-level octant $C_i$ with $\gamma = \gamma_i$ or $\gamma = \gamma_i \circ \sigma$, where, if applicable, $\sigma$ is a transformation that maps a symmetric curve $\tau$ to its own reverse. Now, for any integer $k \geq 2$, we have $\gamma \in \Gamma^k(\tau)$ if and only if there are $\alpha \in \Gamma^1(\tau)$ and $\beta \in \Gamma^{k-1}(\tau)$ such that $\gamma = \alpha \circ \beta$.

\begin{lemma}\label{lem:isotropy}
A three-dimensional Hilbert curve $\tau$ is pattern-isotropic if and only if there is a $k$ such that $\Gamma^k(\tau)$ is the set of all 48 symmetries of the unit cube.
\end{lemma}
\begin{proof}
For ease of explanation, first consider asymmetric curves $\tau$.

By definition, an asymmetric three-dimensional traversal is pattern-isotropic if and only if, in the limit, we see each of the 48 possible transformations of the base pattern equally often.
Let $\pi_1,...,\pi_{48}$ be the symmetries of the unit cube, numbered such that $\pi_1$ is the trivial symmetry (the identity transformation). Let $P$ be the $48 \times 48$ matrix defined by $P_{ij} = z/8$ if exactly $z$ out of the $8$ first-level subcubes have the transformation $\beta$ such that $\beta \circ \pi_j = \pi_i$. Note that for each permutation $\beta$ used in $z$ first-level subcubes, we put an entry with value $z/8$ in each row $i$ (namely at $P_{ij}$, where $\pi_j = \beta^{-1} \circ \pi_i$) and in each column $i$ (namely at $P_{hi}$ where $\pi_h = \beta \circ \pi_i$). Thus each row sums up to one and each column sums up to one. Let $u$ be the 48 elements' vector with $u[1] = 1$ and $u[i] = 0$ for all $i > 1$.
Now, if we choose a subcube from level $m$ at random, the probability that is has transformation $\pi_i$ is given by element $i$ of the vector $P^m u$.

Suppose there is a $k$ such that $\Gamma^k(\tau)$ is complete, that is, it contains all 48 symmetries of the unit cube. Then all entries in the first column of $P^k$ are strictly positive. These entries indicate that any transformation $\pi_i$ can be constructed from the identity transformation $\pi_1$ by composing $k$ transformations from those from the eight octants with $\pi_1$. Then any permutation $\pi_i$ can actually be constructed from any permutation $\pi_j$ in this way, and therefore all entries of $P^k$ are strictly positive. This implies that $P$ is a regular doubly stochastic matrix, and as $m$ goes to infinity, $P^m u$ converges to the vector of which all elements are 1/48.
Conversely, if there is no $k$ such that $\Gamma^k(\tau)$ is complete, then, for any $k$, the vector $P^k u$ contains at least one zero, and thus, by definition, $\tau$ is not pattern-isotropic.

If $\tau$ is symmetric, then the situation is slightly more subtle: because direction is irrelevant in the definition of pattern-isotropy, we should now consider 24 pairs of possible transformations such that the transformations in each pair result in each other's reverse; the traversal is pattern-isotropic if and only if, in the limit, 1/24 of the transformations of the base pattern comes from each pair. Therefore we put a non-zero entry in $P_{ij}$ when $\beta \circ \pi_j = \pi_i$ or $\beta \circ \sigma \circ \pi_j = \pi_i$, where $\sigma$ is the symmetry transformation that maps $\tau$ to its own reverse. We now fill the matrix with multiples of 1/16 rather than multiples of 1/8. Otherwise, the proof goes through verbatim.
\end{proof}

In practice, Lemma~\ref{lem:isotropy} allows us to calculate efficiently whether a curve is pattern-isotropic: we simply calculate $\Gamma^k(\tau)$ for increasing values of $k$, until we find we have completed a cycle, that is, until we have found two values $a < b$ such that $\Gamma^a(\tau) = \Gamma^b(\tau)$. Since the number of different values which $\Gamma^k(\tau)$ can assume is finite, such a cycle must eventually be found---and in practice it is found fast. Then, we can decide whether $\tau$ is pattern-isotropic by checking if $\Gamma^a(\tau)$ contains all 48 symmetries of the unit cube.

\paragraph{Edge-isotropy}
We now discuss how to recognize edge-isotropy, which turns out to be very simple: the edge-isotropic curves are exactly the face-continuous curves. To prove this, we first need to adapt our notation. Recall that a symmetry of the unit cube is given by a signed permutation $\Pi$, which we write as a square-bracketed sequence of three numbers whose absolute values are a permutation of $\{1,2,3\}$. Let $\Gamma^k_*(\tau)$ be the transformations in $\Gamma^k(\tau)$ without the signs, so $\Gamma^k_*(\tau)$ is a subset of $U = \{[1,2,3],[1,3,2],[2,1,3],[2,3,1],[3,1,2],[3,2,1]\}$. We call the transformation $[1,2,3]$ the \emph{identity} transformation, the transformations $[2,3,1]$ and $[3,1,2]$ are \emph{shifts}, and the transformations $[1,3,2]$, $[2,1,3]$ and $[3,2,1]$ are \emph{swaps}. The following lemma is easy to verify by trying all combinations:

\begin{lemma}\label{lem:unsignedisotropy}
If $\Gamma^1_*(\tau)$ contains...
\begin{itemize}
\item[(i)] ...at least both shifts and no swaps, then $\Gamma^k_*(\tau) = \{[1,2,3], [2,3,1], [3,1,2]\}$ for all $k \geq 2$;
\item[(ii)] ...at least one swap and one shift, or at least two swaps and identity, then $\Gamma^k_*(\tau)$ is the complete set $U$ for all $k \geq 3$;
\item[(iii)] ...at least two swaps, no shifts and no identity, then $\Gamma^k_*(\tau) = \{[1,2,3], [2,3,1], [3,1,2]\}$ for all even $k \geq 2$, and $\Gamma^k_*(\tau) = \{[1,3,2], [2,1,3], [3,2,1]\}$ for all odd $k \geq 3$.
\end{itemize}
\end{lemma}

Lemma~\ref{lem:unsignedisotropy} is very helpful in analyzing vertex-gated curves:

\begin{lemma}\label{lem:unsignedisotropyvertexgated}
If $\tau$ is a vertex-gated three-dimensional Hilbert curve, then case (i), (ii) or (iii) of Lemma~\ref{lem:unsignedisotropy} applies.
\end{lemma}
\begin{proof}
We distinguish two cases: edge-crossing and facet-crossing curves.

Suppose $\tau$ is edge-crossing, and suppose the edge that connects the entrance and the exit gate is parallel to the third coordinate axis (the cases of the first and the second coordinate axis are similar). For each $i \in \{1,2,3\}$, there must be an octant in which the edge that connects the octant's entrance and exit gate is parallel to the $i$-th coordinate axis, otherwise the traversal cannot reach the opposite octant on the interior diagonal through the first octant. Therefore, $\Gamma^1_*(\tau)$ must contain at least one of $[2,3,1]$ and $[3,2,1]$, at least one of $[1,3,2]$ and $[3,1,2]$, and at least one of $[1,2,3]$ and $[2,1,3]$.

Now suppose $\tau$ is facet-crossing. We call a facet a \emph{$k$-facet} if it is orthogonal to the $k$-th coordinate axis. Suppose the facet that contains the entrance and the exit gate is a 3-facet (the cases of 1- and 2-facets are similar), the entrance gate is at $(-\frac12, -\frac12, -\frac12)$ and the exit gate is at $(\frac12, \frac12, -\frac12)$. We say an octant is \emph{$i$-low} if it is on the same $i$-facet of the unit cube as the entrance gate, and \emph{$i$-high} otherwise. For each $i \in \{1,2\}$, there must be an $i$-low octant in which the facet that contains the octant's entrance and exit gate is an $i$-facet, otherwise the $i$-th coordinates of the exit gates of the four $i$-low octants alternate between $0$ and $-\frac12$, ending with $-\frac12$, and the remaining $i$-high octants, among which the octant that contains the exit gate of $\tau$, cannot be reached. Therefore, $\Gamma^1_*(\tau)$ must contain at least one of $[2,3,1]$ and $[3,2,1]$ and at least one of $[1,3,2]$ and $[3,1,2]$.

Both for edge-crossing and for facet-crossing curves, it follows that case (i), (ii) or (iii) of Lemma~\ref{lem:unsignedisotropy} applies.
\end{proof}

\begin{theorem}\label{thm:FCimpliesEI}
All face-continuous three-dimensional Hilbert curves are edge-isotropic.
\end{theorem}
\begin{proof}
Let $\tau$ be a face-continuous three-dimensional Hilbert curve. By Theorem~\ref{thm:allfacecontinuous}, $\tau$ must be vertex-gated, or it is \curvename{Ca00.gs}. In the latter case $\Gamma^1_*(\tau) = \{[2,1,3], [2,3,1], [3,1,2], [3,2,1]\}$, as can be seen in Figure~\ref{fig:facetgatedcurve}, hence case (ii) of Lemma~\ref{lem:unsignedisotropy} applies. Otherwise $\tau$ is a vertex-gated curve, and case (i), (ii) or (iii) of Lemma~\ref{lem:unsignedisotropy} applies, by Lemma~\ref{lem:unsignedisotropyvertexgated}.

By the same analysis as in the proof of Lemma~\ref{lem:isotropy}, asymptotically, as $k$ goes to infinity, for each element $\gamma \in \Gamma^k_*(\tau)$ there is an equal number of $k$-th level subcubes $C$ that are traversed according to the image of $\tau$ under the transformation $\gamma$ (modulo reflections, reversals, translation and scaling)---note that this also applies in case (iii): one can simply do the analysis for odd and even $k$ separately. Thus, by the composition of the sets $\Gamma^k_*(\tau)$ as described by Lemma~\ref{lem:unsignedisotropy}, for each edge of the approximating curve of the base pattern, its images within the $k$-th level subcubes are equally distributed among parallels of the first, the second, and the third coordinate axis. Thus $\tau$ is edge-isotropic.
\end{proof}

Note that Theorem~\ref{thm:FCimpliesEI} implies that the edge-isotropic curves are exactly the face-continuous curves, since for non-face-continuous curves, the concept of edge-isotropy is not defined.

\paragraph{Standing curves}
In fact, Lemma~\ref{lem:unsignedisotropyvertexgated} has another interesting consequence. By definition, a curve $\tau$ is standing if $\Gamma^1_*(\tau)$ contains a single swap and/or identity, and nothing else. Thus, Lemma~\ref{lem:unsignedisotropyvertexgated} immediately implies:
\begin{corollary}\label{cor:standing}
No standing three-dimensional Hilbert curve is vertex-gated.
\end{corollary}
In fact, one can check the 17 edge- and facet-gated curves one by one and find that none of them are standing. Thus we find:
\begin{finding}
All standing three-dimensional Hilbert curves are vertex-edge-gated or vertex-facet-gated.
\end{finding}
Examples of standing curves of each type are shown in Figures \ref{fig:vertexedgegatedcurves}acd and~\ref{fig:boringcurves}a, respectively.

\paragraph{Coordinate-shifting curves}
In the proof of Lemma~\ref{lem:unsignedisotropyvertexgated} we derived necessary conditions on $\Gamma^1_*(\tau)$ for vertex-gated curves. In fact, these conditions are almost sufficient for the realization of a given gate sequence. More precisely, The gate sequences for not only vertex-gated curves, but also for vertex-facet-gated curves, fix one axis in each octant, namely the axis of the edge that connects the gates (if the curve is edge-crossing), the axis orthogonal to the facet that contains the gates (if the curve is facet-crossing), or the axis orthogonal to the facet that contains the facet gate (if the curve is vertex-facet-gated). Otherwise, as discussed in Sections \ref{sec:vertexgated} and~\ref{sec:vertexfacetgated}, for each octant, one is free to choose whether or not to reflect it in a diagonal plane that contains the gates---in other words, one can choose freely how to permute the two non-fixed axes. Thus, for each octant, one can choose between a swap and a non-swap (shift or identity). Choosing a shift or identity in each octant results in a coordinate-shifting curve and we obtain:
\begin{theorem}
Each gate sequence for a vertex-gated or vertex-facet-gated curve admits a coordinate-shifting curve.
\end{theorem}
On the other hand, one can examine the 17 edge- and facet-gated curves and find that none of them are coordinate-shifting; one can also verify by hand that it is not possible to assemble octants of vertex-edge-gated curves in such a way that they only differ by rotations of the coordinate axes (and reflections, traversals and translations). Thus we find:
\begin{finding}\label{thm:coordinaterotating}
A gate sequence for a three-dimensional Hilbert curve can be realized by a coordinate-shifting curve if and only if it is vertex-gated or vertex-facet-gated.
\end{finding}

\subsection{Locality-preserving properties}\label{sec:observationslocality}

\begin{table}
\caption{Worst-case locality metrics for selected curves, and analytical bounds for octant-by-octant cube-filling curves known from the literature. For the Butz curve (along with \curvename{Ca00.chI} and \curvename{Ca00.cTI}), bounds on \WLMan, \WLEuc\ and \WLMax\ were also calculated by Niedermeier et al.~\cite{niedermeier}. The numbers reported below improve on their bounds on \WLEuc\ and \WLMax.}\label{tab:metrics}
\begin{tabular}{|>{\ttfamily}l@{ }lc@{ \,}c@{ \,}c@{ \,}c@{ \,}c@{ \,}c|}
\hline
name & nickname & $\WLMan^{1/3}$ & $\WLEuc^{1/3}$ & $\WLMax^{1/3}$ & $\WS^{2/3}$ & $\WBV$ & $\WBS^{2/3}$ \\
\hline
\hline
Ca00.c4I & Butz                             & 4.62       & 2.97       & 2.89       & 1.48       & 3.11       & 3.14       \\
Ca00.cc.44.hh.db & Alfa                     & 4.64       & 2.84       & 2.32       & 1.48       & 3.11       & 2.69       \\
Ca00.c4Z & Harmonious                       & 4.63       & 3.07       & 3.04       & 1.48       & 3.50       & 3.46       \\
Ca00.cT4 & Sasburg                          & 4.58       & 3.00       & 2.66       & 1.45       & 3.50       & 3.08       \\
Ca00.cv.4h & Base camp                      & 5.27       & 3.21       & 3.04       & 1.62       & 5.25       & 3.68       \\
Ca00.gs & Beta                              & 4.48       & 2.65       & 2.41       & 1.48       & 3.14       & 2.54       \\
Cu00.cc.4d.4d.Z7 & Long-legs                & 5.36       & 3.28       & 3.04       &            & 5.67       & 4.03       \\
Ll36.cc.II.hT.33 & Rough-edge               & 6.73       & 4.29       & 3.04       & $\geq$2.09 &10.50       & 6.34       \\
Se33.c7T & Mosquito                         & 6.73       & 4.30       & 3.04       &            &10.50       & 6.34       \\
Si11.ct.P9 & Rollercoaster                  & 7.23       & 4.16       & 3.04       & 1.82       &14.00       & 5.81       \\
\hline
\multicolumn{2}{|l}{symmetric face-continuous curves} & 4.49--4.82 & 2.97--3.10 & 2.65--3.04 & 1.45--1.50 & 3.11--3.73 & 2.99--3.46 \\
\multicolumn{2}{|l}{vertex-edge-gated curves}& 4.69--7.23 & 2.95--4.16 & 2.55--3.04 & 1.53--2.07 & 4.31--14.11& 2.97--5.84 \\
\multicolumn{2}{|l}{all Hilbert curves}      & 4.48--7.23 & 2.65--4.30 & 2.32--3.04 & 1.45--$\geq$2.09&3.11--14.11&2.54--6.34 \\
\hline\hline
\multicolumn{8}{|c|}{known bounds from the literature}\\
\hline
\multicolumn{2}{|l}{not necessarily self-similar~\cite{gotsman}}            &            & $\leq$4.90 &            &            &            &            \\
\multicolumn{2}{|l}{not necessarily self-similar~\cite{niedermeier}}        & $\geq$3.49 & $\geq$2.23 &$\geq$2.02  &            &            &            \\
\multicolumn{2}{|l}{face-cont.\ order-preserving~\cite{chochia}}& $\geq$4.40 &&&&&\\
\multicolumn{2}{|l}{face-cont.\ order-preserving~\cite{chochia}}& $\leq$5.04 &&&&&\\
\hline
\end{tabular}
\end{table}

Table~\ref{tab:metrics} shows some results on metrics of locality-preserving properties as discussed in Section~\ref{sec:localityproperties}, computed with algorithms from Sasburg~\cite{sasburg} based on our previous work~\cite{boxquality}.
With our current implementation we cannot easily compute $\WS^{2/3}$ with reasonable precision for all curves; the value for the Rough-edge curve is a lower bound. The true value cannot be that much higher: 2.44 is a weak upper bound, as we show below (Theorem~\ref{thm:wsupperbound}). Other results are with an error margin of, theoretically, up to 2\%.

No curve is best on all metrics, but the two hyperorthogonal curves (Alfa and Beta) stand out.
Beta, the facet-gated hyperorthogonal curve, is the unique best curve with respect to bounding-box surface and $L_2$-dilation. On each metric, the curve is within 4\% from optimal.
Alfa, the vertex-gated hyperorthogonal curve is the unique best curve with respect to $L_\infty$-dilation, the unique second-best with respect to bounding-box surface, and optimal with respect to bounding-box volume. On each metric, the curve is within 7\% from optimal. The Alfa curve confirms a result from ``computer simulation'' reported by Gotsman and Lindenbaum~\cite{gotsman} that $\WLEuc^{1/3}$ is at most 2.84 for some three-dimensional Hilbert curve that was left unspecified.

The symmetric face-continuous curves are always within 37\% from optimal. Several examples are listed in the above table, including the Butz curve, which is optimal with respect to bounding-box volume, and the best metasymmetric curve (the Sasburg curve), which is optimal with respect to curve section surface.

Crazy curves, such as Rough-edge, Mosquito and Rollercoaster, can score up to 350\% worse than optimal on the bounding-box volume metric, but on the other metrics the differences between curves are less pronounced. Note that a curve with diagonal edges in the approximating curves is not automatically worse than a face-continuous curve. On the bounding-box surface, $L_\infty$-dilation, and $L_2$-dilation metrics, the best vertex-edge-gated curves actually score slightly better than the best symmetric face-continuous curves---but still always worse than the two hyperorthogonal curves.

Finally, here is the promised upper bound on $\WS^{2/3}$:

\begin{theorem}\label{thm:wsupperbound}
Any octant-wise traversal has $\WS^{2/3}$ at most $\frac23 \sqrt[3]{49}$.
\end{theorem}
\begin{proof}
For a given section $s$ of the traversal, let a \emph{maximal} $k$-level subcube be a $k$-level subcube $Q$ that is completely contained in $s$ while the $(k-1)$-level subcube that contains $Q$ is not completely contained in $s$. For any $k$, let $n_k$ be the number of maximal $k$-level subcubes. Note that the $k$-level subcubes counted by $n_k$ are distributed over at most two $(k-1)$-level cubes, otherwise they would have to include eight $k$-level subcubes of the same $(k-1)$-level cube and therefore they would not be maximal.

Now consider the octants of a $(k-1)$-level subcube $Q$. Together these octants have 36 facets: 24 \emph{exterior} facets on the boundary of $Q$, and 12 \emph{interior} facets inside $Q$; the latter are each shared by two octants. Suppose $s'$ is a curve section that consists of $m$ octants of $Q$. Now consider any of the twelve interior facets $f$ together with the two exterior facets that coincide with $f$ in a projection orthogonal to $f$. Observe that, from these three facets, at most two can lie on the outside of $s'$. Hence $s'$ has a surface area of at most 24 octant facets, with a maximum of $6m$. Therefore, the $n_k$ $k$-level subcubes, distributed over at most two $(k-1)$-level subcubes, contribute at most $6 n_k$ facets of area $1/4^k$ each, with a maximum of $48/4^k$, the surface area of 8 $k$-level subcubes.

Recall from Section~\ref{sec:localityproperties} that $\WS^{2/3}$ is the maximum of $\frac16\cdot\mathrm{surface}(C(a,b))/(b-a)^{2/3}$. An upper bound on $\WS^{2/3}$ for any octant-by-octant traversal is therefore the following:\[
\WS^{2/3} =
\max_{a \in [0,1)}\max_{b \in (a,1]}\frac{\mathrm{surface}(C(a,b))/6}{(b-a)^{2/3}} \leq
\max_{n_1,n_2,n_3,... \in \{0,...8\}} \frac{\sum_{k=1}^\infty n_k / 4^k}{\left(\sum_{k=1}^\infty n_k / 8^k\right)^{2/3}}.
\] This expression is maximized with $n_k = 8$ for all $k$, in which case it evaluates to $\frac23 \sqrt[3]{49}$, which is slightly less than 2.44.
\end{proof}

Note that the bound of Theorem~\ref{thm:wsupperbound} is not tight: the calculation does not account for the fact that if we have $n_k = 8$ for all $k$, then there must be $k$-level subcubes with facets that are contained in facets of $(k-1)$-level subcubes. Such facets do not contribute to the surface area, hence a traversal with $\WS^{2/3} = \frac23 \sqrt[3]{49}$ cannot actually be realized.

\section{Software}\label{sec:software}

From the author's website at http://spacefillingcurves.net/, one may download the C++ sources of a tool to search the curves. The purpose of this tool is to allow us to verify the contents of the present article, to reverse-engineer the curve that underlies any existing, poorly documented implementation of a three-dimensional Hilbert curve, and to facilitate the exploration of three-dimensional Hilbert curves that have given properties.

The search tool takes as input a set of conditions which a curve should fulfill. These conditions could take the form of a prefix of a curve name, the order in which the subcubes in a grid of $8^k$ cubes are visited (for any natural number $k$), a curve description in the numerical style of Section~\ref{sec:definitionbypermutations}, and/or a subset of the properties\footnote{face-continuous, hyperorthogonal, symmetric, metasymmetric, maximally facet-harmonious, fully interior-diagonal-harmonious, well-folded, vertex-gated etc., edge-crossing etc., centred, order-preserving, pattern-isotropic, coordinate-shifting, standing.} described in Section~\ref{sec:optionalproperties}, with the exception of edge-isotropy, which is equivalent to face-continuity by Theorem~\ref{thm:FCimpliesEI}. The tool then searches the space of 10,694,807 three-dimensional Hilbert curves for matching curves, reports how many there are, and, depending on the user's preferences, outputs details for one or all of these curves. Such details may include whatever could be given as input (the name of the curve, the order in which subcubes in a grid are visited etc.), as well a POV-Ray \cite{povray} file for an illustration of the curve in the style of, for example, Figure \ref{fig:needthedots}c, and additional information if known (nickname, references).

The base pattern names (Table~\ref{tab:octantorders}) and the names of metasymmetric and vertex-facet-gated curves are hardcoded. Otherwise the tool computes gate sequences and computes or verifies curve properties on the fly, using the same algorithms that were used to generate the tables in Appendix~\ref{apx:gatesequences}.


Conditions for the search can be given on the command line or in an input file which may contain specifications for many successive searches. Thus the search tool allows one to run, for example, the following rudimentary test. First run the tool with no conditions specified to output the names of all curves, in order. Then run the tool again and output numerical descriptions of all three-dimensional Hilbert curves. Next run the search tool again, using the previous output as input, and generate the traversal orders of the $8\times 8\times 8$ grid for each curve---to make the test more interesting, the tool has an option to produce a random transformation (rotation, reflection and/or reversal) of each traversal order. Finally, run the search tool a fourth time, using the previous output as input, and generate the names of the curves. Verify that the output of the fourth run is the same as from the first run, to confirm that all generated curves are unique and are identified correctly.

\section{How many Hilbert curves are there in three dimensions?}\label{sec:howmanyin3D}

We will now try to answer the title question of this article.

\paragraph{There is 1 three-dimensional Hilbert curve}
The curve \curvename{Ca00.c4Z} (Figure~\ref{fig:edgecrossingcurves}a) is clearly the curve that is the ``most Hilbert'' of all. The curve visits the points on each of five of the six faces of the cube in the order of a two-dimensional Hilbert curve. It is the only curve that does this (Finding~\ref{fnd:harmonious}) and also has the three essential properties of Hilbert curves. Moreover, just like Hilbert's two-dimensional curve, it is also face-continuous, vertex-gated, edge-crossing, symmetric, well-folded, and pattern-isotropic. Its construction generalizes to higher dimensions~\cite{extradimensional}, which gives us a unique Hilbert curve for any number of dimensions.

\paragraph{There is 1 three-dimensional Hilbert curve}
The curve \curvename{Ca00.cv.4h} (Figure~\ref{fig:vertexedgegatedcurves}a) is clearly the curve that is the ``most Hilbert'' of all. It is the only curve that has the three essential properties of Hilbert curves and is also, just like Hilbert's two-dimensional curve, well-folded, centred (the point in the centre of the curve lies in the centre of the cube) and standing, where, just as in Hilbert's two-dimensional curve, the permutations (ignoring the signs) that define the transformations in the octants are restricted to swapping the first and the last axis in (and only in) the first and the last octant\footnote{The uniqueness of the curve is easy to verify: by Theorem~\ref{thm:CimpliesVE}, all centred curves are vertex-edge-gated; one can now simply check all the vertex-edge-gated curves of the well-folded pattern \curvename{Ca00}.}.

\paragraph{There is 1 three-dimensional Hilbert curve}
The curve \curvename{La13.cv.II} (Figure~\ref{fig:vertexedgegatedcurves}e) is clearly the curve that is the ``most Hilbert'' of all. It is the only curve that has the three essential properties of Hilbert curves and can be identified unambiguously by its base pattern, La13, eliminating any room for confusion about exactly which curve is intended.

\paragraph{There are 2 three-dimensional Hilbert curves}
The curves \curvename{Ca00.cc.44.hh.db} (Figure~\ref{fig:edgecrossingcurves}c) and \curvename{Ca00.gs} (Figure~\ref{fig:facetgatedcurve}), described by Bos and Haverkort~\cite{hyperorthogonal}, are clearly the curves that are the ``most Hilbert'' of all. They are the only self-similar well-folded hyperorthogonal curves in three dimensions~\cite{hyperorthogonal}. Bos and Haverkort describe how to construct two such curves in any number of dimensions greater than two. They also show that, regardless of the number of dimensions, each section of such a curve has a bounding box of volume at most four times the volume of the curve section itself.
Just like Hilbert's two-dimensional curve, the curves from Bos and Haverkort are also face-continuous and pattern-isotropic.
The three-dimensional curves have, in some ways, better locality-preserving properties than any other curve that has the three essential properties of Hilbert curves (see Section~\ref{sec:observationslocality}).

\paragraph{There are 3 three-dimensional Hilbert curves}
The curves \curvename{Ca00.chI} (Figure~\ref{fig:edgecrossingcurves}e), \curvename{Ca00.cTI} and \curvename{Ca00.c4I} (Figure~\ref{fig:edgecrossingcurves}b), together singled out by Niedermeier et al.~\cite{niedermeier}, are clearly the curves that are the ``most Hilbert'' of all: they are the only curves that are, just like the two-dimensional Hilbert curve, face-continuous, vertex-gated, edge-crossing, symmetric, well-folded, and coordinate-shifting.
This makes the permutations easy to implement efficiently in software. \curvename{Ca00.c4I} is the three-dimensional curve of Butz's construction, which is well-defined for any number of dimensions~\cite{Butz}.

\paragraph{There are 24 three-dimensional Hilbert curves}
Let $[X]$ denote any character from the string~$X$. The 24 curves $\curvename{Ca00.cT}[\curvename{bC47}]$, $\curvename{Ca11.cT}[\curvename{IPJ3}]$, $\curvename{Cd00.cP}[\curvename{IPJ3}]$, $\curvename{Cd11.cP}[\curvename{bC47}]$, $\curvename{Se00.cT}[\curvename{bPJ7}]$, and\break $\curvename{Se66.cT}[\curvename{IC43}]$ (see Figures \ref{fig:edgecrossingcurves}dfh and \ref{fig:facetcrossingcurves}abcd for examples) are clearly the curves that are the ``most Hilbert'' of all. These 24 curves are the only metasymmetric curves: they are the only curves that have the same degree of symmetry as the two-dimensional Hilbert curve, being composed of two congruent halves, four congruent quarters, and eight congruent octants.

\addvspace\baselineskip
Note that there is no overlap between the answers given so far: these answers do, in fact, mention $1+1+1+2+3+24 = 32$ different curves.

\paragraph{There are 920 three-dimensional Hilbert curves}
The 920 order-preserving curves whose names start with \curvename{Ca00.cc} or \curvename{Si00.cc} (see Figures~\ref{fig:edgecrossingcurves}a--g for examples), described by Alber and Niedermeier~\cite{Alber}, are clearly the curves that are the ``most Hilbert'' of all: they are the only curves that are, just like the two-dimensional Hilbert curve, face-continuous and order-preserving. Moreover, just like Hilbert's two-dimensional curve, they are also vertex-gated and edge-crossing.\footnote{Alber and Niedermeier counted 1\,536 curves with these properties, since they counted some curves twice which we consider to be equivalent: they counted a forward and a reverse copy of each of the asymmetric curves whose names start with \curvename{Ca00.cc.hh.I3}, \curvename{Ca00.cc.TT.I3}, \curvename{Ca00.cc.44.I3} (120 curves each) and \curvename{Si00.cc.LT.I3} (256 curves). Note that versions (a) and (b) of generator $\mathrm{Hil}^3_1.\mathrm{B}$ in their work are congruent under rotation around a line through the midpoints of the lower front and the upper back edge, therefore both versions generate the same curves whose names start with \curvename{Si00.cc.LT.I3}.\label{fn:alber920}}

\paragraph{There are 157,865 three-dimensional Hilbert curves}
The 157,865 curves whose names start with \curvename{Ca00} (see Figures \ref{fig:edgecrossingcurves}a--f, \ref{fig:vertexedgegatedcurves}a and \ref{fig:facetgatedcurve} for examples) are clearly the curves that are the ``most Hilbert'' of all: they are the only curves that are, just like the two-dimensional Hilbert curve, well-folded.\footnote{30,736 curves with each of the gate sequences \curvename{Ca00.cc.hh}, \curvename{Ca00.cc.TT} and \curvename{Ca00.cc.44}; 65,536 curves with gate sequence \curvename{Ca00.cc.h4}; 56 with names starting with \curvename{Ca00.cr} and 64 with names starting with \curvename{Ca00.cv} (see Table~\ref{tab:vertexedgegated1} in Appendix~\ref{apx:gatesequences}), and 1 curve \curvename{Ca00.gs}.}

\paragraph{There are 10,694,807 three-dimensional Hilbert curves}
These are all the curves that have the three-dimensional equivalent of the properties that uniquely define the two-dimensional Hilbert curve: continuity, self-similarity, based on subdivision into $2^d$ subcubes. Clearly, each of these curves is at least as much a Hilbert curve as any other.

\paragraph{There are infinitely many three-dimensional Hilbert curves}
Recall from Section~\ref{sec:justification} that we could also have decided to select vertex-gatedness and face-continuity, instead of self-similarity, as defining properties of the two-dimensional Hilbert curve. As we saw in Section~\ref{sec:vertexgated}, in three dimensions, there are multiple self-similar, face-continuous, vertex-gated, octant-by-octant traversals. Thanks to the vertex gates, you can assemble these into endless non-self-similar combinations\footnote{This is not the only way to construct non-self-similar, vertex-gated, face-continuous, octant-by-octant space-filling curves. One can also construct such curves using a combination of base patterns from \curvename{Ca00}, \curvename{Cl00}, and \curvename{Si00}, even though \curvename{Cl00} does not support self-similar face-continuous curves by itself.}. Therefore, in three dimensions, there are infinitely many octant-by-octant, vertex-gated, face-continuous space-filling curves. Clearly, each of them could be called a Hilbert curve.

\section{How many Hilbert curves are there in four dimensions?}\label{sec:howmanyin4D}

\paragraph{There are 4 four-dimensional Hilbert curves}
I am aware of four generalizations of Hilbert curves that have been described in the literature for any number of dimensions: Butz's generalization~\cite{Butz} (see Table~\ref{tab:4d}), and three of our own: the self-similar, well-folded, hyperorthogonal curves (one vertex-gated, one facet-gated)~\cite{hyperorthogonal}, and the harmonious Hilbert curve~\cite{extradimensional} (see Table~\ref{tab:4d}). All of these curves are well-folded: they differ only in the transformations within the subcubes.

\begin{table}
\caption{Four examples of four-dimensional Hilbert curves.
The description of the Harmonious Hilbert curve and the squared Hilbert curve are translated from my original manuscript~\cite{extradimensional}, Sections 5.1 and~6.3, where the permutations were given by their inverse, and where our coordinate axes $1,...,d$ were numbered from $d-1$ down to $0$.
The description of A\&N, the example from Alber and Niedermeier, is translated from Figure 5 in the original source~\cite{Alber}---it may be interesting to compare the different notation systems.}\label{tab:4d}
\def\arraystretch{2.5}
\centering
\begin{tabular}{|@{ }l@{ }l@{ }|}
\hline
Butz &
\descr{%
\fwd[2\\3\\4\\1]\edge[1]\fwd[3\\4\\1\\2]\edge[2]\fwd[3\\4\\1\\2]\edge[\m1]\fwd[4\\\m1\\\m2\\3]\edge[3]\fwd[4\\\m1\\\m2\\3]\edge[1]\fwd[\m3\\4\\1\\\m2]\edge[\m2]\fwd[\m3\\4\\1\\\m2]\edge[\m1]\fwd[\m1\\2\\\m3\\4]%
\edge[4]
\rev[\m1\\2\\\m3\\\m4]\edge[1]\rev[\m3\\\m4\\1\\\m2]\edge[2]\rev[\m3\\\m4\\1\\\m2]\edge[\m1]\rev[\m4\\\m1\\\m2\\3]\edge[\m3]\rev[\m4\\\m1\\\m2\\3]\edge[1]\rev[3\\\m4\\1\\2]\edge[\m2]\rev[3\\\m4\\1\\2]\edge[\m1]\rev[2\\3\\\m4\\1]%
}
\\
Harmonious &
\descr{%
\fwd[4\\3\\2\\1]\edge[1]\fwd[1\\4\\3\\2]\edge[2]\fwd[4\\3\\1\\2]\edge[\m1]\fwd[\m2\\\m1\\4\\3]\edge[3]\fwd[4\\\m2\\\m1\\3]\edge[1]\fwd[\m3\\1\\4\\\m2]\edge[\m2]\fwd[4\\1\\\m3\\\m2]\edge[\m1]\fwd[\m3\\2\\\m1\\4]%
\edge[4]
\rev[\m3\\2\\\m1\\\m4]\edge[1]\rev[\m4\\1\\\m3\\\m2]\edge[2]\rev[\m3\\1\\\m4\\\m2]\edge[\m1]\rev[\m4\\\m2\\\m1\\3]\edge[\m3]\rev[\m2\\\m1\\\m4\\3]\edge[1]\rev[\m4\\3\\1\\2]\edge[\m2]\rev[1\\\m4\\3\\2]\edge[\m1]\rev[\m4\\3\\2\\1]%
}
\\
A\&N &
\descr{%
\fwd[4\\2\\3\\1]\edge[1]\fwd[1\\4\\3\\2]\edge[2]\fwd[1\\4\\3\\2]\edge[\m1]\fwd[\m1\\4\\\m2\\3]\edge[3]\fwd[\m1\\4\\\m2\\3]\edge[1]\fwd[1\\4\\\m3\\\m2]\edge[\m2]\fwd[1\\4\\\m3\\\m2]\edge[\m1]\fwd[\m3\\2\\\m1\\4]%
\edge[4]
\rev[\m3\\2\\\m1\\\m4]\edge[1]\rev[1\\\m4\\\m3\\\m2]\edge[2]\rev[1\\\m4\\\m3\\\m2]\edge[\m1]\rev[\m1\\\m4\\\m2\\3]\edge[\m3]\rev[\m1\\\m4\\\m2\\3]\edge[1]\rev[1\\\m4\\3\\2]\edge[\m2]\rev[1\\\m4\\3\\2]\edge[\m1]\rev[\m4\\2\\3\\1]%
}
\\
Squared H. &
\vrule width0pt depth3ex
\descr{%
\fwd[2\\1\\4\\3]\edge[3]\fwd[3\\4\\2\\1]\edge[1]\fwd[3\\4\\1\\2]\edge[\m3]\fwd[1\\\m2\\4\\\m3]\edge[2]\fwd[4\\3\\1\\2]\edge[\m1]\fwd[\m2\\\m1\\4\\3]\edge[3]\fwd[\m2\\\m1\\3\\4]\edge[1]\fwd[3\\\m4\\1\\\m2]%
\edge[4]
\rev[3\\4\\1\\\m2]\edge[\m1]\rev[\m2\\\m1\\3\\\m4]\edge[\m3]\rev[\m2\\\m1\\\m4\\3]\edge[1]\rev[\m4\\3\\1\\2]\edge[\m2]\rev[1\\\m2\\\m4\\\m3]\edge[3]\rev[3\\\m4\\1\\2]\edge[\m1]\rev[3\\\m4\\2\\1]\edge[\m3]\rev[2\\1\\\m4\\3]%
}
\\
\hline
\end{tabular}
\end{table}

\paragraph{There are 7 four-dimensional Hilbert curves}
In addition to the above, I have singled out, or have seen singled out, three more curves:
the curve \emph{H4cdNew} from my work on R-trees with Van Walderveen~\cite{jea},
an example from Alber and Niedermeier~\cite{Alber} (see Table~\ref{tab:4d}), and
a curve which we will call the \emph{squared Hilbert curve}.

The last curve results from doubling the number of coordinates of the points on the two-dimensional Hilbert curve $\lambda$, by lifting each of those coordinates to two dimensions using the same Hilbert-curve mapping~$\lambda$. We use subscripts $1$ and $2$ to identify the first and the second coordinate, respectively, of $\lambda(t)$, that is, $\lambda(t) = (\lambda_1(t), \lambda_2(t))$. The resulting four-dimensional curve is now described by the mapping $\tau: [0,1] \rightarrow [0,1]^4$ where:\[
\tau(t) = \Big(\lambda_1\big(\lambda_1(t)\big),\ \ \lambda_2\big(\lambda_1(t)\big),\ \ \lambda_1\big(\lambda_2(t)\big),\ \ \lambda_2\big(\lambda_2(t)\big)\Big).
\]
It can be shown that the curve $\tau$ is a (non-well-folded) four-dimensional Hilbert curve (\cite{extradimensional}, Section 6.3); a description is given in Table~\ref{tab:4d}. Moreover, due to the fact that the two-dimensional Hilbert curve $\lambda$ visits the points on the diagonals in order of ascending second coordinate, the curve $\tau$ visits points of the type $(x,y,x,y)$, that is, with the first two coordinates equal to the last two coordinates, in the same order in which the two-dimensional Hilbert curve $\lambda$ visits the corresponding points $(x,y)$.

\paragraph{There are incredibly many four-dimensional Hilbert curves}
The number of Hilbert curves in higher dimensions is subject to a combinatorial explosion that depends on the requirements one imposes on higher-dimensional Hilbert curves. To get just a glimpse of how bad this explosion can get, consider the example of vertex-gated, edge-crossing curves. When we fix the gate sequence, we have fixed, in each subcube, the axis of the edge that connects the gates. For the transformation that maps $\tau$ to the curve within the subcube we can choose freely from the six permutations of the remaining three axes, and we can still choose whether or not to reverse the curve. That gives us 12 choices for each of 16 subcubes, making $12^{16} = 184,884,258,895,036,416$ combinations \emph{per gate sequence} in total. Then we have not even considered the number of possible gate sequences yet.

\section{Evaluation and outlook}\label{sec:evaluation}

In this explorative work we discussed the question how many Hilbert curves exist in three dimensions. This question is ill-defined and the answers are debatable. No three-dimensional Hilbert curve is perfect: one can always find a combination of properties of the two-dimensional Hilbert curve that cannot be realized in three dimensions. For example, we found that in three dimensions, no octant-by-octant, self-similar space-filling curve exists whose endpoints are vertices of the unit cube and whose midpoint is the centre of the unit cube. Searching for a well-defined question and its answers unlocked a world of 10,694,807 three-dimensional space-filling curves: a large set of curves, most of which are probably ugly, but the set is small enough to search for elegant curves with interesting properties. In particular, I selected 24 curves that I found to be somehow interesting examples, and I sketched these in Figures \ref{fig:edgecrossingcurves}, \ref{fig:facetcrossingcurves}, \ref{fig:vertexedgegatedcurves}, \ref{fig:boringcurves} and~\ref{fig:facetgatedcurve}.

Surely there are more interesting curves. The (prototype) software tool may help readers in searching the realm of three-dimensional Hilbert curves. Furthermore, unanswered questions about locality-preserving properties abound. How do the curves differ with respect to metrics based on the average (rather than worst-case) distance between points along the curve as a function of their distance in $d$-dimensional space, or vice versa? How do the curves differ with respect to the average (rather than worst-case) measures of the boundary or the bounding boxes of curve sections? We may also consider expanding our territory by dropping the requirement of self-similarity (possibly trading it for vertex-gatedness and face-continuity), and attempt to describe and explore the infinite number of three-dimensional space-filling curves that will then come within scope.

However, the main contribution of this work may not be the exhaustive classification and the sometimes goal-oriented, sometimes curiosity-driven exploration of a particular class of three-dimensional space-filling curves. From this work, we may also get new ideas for different ways of constructing Hilbert-like space-filling curves in \emph{arbitrary} numbers of dimensions. In two dimensions there is nothing to choose, and as such, the two-dimensional Hilbert curve by itself does not show us that much about what we could try to achieve in higher dimensions. In three dimensions, we can see more. We discovered a number of interesting space-filling curves, some of which have properties that, in prior work, were established to be relevant to applications.

In particular, in the past years we have succeeded in generalizing the hyperorthogonal Hilbert curves to higher dimensions~\cite{hyperorthogonal}. It is there, in four or more dimensions, that these curves show their strengths as compared to the common generalization from Butz, achieving an exponential improvement on the worst-case bounding-box volume ratio metric. For the harmonious Hilbert curves, a generalization to higher dimensions has been identified as well\footnote{Van Walderveen was the first to find an algorithm to construct a compact description of such a curve for any number of dimensions. Later I found a simpler algorithm with a not-so-simple correctness proof~\cite{extradimensional}.}.
In fact, it is the potential applications of harmony properties of four- and six-dimensional curves~\cite{jea} that led us to studying them.
Generating all 10,694,807 possible three-dimensional Hilbert curves,
we discovered the three-dimensional harmonious Hilbert curve. This put us on the right track for discovering a family of unique Hilbert curves that have the harmony properties required by our application for any number of dimensions~\cite{extradimensional}.

The world of three-dimensional Hilbert curves, unlocked in this article, may contain more treasures that signpost the way to interesting, novel generalizations of Hilbert's curve into higher dimensions. For example, can the well-folded, centred, standing curve \curvename{Ca00.cv.4h}, or the Sasburg curve, \curvename{Ca00.cT4}, be generalized to higher dimensions in a useful way? How does the world of metasymmetric curves develop in higher dimensions? Can we narrow it down to a family of Hilbert curves, one for each number of dimensions, that are in some sense the most symmetric Hilbert curves of all? We have seen how in two dimensions, facet-gated curves are only possible by giving up self-similarity~\cite{betaomega}, while in three dimensions, one self-similar facet-gated curve exists, which is not symmetric. How do the possibilities for facet-gated curves develop in higher dimensions? Is there a symmetric, facet-gated Hilbert curve in four dimensions?

\subsection*{Acknowledgements}
Freek van Walderveen was the first to find an efficient algorithm that generates the transformations for each subcube of the $d$-dimensional harmonious Hilbert curve, for any $d$. All calculations of metrics of locality-preservation were made possible by Simon Sasburg, who developed the algorithm for the surface ratio metric, and who extended and improved our previous algorithms for the dilation and bounding-box metrics to be able to handle three- and higher-dimensional curves. The numerical notation system introduced in Section~\ref{sec:definitionbypermutations} is based on a notation technique from Arie Bos, adapted to suit the needs of the present article.

\bibliographystyle{abbrv}

\appendix

\clearpage
\section{Example curves}\label{apx:examples}

Table~\ref{tab:examplecurves} lists the example curves in this article with their properties and their descriptions in the style of Section~\ref{sec:definitionbypermutations}.

\begin{table}[H]
\begin{centering}
\caption{Example curves in this article}
\label{tab:examplecurves}
\def\arraystretch{1.6}
\begin{tabular}{|>{\ttfamily}lllll|}
\hline
\textrm{name}    & nickname       & description & prop. & figures \\
\hline\hline
Ca00.chI & & \descr{\fwd[2\\3\\1]\edge[1]\fwd[3\\1\\2]\edge[2]\fwd[1\\2\\3]\edge[\m1]\fwd[\m3\\\m1\\2]\edge[3]\rev[3\\\m1\\2]\edge[1]\rev[1\\2\\\m3]\edge[\m2]\rev[\m3\\1\\2]\edge[\m1]\rev[2\\\m3\\1]} & fpsv & \ref{fig:edgecrossingcurves}e \\
Ca00.cT4 & Sasburg & \descr{\fwd[3\\2\\1]\edge[1]\fwd[3\\1\\2]\edge[2]\fwd[2\\1\\3]\edge[\m1]\fwd[2\\\m3\\\m1]\edge[3]\rev[2\\3\\\m1]\edge[1]\rev[2\\1\\\m3]\edge[\m2]\rev[\m3\\1\\2]\edge[\m1]\rev[\m3\\2\\1]} & bfimpv & \ref{fig:edgecrossingcurves}d \\
Ca00.cT7 & Imposter & \descr{\fwd[3\\2\\1]\edge[1]\fwd[1\\3\\2]\edge[2]\fwd[1\\2\\3]\edge[\m1]\fwd[2\\\m3\\\m1]\edge[3]\rev[2\\3\\\m1]\edge[1]\rev[1\\2\\\m3]\edge[\m2]\rev[1\\\m3\\2]\edge[\m1]\rev[\m3\\2\\1]} & fimpv & \ref{fig:imposter-harmony},\ref{fig:edgecrossingcurves}f \\
Ca00.c4I & Butz & \descr{\fwd[2\\3\\1]\edge[1]\fwd[3\\1\\2]\edge[2]\fwd[3\\1\\2]\edge[\m1]\fwd[\m1\\\m2\\3]\edge[3]\rev[\m1\\\m2\\\m3]\edge[1]\rev[\m3\\1\\2]\edge[\m2]\rev[\m3\\1\\2]\edge[\m1]\rev[2\\\m3\\1]} & bfpsv & \ref{fig:fivecurves},\ref{fig:edgecrossingcurves}b\vrule width0pt depth3.1ex \\
Ca00.c4Z & Harmonious & \descr{\fwd[3\\2\\1]\edge[1]\fwd[1\\3\\2]\edge[2]\fwd[3\\1\\2]\edge[\m1]\fwd[\m2\\\m1\\3]\edge[3]\rev[\m2\\\m1\\\m3]\edge[1]\rev[\m3\\1\\2]\edge[\m2]\rev[1\\\m3\\2]\edge[\m1]\rev[\m3\\2\\1]} & fhipv & \ref{fig:allharmonious},\ref{fig:edgecrossingcurves}a \\
Ca00.cc.44.hh.db & Alfa & \descr{\fwd[3\\2\\1]\edge[1]\fwd[3\\1\\2]\edge[2]\rev[3\\1\\\m2]\edge[\m1]\fwd[\m2\\\m1\\3]\edge[3]\rev[\m2\\\m1\\\m3]\edge[1]\fwd[\m3\\1\\\m2]\edge[\m2]\rev[\m3\\1\\2]\edge[\m1]\rev[2\\\m3\\1]} & bfilov & \ref{fig:needthedots}c,\ref{fig:edgecrossingcurves}c \\
Ca00.cv.4h & Base camp & \descr{\fwd[3\\2\\1]\edge[1]\rev[1\\\m2\\\m3]\edge[2]\fwd[1\\2\\\m3]\edge[\m1]\rev[\m1\\2\\\m3]\edge[3]\fwd[\m1\\2\\3]\edge[1]\rev[1\\2\\3]\edge[\m2]\fwd[1\\\m2\\3]\edge[\m1]\rev[\m3\\\m2\\1]} & cu & \ref{fig:fivecurves},\ref{fig:vertexedgegatedcurves}a \\
Ca00.gs & Beta & \descr{\rev[\m3\\\m1\\\m2]\edge[1]\rev[\m3\\\m2\\\m1]\edge[2]\fwd[\m3\\2\\\m1]\edge[\m1]\rev[2\\\m3\\1]\edge[3]\fwd[2\\3\\1]\edge[1]\rev[3\\2\\\m1]\edge[\m2]\fwd[3\\\m2\\\m1]\edge[\m1]\fwd[\m2\\\m1\\3]} & bfilo & \ref{fig:facetgatedcurve}\vrule width0pt depth3.1ex \\
Ca11.cTJ & Wind-fold & \descr{\fwd[3\\2\\1]\edge[1]\fwd[3\\1\\2]\edge[\m1\\2]\fwd[2\\\m1\\3]\edge[1]\fwd[2\\\m3\\1]\edge[3]\rev[2\\3\\1]\edge[\m1]\rev[2\\\m1\\\m3]\edge[1\\\m2]\rev[\m3\\1\\2]\edge[\m1]\rev[\m3\\2\\1]} & impv & \ref{fig:edgecrossingcurves}h \\
Cd00.cPJ & & \descr{\fwd[2\\1\\3]\edge[1]\fwd[1\\\m3\\2]\edge[2]\fwd[1\\2\\3]\edge[\m1]\fwd[\m3\\\m1\\\m2]\edge[\m2\\3]\rev[3\\\m1\\2]\edge[1]\rev[1\\\m2\\\m3]\edge[2]\rev[1\\3\\\m2]\edge[\m1]\rev[\m2\\1\\\m3]} & dimpv & \ref{fig:facetcrossingcurves}d \\
Cd00.ct.4h & Indoor stroll & \descr{\fwd[2\\1\\3]\edge[1]\rev[1\\\m2\\\m3]\edge[2]\fwd[1\\2\\\m3]\edge[\m1]\rev[\m1\\2\\\m3]\edge[\m2\\3]\fwd[\m1\\\m2\\3]\edge[1]\rev[1\\\m2\\3]\edge[2]\fwd[1\\2\\3]\edge[\m1]\rev[\m2\\1\\3]} & cu & \ref{fig:vertexedgegatedcurves}c \\
Cd00.rv.3C & & \descr{\fwd[\m1\\2\\3]\edge[1]\fwd[\m2\\1\\\m3]\edge[2]\rev[3\\1\\\m2]\edge[\m1]\rev[\m1\\\m3\\2]\edge[\m2\\3]\fwd[\m1\\3\\\m2]\edge[1]\rev[1\\3\\\m2]\edge[2]\fwd[1\\3\\2]\edge[\m1]\fwd[\m3\\\m2\\\m1]} & i & \ref{fig:boringcurves}b\vrule width0pt depth3.1ex \\
Cd11.cP4 & & \descr{\fwd[2\\1\\3]\edge[1]\fwd[1\\\m3\\2]\edge[\m1\\2]\fwd[\m1\\2\\3]\edge[1]\fwd[\m3\\1\\\m2]\edge[\m2\\3]\rev[3\\1\\2]\edge[\m1]\rev[\m1\\\m2\\\m3]\edge[1\\2]\rev[1\\3\\\m2]\edge[\m1]\rev[\m2\\1\\\m3]} & dimpv & \ref{fig:facetcrossingcurves}c \\
Ce11.ct.P9 & Outdoor stroll & \descr{\fwd[1\\2\\3]\edge[1]\rev[\m2\\1\\\m3]\edge[\m1\\2]\fwd[2\\\m1\\\m3]\edge[1]\rev[2\\1\\\m3]\edge[\m1\\3]\fwd[2\\\m1\\3]\edge[1]\rev[2\\1\\3]\edge[\m1\\\m2]\fwd[\m2\\\m1\\3]\edge[1]\rev[\m1\\2\\3]} & cu & \ref{fig:vertexedgegatedcurves}d \\
Cl00.cf.ff.dd & & \descr{\fwd[2\\1\\3]\edge[1]\rev[\m2\\\m1\\\m3]\edge[2]\fwd[2\\\m1\\\m3]\edge[\m1]\rev[\m2\\1\\\m3]\edge[3]\fwd[1\\\m2\\3]\edge[\m2]\rev[\m1\\2\\3]\edge[1]\fwd[1\\2\\3]\edge[2]\rev[1\\\m2\\3]} & u & \ref{fig:fivecurves},\ref{fig:boringcurves}a \\
Cu00.cc.4d.4d.Z7 & Long-legs & \descr{\rev[\m1\\3\\\m2]\edge[1]\fwd[1\\\m2\\3]\edge[2]\rev[1\\2\\3]\edge[\m1]\fwd[\m1\\2\\3]\edge[\m2\\3]\fwd[1\\\m2\\3]\edge[2]\fwd[2\\\m1\\\m3]\edge[1]\rev[2\\1\\\m3]\edge[\m2]\fwd[\m2\\\m3\\1]} & iv & \ref{fig:fivecurves},\ref{fig:facetcrossingcurves}e\vrule width0pt depth3.1ex \\
La13.cv.II & Perfect fit & \descr{\fwd[2\\3\\1]\edge[1]\rev[1\\3\\\m2]\edge[2]\fwd[1\\2\\3]\edge[\m1\\3]\rev[\m1\\2\\\m3]\edge[1\\\m2]\fwd[\m2\\1\\\m3]\edge[2]\rev[3\\1\\\m2]\edge[\m1\\\m3]\fwd[\m3\\\m1\\\m2]\edge[\m2\\3]\rev[\m2\\1\\\m3]} & i & \ref{fig:vertexedgegatedcurves}e \\
Ll36.cc.II.hT.33 & Rough-edge & \descr{\fwd[2\\1\\3]\edge[2\\3]\fwd[1\\3\\2]\edge[1\\\m3]\rev[\m1\\\m3\\2]\edge[\m2]\fwd[\m2\\\m3\\\m1]\edge[\m1\\3]\rev[\m2\\3\\1]\edge[2\\\m3]\rev[\m2\\\m3\\\m1]\edge[1\\3]\fwd[\m2\\3\\1]\edge[\m2]\rev[2\\\m1\\\m3]} & v & \ref{fig:facetcrossingcurves}h \\
Ll36.cc.II.CC.J3 & Big cross & \descr{\fwd[2\\1\\3]\edge[2\\3]\rev[\m2\\3\\\m1]\edge[1\\\m3]\fwd[\m2\\\m3\\1]\edge[\m2]\rev[2\\\m3\\1]\edge[\m1\\3]\fwd[2\\3\\\m1]\edge[2\\\m3]\rev[\m2\\\m3\\\m1]\edge[1\\3]\fwd[\m2\\3\\1]\edge[\m2]\rev[2\\\m1\\\m3]} & iv & \ref{fig:fivecurves},\ref{fig:facetcrossingcurves}f \\
Se00.cT7 & Pirouette & \descr{\fwd[3\\2\\1]\edge[1]\fwd[\m3\\1\\2]\edge[2]\fwd[1\\2\\3]\edge[3]\fwd[\m1\\3\\2]\edge[\m1\\\m3]\rev[1\\\m3\\2]\edge[3]\rev[\m1\\2\\\m3]\edge[\m2]\rev[3\\\m1\\2]\edge[1]\rev[\m3\\2\\\m1]} & dimpv & \ref{fig:facetcrossingcurves}a\vrule width0pt depth3.1ex \\
Se33.c7T & Mosquito & \descr{\fwd[2\\3\\1]\edge[1\\2]\fwd[2\\1\\3]\edge[3]\fwd[1\\3\\\m2]\edge[\m2\\\m3]\fwd[\m1\\\m3\\\m2]\edge[\m1\\3]\rev[1\\3\\\m2]\edge[2\\\m3]\rev[\m1\\\m3\\\m2]\edge[3]\rev[2\\\m1\\\m3]\edge[1\\\m2]\rev[2\\\m3\\\m1]} & ipv & \ref{fig:facetcrossingcurves}g \\
Se66.cT3 & Helix & \descr{\fwd[\m3\\2\\\m1]\edge[\m1]\fwd[3\\\m1\\2]\edge[1\\2]\fwd[\m1\\2\\3]\edge[3]\fwd[1\\3\\2]\edge[\m1\\\m3]\rev[\m1\\\m3\\2]\edge[3]\rev[1\\2\\\m3]\edge[1\\\m2]\rev[\m3\\1\\2]\edge[\m1]\rev[3\\2\\1]} & dmpv & \ref{fig:facetcrossingcurves}b \\
Si00.cc.LT.I3.II & & \descr{\fwd[3\\1\\2]\edge[1]\fwd[1\\2\\3]\edge[2]\fwd[2\\3\\1]\edge[3]\fwd[2\\3\\1]\edge[\m2]\fwd[\m2\\\m3\\1]\edge[\m1]\fwd[3\\\m1\\\m2]\edge[2]\fwd[1\\2\\3]\edge[\m3]\fwd[\m2\\\m3\\1]} & dfpsv & \ref{fig:edgecrossingcurves}g \\
Si11.ct.P9 & Rollercoaster & \descr{\fwd[3\\1\\2]\edge[1]\rev[\m3\\\m2\\1]\edge[2\\3]\fwd[2\\1\\3]\edge[\m3]\rev[2\\1\\\m3]\edge[\m1\\\m2\\3]\fwd[\m2\\3\\\m1]\edge[1]\rev[\m2\\3\\1]\edge[\m1\\2]\fwd[\m1\\2\\3]\edge[\m3]\rev[\m1\\3\\\m2]} & cdi & \ref{fig:vertexedgegatedcurves}b\vrule width0pt depth2.3ex \\
\hline
\end{tabular}
\end{centering}
\footnotesize
\par\addvspace\baselineskip
Listed properties are:
b:~optimal worst-case curve section shapes as measured by \WS, \WBV and/or \WBS;\quad
c:~centred;\quad
d:~fully interior-diagonal-harmonious;\quad
f:~face-continuous;\quad
h:~maximally facet-harmonious;\quad
i:~pattern-isotropic;\quad
l:~optimal dilation $\mathrm{WL}_1$, $\mathrm{WL}_2$ and/or $\mathrm{WL}_\infty$;\quad
m:~metasymmetric;\quad
o:~hyperorthogonal;\quad
p:~order-preserving;\quad
s:~coordinate-shifting;\quad
u:~standing;\quad
v:~vertex-gated.\quad
Symmetric curves are recognized by a name of the form $Pcmm.\curvename{c}sq$.
Well-folded curves are recognized by a name starting with \curvename{Ca00}.
\end{table}

\clearpage
\section{No full harmony}\label{apx:nofullharmony}

\begin{theorem}
No three-dimensional Hilbert curve can harmonize with the two-dimensional Hilbert curve on every facet.
\end{theorem}
\begin{proof}
Let $A$ be the first octant in the traversal, and let $B$, $C$, and $D$ be the octants that share an octant facet with $A$, in the order in which they are visited. Let $A'$, $B'$, $C'$ and $D'$ be the octants that are opposite of $A$, $B$, $C$ and $D$, respectively, with respect to the centre of the cube. If we know that an octant $X$ is visited before an octant $Y$, we will write $X \prec Y$. So we have $A \prec B \prec C \prec D$.

To match the two-dimensional Hilbert curve on the unit cube facets adjacent to $A$, we must have $B \prec D' \prec C$; $B \prec C' \prec D$; and $C \prec B' \prec D$; we can summarize this by $A \prec B \prec D' \prec C \prec B' \prec D$ and $B \prec C' \prec D$. In particular, we have $D' \prec C \prec B'$, so, to get their common unit cube facet with octants $A'$, $B'$, $C$ and $D'$ correct, we need to visit $A'$ either (a) after $B'$ or (b) before $D'$. These two options are illustrated in Figure \ref{fig:no-fully-harmonious}a and~\ref{fig:no-fully-harmonious}b, respectively, and we will now discuss them in detail.

\begin{figure}
\centering
\includegraphics[width=\hsize,page=1]{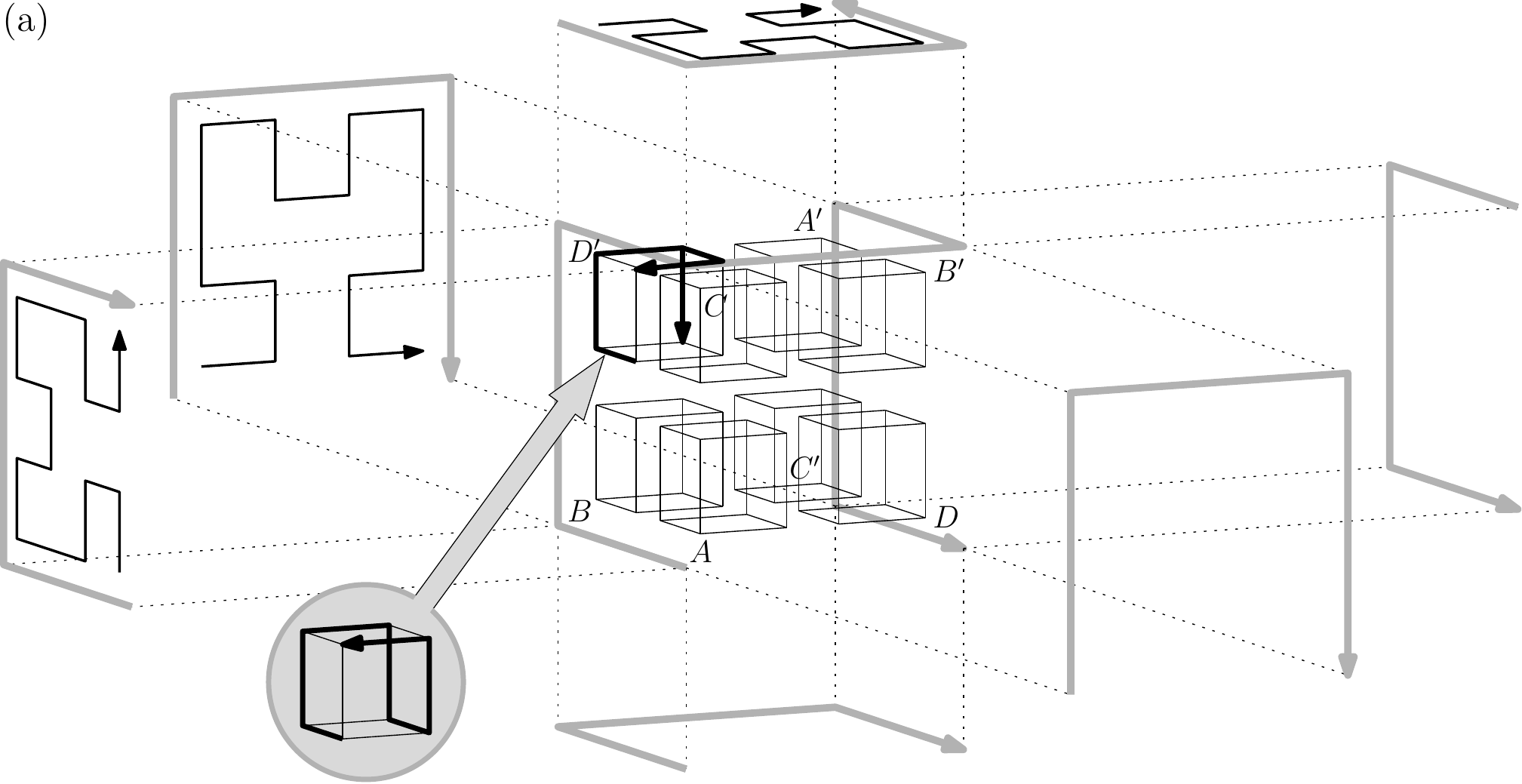}\\[3em]
\includegraphics[width=\hsize,page=2]{no-fully-harmonious.pdf}
\caption{(a) No three-dimensional Hilbert curve with the \curvename{Ca00}-pattern can harmonize with the two-dimensional Hilbert curve on all facets: in the third octant, one can get $A_2$ correct with a \curvename{Si00}-pattern but not with a \curvename{Ca00-pattern}. (b) No three-dimensional Hilbert curve with the \curvename{Si00}-pattern can harmonize with the two-dimensional Hilbert curve on all facets: to get $A_2$ correct in the second octant, one would have to place the pattern such that the exit gate does not connect to the third octant.}
\label{fig:no-fully-harmonious}
\end{figure}

Option (a): Since we now have $B \prec D' \prec C \prec B' \prec A'$, and $B \prec C'$, we can only get the common unit cube facet of $B$, $D'$, $A'$ and $C'$ correct if we visit these octants in that order. The partial order thus obtained can only be completed in one way: $A \prec B \prec D' \prec C \prec B' \prec A' \prec C' \prec D$. This traces out the \curvename{Ca00}-pattern, as illustrated in Figure~\ref{fig:no-fully-harmonious}a. Now the second-order approximating curves of the two-dimensional Hilbert curves on the unit cube facets bordering $D'$ induce a partial order on the suboctants of $D'$, as illustrated in the centre of the figure. The partial order can be completed to a full order, for example with an \curvename{Si00}-pattern, but not with a \curvename{Ca00}-pattern. Therefore a self-similar solution based on the \curvename{Ca00}-pattern is not possible.

Option (b): Since we now have $A' \prec D'$, $B \prec D'$ and $B \prec C'$, we can only get the common unit cube facet of $B$, $C'$, $A'$ and $D'$ correct if we visit these octants in that order. The partial order thus obtained can only be completed in one way: $A \prec B \prec C' \prec A' \prec D' \prec C \prec B' \prec D$. This traces out the \curvename{Si00}-pattern, as illustrated in Figure~\ref{fig:no-fully-harmonious}b. Now the second-order approximating curves of the two-dimensional Hilbert curves on the unit cube facets bordering $B$ induce a partial order on the suboctants of $B$, as illustrated in the centre of the figure. The partial order can be completed to a full order that follows an \curvename{Si00}-pattern, but only in such a way that the traversal of $B$ ends in a suboctant that is not adjacent to the next octant, $C'$. Therefore a continuous self-similar traversal based on the \curvename{Si00}-pattern is not possible.

So, in both cases we find that no traversal is possible that has the defining properties of a three-dimensional Hilbert curve and matches the two-dimensional Hilbert curve on every facet of the unit cube.
\end{proof}

\clearpage
\section{Equivalence relations between dilation, diameter ratio, and $L_i$-bounding ball ratio}
\label{apx:boundingball}

\begin{theorem}\label{thm:dilationequalsdiameter}
The $L_i$-dilation of any space-filling curve is equal to its $L_i$-diameter ratio, for any~$i$.
\end{theorem}
\begin{proof}
Consider a space-filling curve $\tau$ and two points $a,b \in [0,1]$ with $a \leq b$. Let $x,y$, with $a \leq x \leq y \leq b$, be a pair of points that determines the diameter of $C(a,b)$ and is closest along the curve, that is, with minimum $y-x$. Since $C(x,y) \subseteq C(a,b)$ and $\delta_i(\tau(x),\tau(y)) \geq \delta_i(\tau(a),\tau(b))$, reducing the curve section under consideration from $C(a,b)$ to $C(x,y)$ can only increase its dilation $\delta_i(\tau(a),\tau(b))^d / (b-a)$ and its diameter ratio $\mathrm{diam}_i(C(a,b))^d / (b-a)$; when $C(a,b)$ shrinks to $C(x,y)$ both will rise to the same value $\delta_i(\tau(x),\tau(y))^d / (y-x)$. Therefore, $\mathrm{WL}_i$ and $\mathrm{WD}_i$ are determined by the same point pairs and have the same value.
\end{proof}

Since the $L_\infty$-diameter of the minimum bounding $L_\infty$-ball of any set $S$ is equal to the $L_\infty$-diameter of $S$, we also have:
\begin{theorem}
The $L_\infty$-bounding ball ratio of any space-filling curve is equal to its $L_\infty$-diameter ratio.
\end{theorem}

I conjecture that the $L_2$-bounding ball ratio also equals the $L_2$-diameter ratio, but I can prove this only for two-dimensional space-filling curves:

\begin{theorem}
For two-dimensional space-filling curves, $\mathrm{WL}_2 = \mathrm{WD}_2 = \mathrm{WBB}_2$.
\end{theorem}
\begin{proof}
We will first prove that for the $L_2$-metric in two dimensions, the worst-case bounding ball ratio $\mathrm{WBB}_2$ is realized by a curve section $C(a,b)$ whose bounding ball is determined by exactly two points of $C(a,b)$. Suppose that, on the contrary, the worst bounding ball ratio is only realized by curve sections $C(a,b)$ whose bounding balls are determined by \emph{three} points $x,y,z \in [a,b]$. Let $\Delta$ be the triangle with vertices $\tau(x)$, $\tau(y)$ and $\tau(z)$, and let $\xi$, $\upsilon$ and $\zeta$ be the angles of $\Delta$ at the vertices $\tau(x)$, $\tau(y)$ and $\tau(z)$, respectively.
The smallest bounding ball is then the circumscribed circle of $\Delta$, which must be an acute triangle, and the circumscribed circle has diameter $\delta_2(\tau(x),\tau(y))/\sin(\zeta) = \delta_2(\tau(y),\tau(z))/\sin(\xi)$.
If $[a,b] \neq [x,z]$, then shrinking $C(a,b)$ to $C(x,z)$ would increase the bounding ball ratio, so we must have $a = x$ and $b = z$. Since the distance between the end points of a curve section is a lower bound on the bounding ball ratio, and, by assumption, no worst-case bounding ball is determined by only two points, we get:
$\delta_2^2(\tau(x),\tau(y))/\sin^2(\zeta)/(z-x) = \mathrm{diam}_2^2(\mathrm{bbaldiam}(C(x,z)))/(z-x) > \delta_2^2(\tau(x),\tau(y)) /(y-x)$, and thus,
$(y-x)/(z-x) > \sin^2(\zeta)$.
Analogously, we get $(z-y)/(z-x) > \sin^2(\xi)$, and thus,
$\sin^2(\zeta) + \sin^2(\xi) < ((z-y)+(y-x))/(z-x) = 1$.
However, since $\Delta$ is acute, we have $\pi/2 > \xi = \pi-\upsilon-\zeta > \pi/2 - \zeta$, and therefore
$\sin^2(\zeta) + \sin^2(\xi) > \sin^2(\zeta) + \sin^2(\pi/2-\zeta) = \sin^2(\zeta) + \cos^2(\zeta) = 1$: a contradiction. Therefore, there must be curve sections $C(a,b)$ that determine the worst-case bounding ball ratio and have a bounding ball determined by only two points, which must also determine the diameter of $C(a,b)$. Hence the worst-case bounding ball ratio $\mathrm{WBB}_2$ is a lower bound on the worst-case diameter ratio $\mathrm{WD}_2$, which, by Theorem~\ref{thm:dilationequalsdiameter}, equals the dilation $\mathrm{WL}_2$.

Since the diameter ratio of any curve section is also a lower bound on the bounding ball ratio, it follows that under the $L_2$-metric in two dimensions, the worst-case dilation, the worst-case diameter ratio and the worst-case bounding ball ratio are all equal.
\end{proof}

\clearpage
\section{Verifying the list of base patterns}\label{apx:verification}
To verify Table~\ref{tab:octantorders}, we will first analyse how many base patterns there could be. We ignore reversal for the moment, and get back to that later. Let a vector $(x_1,x_2,x_3)$ represent the octant that includes the unit cube vertex $(\frac12 x_1,\frac12 x_2,\frac12 x_3)$, assuming a unit cube of volume 1, centered at the origin. We say that a base pattern is in \emph{directed canonical form} if it starts with octant $(-1,-1,-1)$, and if octants $(-1,-1,+1)$, $(-1,+1,-1)$ and $(+1,-1,-1)$ appear in that order, possibly with other octants in between. We can reflect and rotate any given octant order into directed canonical form in two steps, as follows. First, if the first octant is $(x_1,x_2,x_3)$, then, for any $i$ such that $x_i = 1$, we reflect the pattern in the axis-parallel plane through the origin that is orthogonal to the $i$-th coordinate axis. In effect, we change the coordinates of any octant $(y_1,y_2,y_3)$ to $(-x_1y_1,-x_2y_2,-x_3y_3)$. As a result, the pattern now starts with octant $(-1,-1,-1)$. Second, we permute the coordinate axes as needed to ensure that octants $(-1,-1,+1)$, $(-1,+1,-1)$ and $(+1,-1,-1)$ appear in that order.

Now observe that there are 840 directed canonical octant orders, as each of them can be constructed by starting with the sequence $S_1,...,S_4 = (-1,-1,-1), (-1,-1,+1), (-1,+1,-1), (+1,-1,-1)$, and then inserting, one after another, the remaining four octants $S_5 = (-1,+1,+1), S_6 = (+1,-1,+1), S_7 = (+1,+1,-1)$ and $S_8 = (+1,+1,+1)$. For the $i$-th octant to be inserted ($i \geq 5$), there are $i-1$ positions to choose from (after each of the octants $S_1,...,S_{i-1}$). Thus the total number of canonical octant orders is $4\cdot 5\cdot 6\cdot 7 = 840$.

The base patterns in directed canonical form can be numbered consecutively from 0 to 839 as follows. Let $r_i$ be the position where $S_i$ is inserted in the aforementioned incremental construction of a given canonical order $\rho$; more precisely, let $r_i$ be the number of octants from $S_2,...,S_{i-1}$ that precede $S_i$ in $\rho$. Then we identify $\rho$ by the number $((r_5 \cdot 5 + r_6) \cdot 6 + r_7) \cdot 7 + r_8$.

To verify the correctness and completeness of the list of base patterns in Table~\ref{tab:octantorders}, we compare it to the simple numbering scheme presented above. The search tool described in Section~\ref{sec:software} has an option\footnote{Start it with: \texttt{hilbex patterns}} to iterate over all possible names of base patterns according to Table~\ref{tab:octantorders}, and to output, for each name, the numerical identifier of the directed canonical form of the pattern. More precisely, if the name indicates a symmetric order, the tool outputs the numerical identifier of the directed canonical form of that order. If the name indicates an asymmetric order, the tool outputs the numerical identifiers of the directed canonical forms of that order and its reverse. One may now verify that the tool outputs each of the numbers $\{0,...,839\}$ exactly once. Note that this is consistent with the counts of symmetric and asymmetric patterns in the table: one directed canonical form for each of 104 symmetric patterns, and two canonical forms for each of 368 asymmetric patterns, adds up to 840 directed canonical forms. This confirms that our naming scheme has a unique name for each possible base pattern, given that we consider patterns that only differ by reflection, rotation and/or reversal to be equivalent.

\clearpage
\section{Full list of gate sequences}\label{apx:gatesequences}
\label{apx:schemes}

Table~\ref{tab:vertexgated} lists all gate sequences for vertex-gated curves: each entry of the table consists of a prefix and one or more completions.

Tables \ref{tab:vertexedgegated1} and~\ref{tab:vertexedgegated2} list all gate sequences for vertex-edge-gated curves: each entry of the table consists of a prefix, one or more options for the symbol specifying the gates in the first half of the curve, and one or more options for the symbol specifying the gates in the second half of the curve. Each combination of a prefix, one symbol from the second column, and one symbol from the third column, constitutes a gate sequence name for a vertex-edge-gated curve.

Table \ref{tab:edgegated} lists the remaining gate sequences, that is, for vertex-facet-gated, edge-gated, and facet-gated curves.

Figure~\ref{fig:realizablebasepatterns} shows all base patterns that are realized by one or more Hilbert curves.

\begin{table}[h]
\centering
\caption{Gate sequences for vertex-gated curves. Symmetric sequences in boldface.}\label{tab:vertexgated}
\newcommand\symmetric[1]{\leavevmode\rlap{\raisebox{-0.09ex}{\rlap{#1}\kern0.05ex #1}}\rlap{#1}\kern0.05ex #1}
\begin{tabular}{|>{\footnotesize\ttfamily}l@{ }>{\footnotesize\ttfamily}l||>{\footnotesize\ttfamily}l@{ }>{\footnotesize\ttfamily}l||>{\footnotesize\ttfamily}l@{ }>{\footnotesize\ttfamily}l|}
\hline
\multicolumn{6}{|c|}{edge-crossing curves} \\
\hline
Ca00.cc. & \symmetric{hh} h4 \symmetric{TT} \symmetric{44}  & La11.cc. & \symmetric{hh} \symmetric{TT} & Se09.cc. & LL TT \\
Ca01.cc. & hh h4 4h 44                             & La77.cc. & \symmetric{hh} \symmetric{TT} & Si00.cc. & LT \\
Ca11.cc. & \symmetric{hh} h4 \symmetric{TT} \symmetric{44}  & Ll77.cc. & hT                      & Yi00.cc. & LT \\
Ck01.cc. & hT 4T                                   & Se06.cc. & hh                      & Yz00.cc. & LT hL \\
Ck11.cc. & hT T4                                   & Se08.cc. & hh                      & & \\
\hline\hline
\multicolumn{6}{|c|}{facet-crossing curves} \\
\hline
Cd00.cc. & \symmetric{II} Ih IP IX I4          & La36.cc. & hI hX                         & Se6b.cc. & Ph          \\
         & \symmetric{hh} hP hX h4 \symmetric{PP} & La37.cc. & LE                            & Sebb.cc. & \symmetric{bb} \\
         & PX P4 \symmetric{XX} X4 \symmetric{44} & La6c.cc. & Ih Xh                         & Si06.cc. & IC IE LP CP PC   \\
Cd02.cc. & I4 h4 P4 X4 44                   & La7c.cc. & EL EC                         &          & PE EP XC XE JP   \\
Cd11.cc. & \symmetric{PP}                      & La7d.cc. & EC                            & Si0b.cc. & Tb                              \\
Cd22.cc. & \symmetric{44}                      & Ll16.cc. & LL LJ EL EJ JL JJ             & Si16.cc. & IC IE XC XE                     \\
Ce00.cc. & \symmetric{bb} bd \symmetric{dd}       & Ll36.cc. & II IX bI bX                   & Si1b.cc. & Lh Jh                           \\
Ce01.cc. & IP hP PP XP 4P                   & Ll37.cc. & LE                            & Si36.cc. & 7T                              \\
Ce02.cc. & bd dd                            & Ll6d.cc. & IT XT                         & Si3b.cc. & Zh                              \\
Ce12.cc. & P4                               & Ll7c.cc. & PI Pb PT EL EC                & Yh06.cc. & TC CT                           \\
Ce22.cc. & \symmetric{dd}                      & Ll7d.cc. & EC                            & Yh26.cc. & JT                              \\
Cu00.cc. & Ib Id hb hd Pb                   & Ne00.cc. & \symmetric{bb} bd \symmetric{dd}    & Yi26.cc. & XC                              \\
         & Pd Xb Xd 4b 4d                   & Ne05.cc. & LL LC JL JC                   & Yi36.cc. & XC                              \\
Cu01.cc. & bP dP                            & Ne55.cc. & \symmetric{II} IT \symmetric{TT}    & Yk36.cc. & JC                              \\
Cu02.cc. & Id hd Pd Xd 4d                   & Se00.cc. & \symmetric{TT}                   & Yo36.cc. & JC                              \\
Cu20.cc. & 4b 4d                            & Se01.cc. & LL LJ CL CJ                   & Yz06.cc. & TC                              \\
Cu21.cc. & dP                               &          & EL EJ JL JJ                   & Yz60.cc. & TC                              \\
Cu22.cc. & 4d                               & Se03.cc. & LZ CZ EZ JZ                   & Yz62.cc. & TJ CX                           \\
La16.cc. & LL LJ EL EJ JL JJ                & Se33.cc. & \symmetric{77}                   & Yz63.cc. & CX                              \\
La17.cc. & IP PP XP                         & Se66.cc. & \symmetric{TT}                                                   && \\
\hline
\end{tabular}
\end{table}

\begin{table}
\centering
\caption{Vertex-edge-gated curves \curvename{C}, \curvename{L} or \curvename{N}}\label{tab:vertexedgegated1}
\noindent\begin{tabular}{|>{\footnotesize\ttfamily}l@{$[$}>{\footnotesize\ttfamily}l@{$][$}>{\footnotesize\ttfamily}l@{$]$ }|}
\hline
Ca00.cr.&I3&bCPJ47\\
Ca00.cr.&LhTX&EdZ9\\
Ca00.cr.&bCPJ47&I3\\
Ca00.cr.&EdZ9&LhTX\\
Ca00.cv.&I3&EdZ9\\
Ca00.cv.&LhTX&I3\\
Ca00.cv.&bCPJ47&LhTX\\
Ca00.cv.&EdZ9&bCPJ47\\
Ca01.cr.&bCPJ47&b\\
Ca01.cv.&bCPJ47&Lh\\
Ca01.cv.&EdZ9&I\\
Ca01.rc.&I3&ICPJ4\\
Ca01.rc.&LhTX&TX\\
Ca01.vc.&bCPJ47&Ed\\
Ca01.vc.&EdZ9&CPJ43\\
Ca02.cr.&I3&47\\
Ca02.cr.&LhTX&dZ\\
Ca02.cv.&I3&dZ\\
Ca02.cv.&EdZ9&47\\
Ca11.cr.&TXZ9&Lh\\
Ca11.cr.&CPJ43&b\\
Ca11.cr.&7&I\\
Ca11.cv.&ICPJ4&Lh\\
Ca11.cv.&LhEd&b\\
Ca11.cv.&TX&I\\
Ca12.cr.&Ed&dZ\\
Ca12.cv.&TX&47\\
Cd00.ct.&I3&EdZ9\\
Cd00.ct.&bCPJ47&LhTX\\
Cd00.cv.&I3&EdZ9\\
Cd00.cv.&bCPJ47&LhTX\\
Cd01.ct.&bCPJ47&LhEd\\
Cd01.cv.&bCPJ47&Lh\\
Cd01.tc.&LhTX&CPJ43\\
Cd01.vc.&LhTX&ICPJ4\\
Cd02.cv.&I3&dZ\\
Cd11.ct.&CPJ43&LhEd\\
Cd11.ct.&7&TX\\
Cd11.cv.&ICPJ4&Lh\\
Ce00.cr.&LhTX&EdZ9\\
Ce00.cr.&bCPJ47&I3\\
Ce00.ct.&bCPJ47&EdZ9\\
Ce00.ct.&EdZ9&I3\\
Ce01.cr.&I3&I\\
Ce01.cr.&bCPJ47&b\\
Ce01.cr.&EdZ9&Lh\\
Ce01.ct.&I3&Ed\\
Ce01.ct.&LhTX&CPJ43\\
Ce01.ct.&bCPJ47&TXZ9\\
Ce01.ct.&EdZ9&7\\
\hline
\end{tabular}
\noindent\begin{tabular}{|>{\footnotesize\ttfamily}l@{$[$}>{\footnotesize\ttfamily}l@{$][$}>{\footnotesize\ttfamily}l@{$]$ }|}
\hline
Ce01.rc.&I3&CPJ43\\
Ce01.rc.&LhTX&TXZ9\\
Ce01.rc.&bCPJ47&7\\
Ce01.rc.&EdZ9&Ed\\
Ce01.tc.&I3&TX\\
Ce01.tc.&LhTX&b\\
Ce01.tc.&bCPJ47&LhEd\\
Ce01.tc.&EdZ9&ICPJ4\\
Ce02.cr.&LhTX&dZ\\
Ce11.cr.&CPJ43&b\\
Ce11.ct.&ICPJ4&TXZ9\\
Ce11.ct.&TX&7\\
Ce12.cr.&Ed&dZ\\
Ce12.cr.&7&47\\
Ck00.cr.&bCPJ47&I3\\
Ck00.cv.&bCPJ47&EdZ9\\
Ck01.cr.&bCPJ47&7\\
Ck01.cr.&EdZ9&3\\
Ck01.cv.&bCPJ47&Z9\\
Ck01.rc.&I3&ICPJ4\\
Ck01.rc.&LhTX&b\\
Ck01.vc.&bCPJ47&7\\
Ck01.vc.&EdZ9&CPJ43\\
Ck11.cr.&ICPJ4&7\\
Ck11.cr.&LhEd&Z9\\
Ck11.cr.&TX&3\\
Ck11.cv.&TXZ9&7\\
Ck11.cv.&CPJ43&Z9\\
Ck11.cv.&7&3\\
Ck12.cr.&b&hT\\
Ck12.cv.&7&bC\\
Cu00.ct.&I3&I3\\
Cu00.ct.&bCPJ47&EdZ9\\
Cu00.cv.&I3&bCPJ47\\
Cu00.cv.&bCPJ47&EdZ9\\
Cu00.rc.&I3&bCPJ47\\
Cu00.rc.&bCPJ47&LhTX\\
Cu00.tc.&LhTX&EdZ9\\
Cu00.tc.&EdZ9&bCPJ47\\
Cu01.ct.&I3&7\\
Cu01.ct.&LhTX&CPJ43\\
Cu01.ct.&bCPJ47&TXZ9\\
Cu01.ct.&EdZ9&Ed\\
Cu01.cv.&I3&3\\
Cu01.cv.&bCPJ47&Z9\\
Cu01.cv.&EdZ9&7\\
Cu01.rc.&I3&CPJ43\\
Cu01.rc.&LhTX&TXZ9\\
Cu01.rc.&bCPJ47&Ed\\
Cu01.rc.&EdZ9&7\\
\hline
\end{tabular}
\noindent\begin{tabular}{|>{\footnotesize\ttfamily}l@{$[$}>{\footnotesize\ttfamily}l@{$][$}>{\footnotesize\ttfamily}l@{$]$ }|}
\hline
Cu01.tc.&I3&b\\
Cu01.tc.&LhTX&TX\\
Cu01.tc.&bCPJ47&LhEd\\
Cu01.tc.&EdZ9&ICPJ4\\
Cu02.cv.&I3&bC\\
Cu10.ct.&ICPJ4&EdZ9\\
Cu10.cv.&CPJ43&EdZ9\\
Cu10.rc.&b&bCPJ47\\
Cu10.tc.&TXZ9&bCPJ47\\
Cu11.ct.&ICPJ4&TXZ9\\
Cu11.ct.&TX&Ed\\
Cu11.cv.&CPJ43&Z9\\
Cu11.rc.&b&CPJ43\\
Cu11.tc.&TXZ9&ICPJ4\\
Cu11.tc.&7&b\\
Cu20.rc.&47&LhTX\\
Cu21.rc.&47&Ed\\
Cu21.rc.&dZ&7\\
La11.cr.&9&hTX\\
La11.cr.&3&bCJ\\
La11.cv.&I&9\\
La11.cv.&L&3\\
La13.cv.&I&I\\
La16.cr.&9&TX\\
La16.cr.&3&CJ\\
La16.ct.&I&9\\
La16.ct.&L&3\\
La17.cr.&9&9\\
La17.cr.&3&3\\
La17.ct.&I&TX\\
La17.ct.&L&CJ\\
La17.rc.&hTX&h\\
La17.rc.&bCJ&b\\
La17.tc.&hTX&7\\
La17.tc.&bCJ&Z\\
La1c.ct.&9&E\\
La1c.ct.&3&LhT\\
La1d.ct.&3&Lh\\
La37.tc.&X&7\\
La67.rc.&TX&h\\
La67.rc.&CJ&b\\
La67.vc.&h&7\\
La67.vc.&b&Z\\
La77.cr.&7&I\\
La77.cr.&Z&L\\
La7c.cv.&7&bCP\\
La7c.cv.&Z&I\\
La7d.cv.&7&CP\\
Ll16.ct.&I&9\\
Ll16.cv.&3&Z\\
\hline
\end{tabular}
\noindent\begin{tabular}{|>{\footnotesize\ttfamily}l@{$[$}>{\footnotesize\ttfamily}l@{$][$}>{\footnotesize\ttfamily}l@{$]$ }|}
\hline
Ll17.ct.&I&TX\\
Ll17.rc.&bCJ&b\\
Ll17.tc.&hTX&7\\
Ll1c.ct.&9&bCP\\
Ll1c.ct.&3&LhT\\
Ll1c.cv.&I&bCP\\
Ll1c.cv.&L&LhT\\
Ll1d.ct.&3&Lh\\
Ll1d.cv.&I&CP\\
Ll37.tc.&X&7\\
Ll67.rc.&CJ&b\\
Ll67.vc.&h&7\\
Ll77.cr.&b&3\\
Ll7c.cv.&h&dZ9\\
Ll7c.cv.&b&J47\\
Ll7d.cv.&b&J4\\
Na00.cr.&TX&Ed\\
Na00.cr.&Ed&TX\\
Na00.cv.&TX&I3\\
Na00.cv.&CPJ4&Lh\\
Na00.cv.&Ed&b7\\
Na01.cv.&TX&b\\
Na01.cv.&CPJ4&Lh\\
Na01.rc.&Ed&Ed\\
Na01.vc.&I3&TX\\
Na01.vc.&Z9&CPJ4\\
Na05.cv.&Ed&hT\\
Na11.cv.&TX&I\\
Na11.cv.&CPJ4&Lh\\
Na11.cv.&Ed&b\\
Na15.cv.&TX&hT\\
Na15.cv.&Ed&bC\\
Ne00.cr.&Ed&TX\\
Ne00.ct.&TX&b7\\
Ne00.ct.&CPJ4&Z9\\
Ne01.ct.&TX&3\\
Ne01.ct.&CPJ4&Z9\\
Ne01.ct.&Ed&7\\
Ne01.rc.&TX&TX\\
Ne01.rc.&Ed&Ed\\
Ne01.tc.&I3&TX\\
Ne01.tc.&b7&Ed\\
Ne01.tc.&Z9&CPJ4\\
Ne05.ct.&TX&dZ\\
Ne05.ct.&Ed&47\\
Ne11.ct.&CPJ4&Z9\\
Ne11.ct.&Ed&3\\
Ne15.ct.&TX&47\\
\multicolumn{3}{|c|}{}\\
\multicolumn{3}{|c|}{}\\
\hline
\end{tabular}
\end{table}

\begin{table}
\centering
\caption{Vertex-edge-gated curves with partition \curvename{S} or \curvename{Y}}\label{tab:vertexedgegated2}
\noindent\begin{tabular}{|>{\footnotesize\ttfamily}l@{$[$}>{\footnotesize\ttfamily}l@{$][$}>{\footnotesize\ttfamily}l@{$]$ }|}
\hline
Se00.cr.&ICPJ4&Ed7\\
Se00.cr.&LhEd7&IZ9\\
Se00.ct.&bTXZ9&Ed7\\
Se00.ct.&CPJ43&IZ9\\
Se01.cr.&ICPJ4&TX\\
Se01.cr.&bZ9&Ed\\
Se01.cr.&TX3&CPJ4\\
Se01.ct.&IEd&Lh7\\
Se01.ct.&Lh7&I\\
Se01.ct.&bTXZ9&3\\
Se01.ct.&CPJ43&bZ9\\
Se01.rc.&IZ9&TX\\
Se01.rc.&bCPJ4&Ed\\
Se01.rc.&Ed7&CPJ4\\
Se01.tc.&IZ9&CPJ4\\
Se01.tc.&LhTX3&TX\\
Se01.tc.&Ed7&Ed\\
Se02.cr.&LhEd7&X\\
Se03.cr.&bZ9&J\\
Se06.cr.&LhEd7&7\\
Se06.cv.&bTXZ9&7\\
Se06.cv.&CPJ43&Z9\\
Se06.rc.&Lh3&TX3\\
Se06.rc.&bTX&CPJ4\\
Se06.vc.&bTXZ9&CPJ4\\
Se06.vc.&CPJ43&IEd\\
Se07.cv.&IEd&7\\
Se07.cv.&bTXZ9&3\\
Se07.cv.&CPJ43&Z9\\
Se07.rc.&Lh3&Ed\\
Se07.rc.&bTX&CPJ4\\
Se07.rc.&CPJ47&TX\\
Se07.vc.&Lh7&Ed\\
Se07.vc.&bTXZ9&CPJ4\\
Se07.vc.&CPJ43&TX\\
Se0b.cv.&IEd&9\\
Se0b.cv.&bTXZ9&3\\
Se11.cr.&CPJ4&TX\\
Se11.ct.&CPJ4&bZ9\\
Se11.ct.&Ed&3\\
Se12.cr.&TX&X\\
Se16.cr.&TX&7\\
Se16.cr.&Ed&3\\
Se16.cv.&TX&3\\
Se16.cv.&CPJ4&Z9\\
Se16.cv.&Ed&7\\
Se16.rc.&TX&Z9\\
Se16.rc.&CPJ4&Ed7\\
\hline
\end{tabular}
\noindent\begin{tabular}{|>{\footnotesize\ttfamily}l@{$[$}>{\footnotesize\ttfamily}l@{$][$}>{\footnotesize\ttfamily}l@{$]$ }|}
\hline
Se16.rc.&Ed&CPJ4\\
Se16.vc.&IZ9&CPJ4\\
Se16.vc.&Lh3&bTX\\
Se16.vc.&b&Lh\\
Se16.vc.&7&IEd\\
Se17.cv.&CPJ4&Z9\\
Se17.cv.&Ed&3\\
Se17.rc.&Ed&CPJ4\\
Se17.vc.&IZ9&CPJ4\\
Se17.vc.&7&TX\\
Se1b.cv.&TX&9\\
Se26.rc.&E&TX3\\
Se26.vc.&I&IEd\\
Se27.rc.&E&Ed\\
Se27.vc.&I&TX\\
Se36.rc.&P&Z9\\
Se36.vc.&L&Lh\\
Se66.cr.&TX3&b\\
Se66.ct.&IEd&TX3\\
Se66.ct.&CPJ4&Z9\\
Se67.ct.&IEd&7\\
Se67.ct.&bTX&3\\
Se67.ct.&CPJ4&Z9\\
Se67.rc.&I&TX\\
Se67.rc.&b&Ed\\
Se67.tc.&TX3&TX\\
Se67.tc.&Ed7&Ed\\
Se67.tc.&Z9&CPJ4\\
Se6b.ct.&IEd&X\\
Se6b.ct.&bTX&J\\
Se77.ct.&TX&7\\
Se77.ct.&CPJ4&Z9\\
Se7b.ct.&Ed&J\\
Si00.cr.&ICPJ4&Ed7\\
Si00.cr.&LhEd7&IZ9\\
Si00.cr.&bZ9&bCPJ4\\
Si00.cr.&TX3&LhTX3\\
Si00.ct.&IEd&bCPJ4\\
Si00.ct.&Lh7&LhTX3\\
Si00.ct.&bTXZ9&Ed7\\
Si00.ct.&CPJ43&IZ9\\
Si01.cr.&ICPJ4&TX\\
Si01.cr.&bZ9&CPJ4\\
Si01.ct.&Lh7&Lh7\\
Si01.ct.&CPJ43&bZ9\\
Si01.rc.&LhTX3&Ed\\
Si01.rc.&Ed7&CPJ4\\
\multicolumn{3}{|c|}{}\\
\hline
\end{tabular}
\noindent\begin{tabular}{|>{\footnotesize\ttfamily}l@{$[$}>{\footnotesize\ttfamily}l@{$][$}>{\footnotesize\ttfamily}l@{$]$ }|}
\hline
Si01.tc.&IZ9&CPJ4\\
Si01.tc.&bCPJ4&TX\\
Si02.cr.&LhEd7&X\\
Si02.cr.&TX3&J\\
Si03.cr.&LhEd7&X\\
Si03.cr.&TX3&J\\
Si06.cr.&LhEd7&7\\
Si06.cr.&bZ9&3\\
Si06.cr.&TX3&Z9\\
Si06.cv.&Lh7&3\\
Si06.cv.&bTXZ9&7\\
Si06.cv.&CPJ43&Z9\\
Si06.tc.&IEdZ9&Lh\\
Si06.tc.&Lh3&CPJ4\\
Si06.tc.&bTX&IEd\\
Si06.tc.&CPJ47&bTX\\
Si06.vc.&ICPJ4&TX3\\
Si06.vc.&LhEd7&CPJ4\\
Si06.vc.&bZ9&Z9\\
Si06.vc.&TX3&Ed7\\
Si07.cv.&Lh7&7\\
Si07.cv.&CPJ43&Z9\\
Si07.tc.&Lh3&CPJ4\\
Si07.tc.&CPJ47&Ed\\
Si07.vc.&LhEd7&CPJ4\\
Si07.vc.&TX3&TX\\
Si0b.cv.&bTXZ9&3\\
Si11.cr.&CPJ4&TX\\
Si11.ct.&CPJ4&bZ9\\
Si12.cr.&Ed&J\\
Si13.cr.&TX&X\\
Si16.cr.&Ed&Z9\\
Si16.cv.&CPJ4&Z9\\
Si16.tc.&Lh7&CPJ4\\
Si16.tc.&bZ9&Lh\\
Si16.vc.&IZ9&Z9\\
Si16.vc.&Lh3&CPJ4\\
Si17.cv.&CPJ4&Z9\\
Si17.tc.&Lh7&CPJ4\\
Si17.vc.&Lh3&CPJ4\\
Si1b.cv.&Ed&3\\
Si26.vc.&9&Ed7\\
Si26.vc.&3&TX3\\
Si27.vc.&9&TX\\
Si36.vc.&9&Ed7\\
Si36.vc.&3&TX3\\
Si37.vc.&3&Ed\\
\multicolumn{3}{|c|}{}\\
\hline
\end{tabular}
\noindent\begin{tabular}{|>{\footnotesize\ttfamily}l@{$[$}>{\footnotesize\ttfamily}l@{$][$}>{\footnotesize\ttfamily}l@{$]$ }|}
\hline
Yh06.ct.&h7&Ed\\
Yh06.ct.&bZ&P4\\
Yh06.tc.&TX&L\\
Yh06.tc.&CJ&I\\
Yh06.vc.&TX&3\\
Yh06.vc.&CJ&9\\
Yh26.vc.&X3&9\\
Yh26.vc.&Z9&3\\
Yi00.cr.&bZ&L3\\
Yi00.ct.&h7&P4\\
Yi02.cr.&h7&Z9\\
Yi02.cr.&bZ&J4\\
Yi03.cr.&h7&Z9\\
Yi06.cr.&bZ&Z\\
Yi06.tc.&CJ&L\\
Yi06.vc.&TX&9\\
Yi26.vc.&J4&3\\
Yi26.vc.&Z9&9\\
Yi36.vc.&J4&3\\
Yk03.cr.&bZ&b\\
Yk03.cv.&h7&T\\
Yk36.vc.&d&3\\
Yo36.vc.&d&3\\
Yz00.cr.&bZ&I9\\
Yz00.cv.&h7&CJ\\
Yz00.rc.&L3&h7\\
Yz00.tc.&P4&bZ\\
Yz02.cr.&h7&Z9\\
Yz02.cr.&bZ&X3\\
Yz02.cv.&h7&X3\\
Yz02.cv.&bZ&J4\\
Yz03.cr.&h7&Z9\\
Yz03.cv.&bZ&J4\\
Yz06.cr.&bZ&7\\
Yz06.ct.&h7&Ed\\
Yz06.tc.&CJ&I\\
Yz06.vc.&TX&3\\
Yz20.rc.&J4&h7\\
Yz20.rc.&Z9&bZ\\
Yz26.vc.&J4&9\\
Yz26.vc.&Z9&3\\
Yz60.ct.&9&Ed\\
Yz60.cv.&L&P4\\
Yz60.rc.&Z&h7\\
Yz62.cv.&I&CP\\
Yz62.cv.&L&IE\\
Yz63.cv.&I&CP\\
\multicolumn{3}{|c|}{}\\
\hline
\end{tabular}
\end{table}

\begin{table}
\centering
\caption{Gate sequences for vertex-facet-gated, edge-gated, and facet-gated curves.}\label{tab:edgegated}
\begin{tabular}{|>{\footnotesize\ttfamily}l>{\footnotesize\ttfamily}l>{\footnotesize\ttfamily}l>{\footnotesize\ttfamily}l|}
\hline
\multicolumn{4}{|c|}{vertex-facet-gated gate sequences (256 curves per gate sequence)} \\
\hline
Cl00.cf.ee & Cl00.cf.ef & Cl00.cf.fe & Cl00.cf.ff \\
\hline
\hline
\multicolumn{4}{|c|}{edge-gated curves (1 curve per gate sequence)} \\
\hline
Cd00.rt.Ib & Cd00.rt.Eb & Cd00.rv.Xb & Cd00.rv.3b \\
Cd00.rt.IC & Cd00.rt.EC & Cd00.rv.XC & Cd00.rv.3C \\
Cd00.rt.IP & Cd00.rt.EP & Cd00.rv.XP & Cd00.rv.3P \\
Cd00.rt.IX & Cd00.rt.EX &            &            \\
Cd00.rt.I3 & Cd00.rt.E3 &            &            \\
\hline
\hline
\multicolumn{4}{|c|}{facet-gated curves (1 curve per gate sequence)} \\
\hline
Ca00.gs &&& \\
\hline
\end{tabular}
\end{table}

\begin{figure}[b]
\centering
\includegraphics[scale=0.9,page=1]{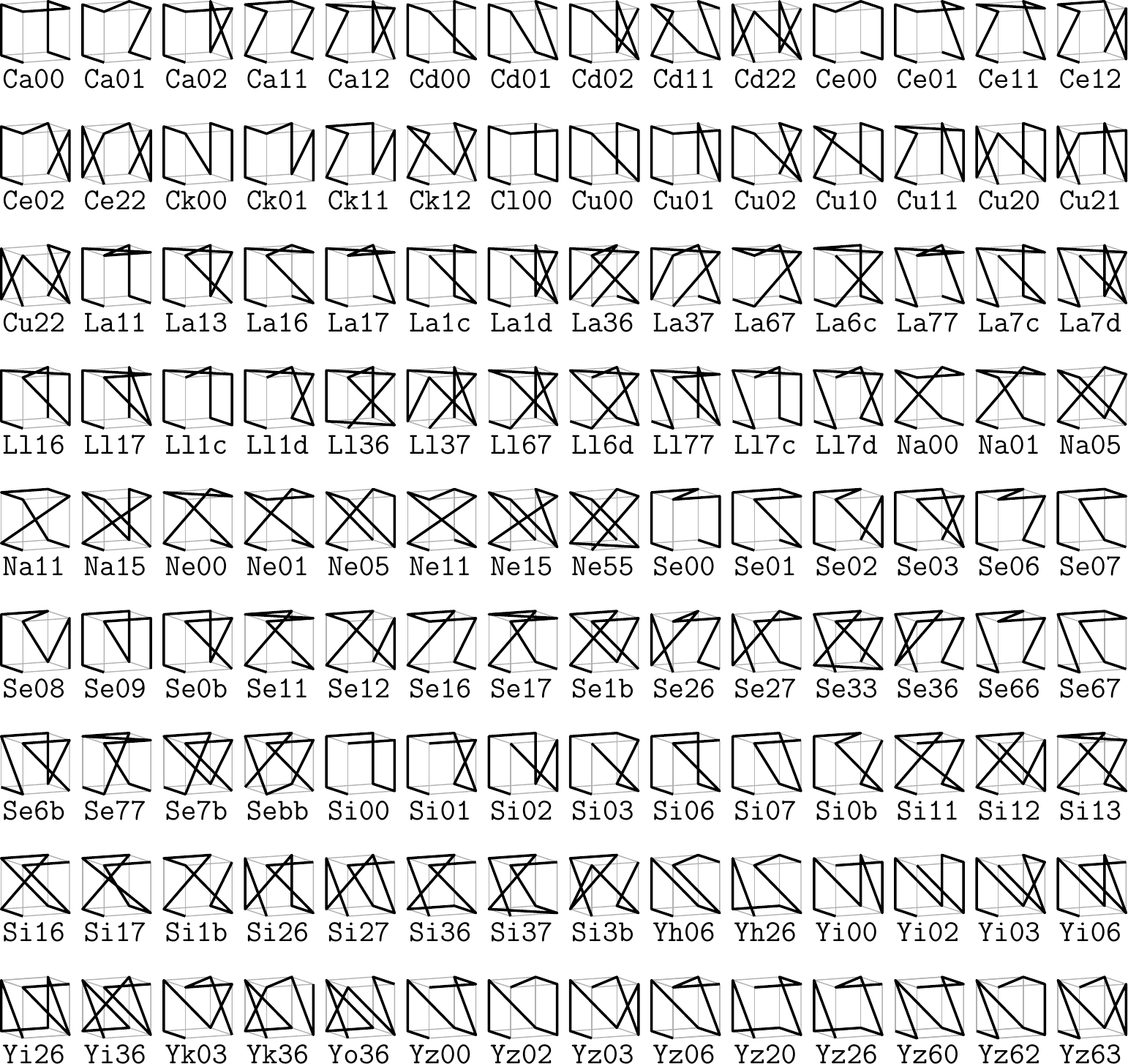}
\caption{There are 472 possible base patterns (see Table~\ref{tab:octantorders}), but only the 126 patterns shown above can be realized by Hilbert curves (see Tables~\ref{tab:vertexgated}--\ref{tab:edgegated}).}
\label{fig:realizablebasepatterns}
\end{figure}

\clearpage
\section{Analytical confirmation of certain observations}\label{apx:verifyobservations}

\begin{theorem}\label{thm:noX}
No three-dimensional Hilbert curve follows partition \curvename{X}.
\end{theorem}
\begin{proof}
With partition \curvename{X}, each pair of octants in the first half of the traversal shares an octant edge, but no octant facet. This immediately rules out curves with facet gates.

Edge-crossing vertex-gated curves are not possible, since neither of the vertices of the octant edge that is shared by the first and the second octant is on a common octant edge with the entrance gate (the outer vertex) of the first octant. In other words, the second octant is too far from the entrance gate to be reachable by an edge-crossing vertex-gated curve in the first octant. By an analogous argument, vertex-edge-gated curves are not possible, given Theorem~\ref{thm:vertex-edge-gates}.

Facet-crossing vertex-gated curves are not possible, since in such curves, the last octant is opposite to the first octant with respect to a facet diagonal, but with partition \curvename{X}, all such octants are in the first half of the traversal. By the same argument, edge-gated curves are not possible, given Theorem~\ref{thm:edge-edge-gates}, which says that the first and the last octant must be opposite of each other with respect to a facet diagonal.

Finally, cube-crossing vertex-gated curves are not possible by Theorem~\ref{thm:vertex-vertex-gates}.
\end{proof}

\begin{theorem}\label{thm:symmetric}
All symmetric curves are vertex-gated curves whose names start with \curvename{Ca}, \curvename{Cd}, \curvename{Ce}, \curvename{La}, \curvename{Ne} or \curvename{Se}.
\end{theorem}
\begin{proof}
A symmetric curve must be vertex-gated, edge-gated, or facet-gated. We know there is only one facet-gated curve (Theorem~\ref{thm:facet-facet-gates}), which is asymmetric, and the gates of edge-gated curves are positioned asymmetrically (Theorem~\ref{thm:edge-edge-gates}). This only leaves vertex-gated curves to consider.

A vertex-gated curve starts at a vertex of the unit cube whose coordinates sum up to $\frac12 \pmod 1$, and the coordinate sums of the entrance and exit gates of each octant differ by $\frac12$. Hence all even-indexed gates must be at an octant vertex whose coordinates sum up to $\frac12 \pmod 1$, that is, at a vertex or at a facet midpoint of the unit cube. This holds, in particular, for $\gate 4$, which must be at a facet midpoint, because at a unit cube vertex, it could not connect two octants. If a vertex-gated curve is symmetric, then the facet midpoint $\gate 4$ must be a fixed point of the symmetry.

Only the transformations \curvename{a}--\curvename{e} and \curvename{q}--\curvename{s} in Table~\ref{tab:transformations} have facet midpoints as fixed points, but transformations \curvename{q}--\curvename{s} are not symmetric (they are not equal to their own inverse). Curves with partition \curvename{X} do not exist (Theorem~\ref{thm:noX}). As we can verify with the help of Table~\ref{tab:transformations}, this leaves \curvename{Ca}, \curvename{Cd}, \curvename{Ce}, \curvename{La}, \curvename{Na}, \curvename{Ne} and \curvename{Se} as possible prefixes of names of symmetric curves. From these prefixes, \curvename{Na} can be ruled out because symmetric vertex-gated curves with transformation \curvename{a} must be edge-crossing, and vertex-gated, edge-crossing curves with partition \curvename{N} can easily be seen to be impossible by arguments similar to those in the proof of Theorem~\ref{thm:noX}.
\end{proof}

\begin{theorem}\label{thm:CimpliesVE}
All centred three-dimensional Hilbert curves are vertex-edge-gated.
\end{theorem}
\begin{proof}
A three-dimensional Hilbert curve is centred if and only if $\gate 4$, the exit gate of the fourth octant, is the centre of the unit cube. Since the centre of the unit cube is a vertex of each octant that touches it, a three-dimensional Hilbert curve can only be centred if it has at least one vertex gate, that is, the curve must be vertex-gated, vertex-edge-gated, or vertex-facet-gated.

As argued in the proof of Theorem~\ref{thm:symmetric}, if a curve is vertex-gated, the point $\gate 4$ halfway on the curve must be the centre of a facet of the unit cube and cannot be in the interior of the cube. In vertex-facet-gated curves, octants appear in pairs that share the octant facet that contains the facet gates, while the vertex gates are on the opposite sides of these octants, on the boundary of the unit cube. This leaves vertex-edge-gated curves as the only class of curves that may have $\gate 4$ in the centre of the unit cube.
\end{proof}

\bgroup\renewcommand{\thefinding}{\ref{fnd:metasymmetry}}
\begin{finding}
The only metasymmetric three-dimensional Hilbert curves are the curves
$\curvename{Ca00.cT}[\curvename{bC47}]$, $\curvename{Ca11.cT}[\curvename{IPJ3}]$, $\curvename{Cd00.cP}[\curvename{IPJ3}]$, $\curvename{Cd11.cP}[\curvename{bC47}]$, $\curvename{Se00.cT}[\curvename{bPJ7}]$, and $\curvename{Se66.cT}[\curvename{IC43}]$.
\end{finding}

\begin{howfound}
Unfortunately I cannot provide a non-tedious way to verify these findings at this point. Curves that are both symmetric and metasymmetric can be identified by checking the gate sequences for symmetric vertex-gated curves whose names start with \curvename{Ca}, \curvename{Cd}, \curvename{Ce}, \curvename{La}, \curvename{Ne} or \curvename{Se}, following Theorem~\ref{thm:symmetric}. As one may verify with the help of Table~\ref{tab:vertexgated} in Appendix~\ref{apx:gatesequences}, there are 28 symmetric gate sequences. In six of these, the gates are placed such that a metasymmetric curve results, provided the reflections in the third to eighth octant are chosen to agree with those in the first two octants. This results in four metasymmetric curves (corresponding to four options for  reflections in the first two octants) for each of the six gate sequences.

However, recall from Section~\ref{sec:optionalproperties} that the definition of metasymmetry does not require the curve to be symmetric: the curve may also be ``pseudo-symmetric'', that is, the correspondence between the first and the second half of the curve may have the form of a similarity transformation that does not equal its own inverse. The reader may verify that pseudo-symmetric curves that are not vertex-gated cannot exist, and that for vertex-gated pseudo-symmetric curves, much of the arguments in the proof of Theorem~\ref{thm:symmetric} still applies, to the effect that no pseudo-symmetric three-dimensional Hilbert curves exist that are not truly symmetric.
\end{howfound}
\egroup

\end{document}